\documentclass[11pt,letterpaper]{article}

\usepackage[T1]{fontenc}
\usepackage[utf8]{inputenc}
\usepackage[letterpaper,margin=1in]{geometry}
\IfFileExists{lmodern.sty}{\usepackage{lmodern}\usepackage{microtype}}{\usepackage[expansion=false]{microtype}}
\usepackage{amsmath,amssymb,amsthm,mathtools}
\usepackage{enumitem}
\usepackage{booktabs}
\usepackage{tabularx}
\usepackage{graphicx}
\usepackage{float}
\usepackage{caption}
\usepackage{xcolor}
\usepackage{hyperref}
\usepackage[nameinlink,capitalise,noabbrev]{cleveref}
\usepackage{indentfirst}

\hypersetup{
  colorlinks=true,
  linkcolor=blue,
  citecolor=blue,
  urlcolor=blue,
  pdftitle={Improved Depth-2 Linear Circuits for Disjointness via Quenched Lyapunov Exponents},
  pdfauthor={Lixi Ye}
}

\setlist{noitemsep,topsep=0.35em}
\newtheorem{theorem}{Theorem}[section]
\newtheorem{lemma}[theorem]{Lemma}
\newtheorem{proposition}[theorem]{Proposition}
\newtheorem{corollary}[theorem]{Corollary}
\theoremstyle{definition}
\newtheorem{definition}[theorem]{Definition}
\theoremstyle{plain}
\newtheorem*{remark}{Remark}

\newtheorem{innerlemma}{Lemma}
\newenvironment{manuallemma}[1]{\renewcommand\theinnerlemma{#1}\innerlemma}{\endinnerlemma}
\newtheorem{innerfact}{Fact}
\newenvironment{manualfact}[1]{\renewcommand\theinnerfact{#1}\innerfact}{\endinnerfact}
\newtheorem*{namedlemma}{Lemma}

\DeclareMathOperator{\nnz}{nnz}

\DeclareMathOperator{\clamp}{clamp}
\DeclareMathOperator{\freq}{freq}
\DeclareMathOperator{\FC}{FC}
\DeclareMathOperator{\conv}{conv}
\DeclareMathOperator{\Ber}{Ber}
\DeclareMathOperator{\Bin}{Bin}
\DeclareMathOperator{\wnnz}{wnnz}
\DeclareMathOperator{\sgn}{sgn}
\DeclareMathOperator{\Lip}{Lip}
\DeclareMathOperator{\Ch}{Ch}
\DeclareMathOperator{\Mat}{Mat}

\newcommand{\one}{\mathbf{1}}
\newcommand{\T}{\mathsf{T}}
\newcommand{\eps}{\varepsilon}

\newcommand{\paperfigure}[4]{%
\begin{figure}[H]
\centering
\IfFileExists{#1}{\includegraphics[width=0.95\textwidth,height=#3,keepaspectratio]{#1}}{\fbox{\parbox[c][#3][c]{0.88\textwidth}{\centering Image file \texttt{\detokenize{#1}} will be inserted here.}}}
\caption{#2}
\label{#4}
\end{figure}%
}

\title{Improved Depth-2 Linear Circuits for Disjointness\\ via Quenched Lyapunov Exponents}
\author{Lixi Ye\thanks{Institute for Interdisciplinary Information Sciences, Tsinghua University, Beijing, China.}\\
{\small\nolinkurl{ylx25@mails.tsinghua.edu.cn}}}
\date{\today}

\begin{document}

\renewcommand{\thefootnote}{\fnsymbol{footnote}}
\setcounter{footnote}{1}
\maketitle
\setcounter{footnote}{0}
\renewcommand{\thefootnote}{\arabic{footnote}}

\begin{abstract}
Let \(R_1 := \begin{pmatrix}1&1\\1&0\end{pmatrix}\), and write \(R_n := R_1^{\otimes n}\) for the \(N \times N\) disjointness matrix, where \(N = 2^n\). We construct new depth-2 linear circuits for \(R_n\): one of size \(O(N^{1.2449})\) and one of maximum input-output degree \(O(N^{0.3199})\), improving upon the results of Alman and Li \cite{AL25} (FOCS 2025), who achieved size \(O(N^{1.2495})\) and degree \(O(N^{1/3})\).

In this paper, we develop a number of mechanisms for designing depth-2 linear circuits that compute linear transforms represented by Kronecker powers of a general base matrix. For \(R_1\) in particular, we obtain our refined bounds by expressing a certain recursive circuit construction as a random matrix product and analyzing its quenched Lyapunov exponent, a quantity studied in the theory of random dynamical systems. Our new upper bounds imply faster deterministic algorithms for \#OV, more efficient online data structures for the Subset Query and Partial Match primitives, and smaller low-depth linear circuits for both disjointness and a class of related non-Kronecker-power matrices.
\end{abstract}

\newpage

\tableofcontents

\newpage

\section{Introduction}\label{sec:1}

The Kronecker product of two matrices \(A\in\Mat_{n_0\times m_0}\) and \(B\in\Mat_{n_1\times m_1}\), denoted by \(A\otimes B\), is defined by
\[
(A\otimes B)[(i_0,i_1),(j_0,j_1)]:=A[i_0,j_0]\cdot B[i_1,j_1].
\]
For \(k\geq1\), the \(k\)-th Kronecker power of a matrix \(A\) is \(A^{\otimes k}:=\underbrace{A\otimes\cdots\otimes A}_{k\text{ times}}\), with the convention \(A^{\otimes 0}:=[1]\).

We study \emph{linear circuits}: circuits in which each gate computes a fixed linear combination of its input gates, so that the circuit as a whole computes a linear transform. We focus on the following question:

\begin{itemize}
\item Fix a base matrix \(A\in\Mat_{\mathfrak{c}_0\times\mathfrak{c}_1}\). As \(k\to\infty\), how can one design depth-2 linear circuits that compute the linear transform represented by \(A^{\otimes k}\), with asymptotically small \emph{size} or \emph{maximum input-output degree}?
\end{itemize}

\subsection{Problem, Results, and Related Work}\label{sec:1.1}

\subsubsection{Depth-2 Linear Circuits}\label{sec:1.1.1}

The main object of study in this paper is the depth-2 linear circuit. A depth-2 linear circuit is a linear circuit with \(2+1=3\) layers: an \emph{input layer}, a \emph{middle layer}, and an \emph{output layer}. Each middle gate computes a fixed linear combination of the input gates, and each output gate computes a fixed linear combination of the middle gates.

If the three layers of a depth-2 linear circuit \(C\) have \(n\), \(r\), and \(m\) gates, respectively, then \(C\) can be written as an ordered pair of matrices \(C=(A,B)\), where \(A\in\Mat_{n\times r}\) and \(B\in\Mat_{r\times m}\) record the weights from the input gates to the middle gates and from the middle gates to the output gates, respectively.

Interpreting the input as a row vector, an input \(x\) is transformed as \(x\mapsto xA\mapsto xAB\). The input-output behavior of the circuit is therefore exactly the linear transformation represented by the matrix \(M=AB\), and we say that \(C\) \emph{computes} \(M\).

For a circuit \(C=(A,B)\), its \emph{size} \cite{JS13} is defined as
\[
\mathrm{S}(C):=\nnz(A)+\nnz(B),
\]
and its \emph{maximum input-output degree} \cite[Section 3.2]{AL25}, abbreviated as ``degree,'' is defined as
\[
\mathrm{D}(C):=\max\left\{\max_x\nnz(A[x,*]),\,\max_y\nnz(B[*,y])\right\}.
\]
Here \(\nnz(\cdot)\) denotes the number of nonzero entries of a matrix or vector; only edges with nonzero weight need to be present in the circuit. A detailed discussion of the circuit model appears in \cref{sec:A.1}.

\subsubsection{Asymptotic Size and Degree Constants}\label{sec:1.1.2}

Let \(\mathcal{R}\) be a commutative semiring containing \(0\) and \(1\), and let \(\Gamma\subseteq\mathcal{R}\) be a multiplicatively closed subset containing \(0\) and \(1\). Let \(\mathsf{Circ}_{\mathcal{R},\Gamma}^{(2)}\) denote the set of all depth-2 linear circuits whose weights belong to \(\Gamma\).

For a matrix \(M\in\mathcal{R}^{n\times m}\), we define its \emph{size cost} \(\mathrm{S}(M)\) and its \emph{degree cost} \(\mathrm{D}(M)\) as the minimum size and the minimum degree, respectively, over all circuits \(C\in\mathsf{Circ}_{\mathcal{R},\Gamma}^{(2)}\) computing \(M\). Formally,
\[
\mathrm{S}_{\mathcal{R},\Gamma}(M):=\min_{\substack{C=(A,B)\in\mathsf{Circ}_{\mathcal{R},\Gamma}^{(2)}\\ AB=M}}\mathrm{S}(C),
\qquad
\mathrm{D}_{\mathcal{R},\Gamma}(M):=\min_{\substack{C=(A,B)\in\mathsf{Circ}_{\mathcal{R},\Gamma}^{(2)}\\ AB=M}}\mathrm{D}(C).
\]
We further define the \emph{asymptotic size constant} and the \emph{asymptotic degree constant} \cite[Section 3.2]{AL25} of \(M\) as
\[
\sigma_{\mathcal{R},\Gamma}(M):=\lim_{n\to\infty}\mathrm{S}_{\mathcal{R},\Gamma}(M^{\otimes n})^{1/n},
\qquad
\delta_{\mathcal{R},\Gamma}(M):=\lim_{n\to\infty}\mathrm{D}_{\mathcal{R},\Gamma}(M^{\otimes n})^{1/n}.
\]

Throughout this paper, we deliberately work in the restrictive model \(\mathcal{R}=\mathbb{Z}\) and \(\Gamma=\{0,\pm1\}\). Since \(\mathcal{R}\) and \(\Gamma\) are fixed, we suppress them from the notation and abbreviate \(\mathrm{S}(M)\), \(\mathrm{D}(M)\), \(\sigma(M)\), and \(\delta(M)\) for \(\mathrm{S}_{\mathcal{R},\Gamma}(M)\), \(\mathrm{D}_{\mathcal{R},\Gamma}(M)\), \(\sigma_{\mathcal{R},\Gamma}(M)\), and \(\delta_{\mathcal{R},\Gamma}(M)\), respectively\footnote{Because the elements \(0\), \(1\), and \(-1\) can be identified in any field, we have \(\sigma_{\mathbb{F},\mathbb{F}}(M)\leq\sigma_{\mathbb{Z},\{0,\pm1\}}(M)\) and \(\delta_{\mathbb{F},\mathbb{F}}(M)\leq\delta_{\mathbb{Z},\{0,\pm1\}}(M)\) for every integer matrix \(M\). Thus our upper bounds for \(R_1\) and for \(P_1^{(c)}\) hold verbatim, with the same numerical constants, over an arbitrary field \(\mathbb{F}\).}.

In this convention, every integer matrix \(M\) admits at least one depth-2 linear circuit, and the limits defining \(\sigma(M)\) and \(\delta(M)\) exist because both circuit costs are submultiplicative under Kronecker products. Indeed, setting
\[
s_n:=\log_2\mathrm{S}(M^{\otimes n}),
\qquad
d_n:=\log_2\mathrm{D}(M^{\otimes n}),
\]
the sequences \(\{s_n\}_{n\ge1}\) and \(\{d_n\}_{n\ge1}\) are subadditive, so Fekete's lemma \cite{Fek23} applies to them\footnote{This phenomenon occurs widely in theoretical computer science and related fields: many important parameters are defined by limits of the form \(\lim_{n\to\infty}(a_n)^{1/n}\), where \(a_n\) is associated with the \(n\)-th power of an underlying object. Classic examples include the asymptotic tensor rank \(\widetilde{\mathrm{R}}(T):=\lim_{n\to\infty}\mathrm{R}(T^{\otimes n})^{1/n}\), which captures the matrix multiplication exponent through the identity \(\widetilde{\mathrm{R}}(\langle2,2,2\rangle)=2^\omega\) \cite{Str69, CW82}, and the Shannon capacity \(\Theta(G):=\lim_{n\to\infty}\alpha(G^{\boxtimes n})^{1/n}\) of a graph \(G\) \cite{Sha56}. Submultiplicativity or supermultiplicativity guarantees the existence of these limits. Determining their values, however, is notoriously difficult, and many such parameters are not even known to be computable.}. Further details appear in \cref{sec:A.1}.

\subsubsection{Main Theorem and Related Work}\label{sec:1.1.3}

Our main theorem is the following.

\begin{theorem}[see \cref{cor:sigma,cor:delta,thm:pm}]\label{thm:main}
For \(R_1=\begin{pmatrix}1&1\\1&0\end{pmatrix}\), we have
\[
\sigma(R_1)<2^{1.244844}<2.37,\qquad\delta(R_1)<2^{0.319866}<1.25.
\]
Moreover, for every integer \(c\ge2\), let
\[
P_1^{(c)}:=\begin{pmatrix}I_c&\one_c\end{pmatrix}=
\begin{pmatrix}
1&0&\cdots&0&1\\
0&1&\cdots&0&1\\
\vdots&\vdots&\ddots&\vdots&\vdots\\
0&0&\cdots&1&1
\end{pmatrix}
\]
denote the \(c\)-ary partial-match incidence matrix. Then
\[
\sigma(P_1^{(c)})\le c\cdot\psi(\log_2c)=c\left(1+\Theta\!\left(\frac{\log\log c}{\log c}\right)\right),
\]
where \(\psi(r)\) denotes the unique root \(x>1\) of the equation \(x=1+\frac{1}{x^r}\), and
\[
\delta(P_1^{(c)})\le c^{\frac{c-1}{(c-1)+c\log_2 c}}=2-\Theta\!\left(\frac{1}{\log c}\right).
\]
In particular, for \(P_1:=P_1^{(2)}\),
\[
\sigma(P_1)\le 1+\sqrt{5}\approx3.236,
\qquad
\delta(P_1)\le 2^{1/3}<1.26.
\]
\end{theorem}

There is a rich history of work on upper bounds for \(\sigma(R_1)\) and \(\delta(R_1)\). The starting point is Yates' algorithm \cite{Yat37}, which in our language gives \(\sigma(R_1)\leq\sqrt{6}\) and \(\delta(R_1)\leq\sqrt{2}\). Subsequently, for \(\sigma(R_1)\), Jukna and Sergeev \cite{JS13}, Alman, Guan, and Padaki \cite{AGP23}, Sergeev \cite{Ser22}, and Alman and Li \cite{AL25} gave successively better upper bounds. For \(\delta(R_1)\), Williams \cite{Wil24} gave an upper bound on the \emph{asymptotic equality rank constant} \(\theta(R_1)\)\footnote{The asymptotic equality rank constant \(\theta(M)\) is another asymptotic constant, defined as \(\lim_{n\to\infty}\mathrm{Q}(M^{\otimes n})^{1/n}\), where \(\mathrm{Q}(M)\) denotes the equality rank of \(M\), namely the minimum \(q\) admitting a decomposition \(M=\sum_{t=1}^qa_tE^{(t)}\) in which each \(E^{(t)}_{ij}=\one[f_t(i)=g_t(j)]\). It is easy to verify that \(\mathrm{D}(M)\leq\mathrm{Q}(M)\), and hence \(\delta(M)\leq\theta(M)\); see \cite[Propositions 3.1 and 3.2]{AL25}. The quantity \(\theta(R_1)\) is used in \cite{Wil24} not only to obtain algorithms for \#OV, but also to connect equality-rank lower bounds with exact-threshold circuit lower bounds. As pointed out in \cite{AL25}, however, if the goal is only to obtain algorithms for \#OV, then one may instead work with the smaller quantity \(\delta(R_1)\).}, and hence also on \(\delta(R_1)\), while Alman and Li \cite{AL25} gave a much better upper bound on \(\delta(R_1)\). \Cref{tab:sigma-history,tab:delta-history} summarize this history.

\begin{table}[H]
\centering
\small
\begin{tabular}{@{}lll@{}}
\toprule
Upper bound & Source/Formula & Reference \\
\midrule
\(2^{1.293}\) & \(\sqrt{6}\) & \cite{Yat37} \\
\(2^{1.272}\) & \(1+\sqrt{2}\) & \cite{JS13} \\
\(2^{1.258}\) & order-\(3\) Sergeev & \cite{AGP23} \\
\(2^{1.2504}\) & order-\(15\) Sergeev & \cite{Ser22} \\
\(2^{1.2495}\) & asymptotic spectrum & \cite{AL25} \\
\(\mathbf{2^{1.2449}}\) & Lyapunov exponent & \textbf{this paper} \\
\bottomrule
\end{tabular}
\caption{Historical progression of upper bounds on \(\sigma(R_1)\).}
\label{tab:sigma-history}
\end{table}

\begin{table}[H]
\centering
\small
\begin{tabular}{@{}lll@{}}
\toprule
Upper bound & Source/Formula & Reference \\
\midrule
\(2^{0.5}\) & \(\sqrt{2}\) & \cite{Yat37} \\
\(2^{0.4644}\) & \(5^{1/5}\) & \cite{Wil24} \\
\(2^{0.4151}\) & \(\frac{4}{3}\) & \cite{AL25}, without hole-fixing \\
\(2^{0.333\overline{3}}\) & \(2^{1/3}\) & \cite{AL25} \\
\(\mathbf{2^{0.3199}}\) & Lyapunov exponent + degree landscape & \textbf{this paper} \\
\bottomrule
\end{tabular}
\caption{Historical progression of upper bounds on \(\delta(R_1)\).}
\label{tab:delta-history}
\end{table}

Write \(R_k:=R_1^{\otimes k}\) for the \(k\)-th Kronecker power of \(R_1\), equivalently the disjointness matrix on \(k\)-bit vectors.

In the remainder of this section, we explain the intuition behind our paper: the limitations we observe in prior work, and the new ideas we introduce to address them. Since \cite{AL25} gives the best previous upper bounds on both size and degree, we compare our approach with that of \cite{AL25}.

A direct way to obtain a depth-2 linear circuit computing \(R_n\) is to handpick depth-2 linear circuits for constant Kronecker powers, which we call \emph{base circuits}, and simply take their Kronecker products.

(a) For optimizing size, starting in fact with \cite{JS13}, prior works observed that there is a construction better than directly taking Kronecker products of base circuits. In the first step, using a base circuit for \(R_d\), equivalently a rank-\(1\) decomposition \(R_d=\sum_i U_iV_i^\T\), one rewrites \(R_n\) as \(R_n=\sum_i U_iV_i^\T\otimes R_{n-d}\). If for every \(i\) we replace \(R_{n-d}\) by the same circuit, this amounts to directly taking Kronecker products of several circuits. However, for each \(i\), one can instead choose a different circuit for \(R_{n-d}\), based on the data of \(U_i\) and \(V_i\). This selection rule is called \emph{rebalancing}. The works \cite{JS13, AGP23, Ser22} each proposed heuristic selection rules, while the \emph{duality theorem} of \cite{AL25} completely solves the circuit-selection problem.

(b) For optimizing degree, when one takes Kronecker products of several base circuits, the degrees of most gates in the resulting product circuit concentrate, by Hoeffding's inequality, around the geometric mean of their degrees. However, an exponentially small fraction of gates may have degrees exceeding this geometric mean by an exponential factor. The work \cite{AL25} exploits the bit-permutation symmetry of the Kronecker-power matrices themselves to compensate for this loss, by removing the part lying far outside the concentration window; they call this procedure \emph{degree uniformization}.

The reason the present paper can improve on (a) is that prior works did not combine the advantages of these two routes. A circuit constructed via the selection rule in (a) may still contain, as in (b), many gates whose degrees are far above the median. If degree uniformization can be applied to such a circuit, it may noticeably improve both size and degree. However, analyzing the degrees of circuits constructed by rebalancing is much harder than for plain Kronecker-product circuits, and this degree analysis is one of the core technical contributions of our paper. Since the duality theorem of \cite{AL25} characterizes the infimum of the final size over circuits constructed ``only through rebalancing,'' whereas our paper essentially applies degree uniformization once more to the final constructed circuit, our construction has the potential to break the lower bound that the duality theorem of \cite{AL25} imposes on rebalancing constructions.

\subsection{Technical Overview}\label{sec:1.2}

\paragraph{The size and degree transfer theorems.}
We give a unified approach to the size and degree questions via a pair of \emph{size and degree transfer theorems}. Given a matrix \(M\in\mathbb{Z}^{\mathfrak{c}_0\times\mathfrak{c}_1}\) with no identically zero row or column, and a collection \(\mathcal{C}\) of depth-2 linear circuits, each computing some positive Kronecker power of \(M\), these theorems output upper bounds on \(\sigma(M)\) and \(\delta(M)\). Specifically, for a circuit \(C=(A,B)\) computing \(M^{\otimes k}\), let \(L_C(x):=\nnz(A[x,*])\) denote the degree of input gate \(x\), and let \(R_C(y):=\nnz(B[*,y])\) denote the degree of output gate \(y\). Writing \(\Delta_{[N]}\) for the set of probability distributions on \([N]:=\{0,\dots,N-1\}\), define
\[
\begin{aligned}
f_C(p)&:=\frac1k\,\mathbb{E}_{X\sim p^k}\bigl[\log_2 L_C(X)\bigr],&p\in\Delta_{[\mathfrak{c}_0]},\\[4pt]
g_C(q)&:=\frac1k\,\mathbb{E}_{Y\sim q^k}\bigl[\log_2 R_C(Y)\bigr],&q\in\Delta_{[\mathfrak{c}_1]}.
\end{aligned}
\]
For a pair of probability distributions \((p,q)\in\Delta_{[\mathfrak{c}_0]}\times\Delta_{[\mathfrak{c}_1]}\), we define the \emph{manifestation} of \(C\) at \((p,q)\) as \(P_C(p,q):=(f_C(p),g_C(q))\), viewed as a point in the first quadrant of the Cartesian plane. Fixing a pair \((p,q)\), let \(\conv\{P_C(p,q):C\in\mathcal{C}\}\) denote the convex hull of the points \(P_C(p,q)\) over \(C\in\mathcal{C}\). The size transfer theorem then states that
\[
\sigma(M)\le2^{\xi(\mathcal{C})},\qquad\xi(\mathcal{C}):=
\max_{(p,q)\in\Delta_{[\mathfrak{c}_0]}\times\Delta_{[\mathfrak{c}_1]}}\left\{
\min_{(a,b)\in\conv\{P_C(p,q):C\in\mathcal{C}\}}
\max\{a+\mathrm{H}_{\mathfrak{c}_0}(p),\,b+\mathrm{H}_{\mathfrak{c}_1}(q)\}\right\},
\]
where \(\mathrm{H}_{\mathfrak{c}_0}(p)\) and \(\mathrm{H}_{\mathfrak{c}_1}(q)\) denote the entropies of \(p\) and \(q\) in bits. The degree transfer theorem states that
\[
\delta(M)\le2^{\phi(\mathcal{C})},\qquad\phi(\mathcal{C}):=
\max_{(p,q)\in\Delta_{[\mathfrak{c}_0]}\times\Delta_{[\mathfrak{c}_1]}}\left\{
\min_{(a,b)\in\conv\{P_C(p,q):C\in\mathcal{C}\}}\max\{a,b\}\right\}.
\]
These two theorems are developed in detail in \cref{sec:2.3} (\cref{thm:size-transfer,thm:degree-transfer}).

\paragraph{Nine circuits for \(R_1\), and \(c+2\) circuits for \(P_1^{(c)}\).}
All of our size and degree results are obtained through the size and degree transfer theorems. Specifically, for \(R_1\), we construct nine circuits: \(C_0,C_1,C_2,C_3,C_4,C_{\mathrm{J}},C_{\mathrm{S}},C_{\mathrm{R}},C_{\mathrm{R}^\T}\). Plugging the singleton collection \(\{C_{\mathrm{J}}\}\) into the size transfer theorem gives \(\sigma(R_1)<2^{1.244844}\), and plugging the collection \(\mathcal{C}_8:=\{C_0,C_1,C_2,C_3,C_4,C_{\mathrm{S}},C_{\mathrm{R}},C_{\mathrm{R}^\T}\}\) into the degree transfer theorem gives \(\delta(R_1)<2^{0.319866}\). For \(P_1^{(c)}\), we construct \(c+2\) circuits \(C'_{\mathsf{pull}},C'_{\mathsf{push}},C'_{(0)},\ldots,C'_{(c-1)}\). Plugging \(\{C'_{\mathsf{pull}},C'_{\mathsf{push}}\}\) into the size transfer theorem gives \(\sigma(P_1^{(c)})\le c\cdot\psi(\log_2 c)\), and plugging \(\{C'_{\mathsf{pull}},C'_{(0)},\ldots,C'_{(c-1)}\}\) into the degree transfer theorem gives \(\delta(P_1^{(c)})\le c^{\frac{c-1}{(c-1)+c\log_2 c}}\). That these two collections yield our claimed upper bounds for \(P_1^{(c)}\) is proved in \cref{thm:pm}. By contrast, proving that the two collections above yield our claimed upper bounds for \(R_1\) is much more difficult. Among the circuits above, \(C_{\mathrm{J}},C_{\mathrm{S}},C_{\mathrm{R}},C_{\mathrm{R}^\T}\) have especially involved constructions: they compute \(R_{7\times10^{11}}\), while \(C_0,C_1,C_2\) compute \(R_1\), the circuits \(C_3,C_4\) compute \(R_2\), and \(C'_{\mathsf{pull}},C'_{\mathsf{push}},C'_{(0)},\ldots,C'_{(c-1)}\) all compute \(P_1^{(c)}\). For details, see \cref{sec:2.1,sec:2.2}.

\paragraph{Computation trees and scripted rebalancing.}
Recall the first step of rebalancing from \cref{sec:1.1.3}: if a base circuit gives a rank-\(1\) decomposition \(R_d=\sum_i U_iV_i^{\T}\), then \(R_n=\sum_i U_iV_i^{\T}\otimes R_{n-d}\), and the circuit chosen for the remaining factor \(R_{n-d}\) depends on the leading term \(U_iV_i^{\T}\). Iterating this rule expands the one-step identity into a rooted tree: at an internal node \(x\), one attaches a depth-2 linear circuit \(C(x)\) computing \(R_k\); each rank-\(1\) term of \(C(x)\) contributes one edge to a child; and the subtree below that child defines a circuit for the remaining Kronecker power. The circuit defined by the whole tree is therefore the sum, over all root-to-leaf paths, of the Kronecker products of the rank-\(1\) terms attached to the edges along that path. This is exactly the computation-tree formulation of \cite[Appendix B]{AL25}. The circuits \(C_{\mathrm{J}},C_{\mathrm{S}},C_{\mathrm{R}},C_{\mathrm{R}^{\T}}\) are all defined by such computation trees with \(10^{11}+1\) node layers: each internal node uses a circuit computing \(R_k\), uniformly with \(k=7\), and each root-to-leaf path has \(10^{11}\) edges, so the resulting root circuit computes \(R_{7\times10^{11}}\).

In \cite[Appendix B]{AL25}, the choice of \(C(x)\) depends only on the accumulated tilt between the left and right sides of the rank-\(1\) terms along the root-to-\(x\) path, and this tilt is real-valued. In this paper, we instead propose \emph{scripted rebalancing}, which replaces the real-valued tilt with a finite state: each node carries a state \(h\in[H]\), and the circuit attached to the node is \(I(h)\), a prescribed circuit computing \(R_k\). For a rank-\(1\) term \((U,V^{\T})\), the state offset is the nearest integer to the base-\(Z\) logarithm of the ratio of the \((1-\beta,\beta)\)-Bernoulli weight of the nonzero support of \(V^{\T}\) to the \((1-\alpha,\alpha)\)-Bernoulli weight of the nonzero support of \(U\). The child state is obtained by adding this offset to the current state and clamping the result to \([0,H-1]\). Thus a parameter set \((k,H,I,Z,\alpha,\beta)\) specifies a finite-state script for building a computation tree. We handpick four such scripts with \(k=7\), named Jessica, Sonetto, Regulus, and Regulus Transpose; starting each script from root state \(0\) and running it for \(10^{11}\) steps produces \(C_{\mathrm{J}},C_{\mathrm{S}},C_{\mathrm{R}},C_{\mathrm{R}^{\T}}\), respectively. For details, see \cref{sec:2.1,sec:2.2}.

\paragraph{Reducing mean log-degree to a Lyapunov exponent.}
It remains to prove that plugging \(\{C_\mathrm{J}\}\) and \(\mathcal{C}_8\) into the respective transfer theorems indeed yields the claimed upper bounds on \(\sigma(R_1)\) and \(\delta(R_1)\).

Fix \(p\in[0,1]\). To compute \(f_{C_{\mathrm{J}}}(\Ber(p))\), recall what this quantity means: choose an input-gate index of \(C_{\mathrm{J}}\) whose \(7\times10^{11}\) bits are i.i.d.\ \(\Ber(p)\), take the base-\(2\) logarithm of its degree, and take the expectation. Recall that \(C_{\mathrm{J}}\) is defined by a computation tree with \(10^{11}+1\) layers, each step consuming a block of \(7\) bits. We reveal the random input-gate index one block at a time. After revealing the first several blocks, we count the nodes in the corresponding layer that are \emph{compatible} with the revealed prefix: namely, those nodes for which the product of the rank-\(1\) terms along the root-to-node path has a left vector that is nonzero at the revealed prefix index. Since the subtree below any node is determined solely by the node's state, it suffices to maintain a vector indexed by \([H]\), recording, for each state, the number of compatible nodes in the current layer with that state.

The layer-to-layer update is therefore linear: the next vector is the current vector multiplied on the right by a matrix. Moreover, this matrix depends only on how many of the newly revealed \(7\) bits equal \(1\). Under the i.i.d.\ \(\Ber(p)\) input distribution, this number has distribution \(\Bin(7,p)\). Hence \(f_{C_{\mathrm{J}}}(\Ber(p))\) is obtained as follows: start from \([1,0,\dots,0]\); at each of the \(10^{11}\) steps, independently choose a matrix from \(\mathsf{A}_0,\dots,\mathsf{A}_7\) according to \(\Bin(7,p)\) and multiply the current vector on the right by that matrix; finally, take the expected base-\(2\) logarithm of the \(\ell_1\)-norm of the resulting vector and divide by \(7\times10^{11}\). The matrices \(\mathsf{A}_0,\dots,\mathsf{A}_7\) are determined by the parameter set alone and can be computed from it. For details, see \cref{sec:3.1}.

\paragraph{Toll-and-potential certificates and machine verification.}
Consider a collection of matrices \(\mathsf{A}_0,\dots,\mathsf{A}_7\) and the distribution \(\Bin(7,p)\) on them. For the corresponding random matrix product, choose a large number of steps, randomly select one matrix at each step to multiply the current vector, and take the logarithm of the resulting vector's \(\ell_1\)-norm divided by the number of steps. As the number of steps tends to infinity, this quantity converges almost surely to a fixed value: the \emph{quenched Lyapunov exponent}, a quantity studied in random dynamical systems, especially in the multiplicative ergodic theory of random matrix products.

Because \(10^{11}\) is so large, the numerical value of \(f_{C_{\mathrm{J}}}(\Ber(p))\) is extremely close to this exponent divided by \(7\). We show that any pair of arrays \((\mathfrak{a},\mathfrak{b})\in\mathbb{R}^{(k+1)\times(k+1)}\times\mathbb{R}^{(k+1)\times H}\) satisfying certain conditions determined by \(\mathsf{A}_0,\dots,\mathsf{A}_7\), called a \emph{valid certificate}, gives rise to a polynomial \(\widetilde\Phi_\mathfrak{a}(p)\) that upper-bounds the Lyapunov exponent simultaneously for every distribution \(\Bin(7,p)\). This certificate theorem is derived from the Lagrangian dual of the convex program underlying the upper bound of Gharavi and Anantharam \cite{GA05}.

Accordingly, for each circuit family \(\mathrm{X}\), we choose a lattice on \([0,1]\), obtain optimal certificates at the lattice points, and assemble a global upper bound \(\ell_{\mathrm{X}}:[0,1]\to\mathbb{R}_{\ge0}\). For each \(\mathrm{X}\), the inequality \(f_{C_{\mathrm{X}}(7m,h)}(\Ber(p))\le\ell_{\mathrm{X}}(p)\) holds for all \(m\geq10^{11}\), all \(h\in[H]\), and all \(p\in[0,1]\); moreover, \(\ell_\mathrm{X}\) is \(13\)-Lipschitz on \([0,1]\). These two facts allow us to write two computer programs whose acceptance is a sufficient (though not necessary) condition for the desired numerical bounds: \(\xi(\{C_\mathrm{J}\})<1.244844\) after plugging \(\{C_\mathrm{J}\}\) into the corresponding formula, and \(\phi(\mathcal{C}_8)<0.319866\) after plugging \(\mathcal{C}_8\) into the corresponding formula. The implementation is available at \url{https://osf.io/yz4fr/overview?view_only=2f77f3a2c69e44a2bb57738e01223149}. For details, see \cref{sec:3.2,sec:3.3}.

\subsection{Applications in Algorithm Design}\label{sec:1.3}

In this section, we use our upper bounds on \(\sigma(R_1)\), \(\delta(R_1)\), and \(\delta(P_1^{(c)})\) to derive algorithms and data structures for a variety of problems. We work in a word-RAM model\footnote{In our model, only additions, subtractions, and multiplications of \(w\)-bit integers, as well as memory accesses by \(w\)-bit addresses, are assumed to take unit time.} in which each word is a \(w\)-bit integer and \(w\geq 2\log_2(nd)\).

Our first application is a deterministic offline algorithm for \#OV.

\begin{theorem}[Deterministic \#OV from \(\delta(R_1)\)]\label{thm:ov}
For every \(\eps>0\), there is a deterministic offline algorithm that counts the orthogonal pairs between two sets \(U,V\subseteq\{0,1\}^d\), each of size at most \(n\), in \(O(n\cdot(\delta(R_1)+\eps)^d)\) time using only \(O(nd)\) space. In particular, by \cref{thm:main}, \#OV can be solved in \(O(n\cdot2^{0.319866d})=o(n\cdot2^{d/3.126})=o(n\cdot1.25^d)\) time and \(O(nd)\) space.
\end{theorem}

The proof of \cref{thm:ov} appears in \cref{sec:B.1}. This improves on previously known algorithms in the regime \(1.47\log_2n\lesssim d\lesssim3.12\log_2n\); see \cref{tab:ov-regimes}.

\begin{table}[H]
\centering
\small
\renewcommand{\arraystretch}{1.15}
\begin{tabularx}{\textwidth}{@{}>{\raggedright\arraybackslash}p{0.235\textwidth}>{\raggedright\arraybackslash}p{0.155\textwidth}cc>{\raggedright\arraybackslash}X@{}}
\toprule
& & \multicolumn{2}{c}{Up to \(n^{o(1)}\) factors} & \\
\cmidrule(lr){3-4}
Dimension regime & Descriptive label & Time & Space & Main method \\
\midrule
\(d=c\log_2n\) with \(c\lesssim1.47\) & Logarithmic-dimensional I & \(O(n+2^d)\) & \(O(n+2^d)\) & Yates' algorithm \\
\addlinespace
\(d=c\log_2n\) with \(1.47\lesssim c\lesssim3.12\) & Logarithmic-dimensional II & \(O(n\cdot2^{d/3.126})\) & \(O(n)\) & depth-2 linear circuits (this paper) \\
\addlinespace
\(d=c\log_2n\) with \(c\gtrsim3.12\) & Logarithmic-dimensional III & \(n^{2-1/O(\log c)}\) & \(n^{2-1/O(\log c)}\) & polynomial method + rectangular matrix multiplication \cite{AWY14, CW20} \\
\addlinespace
\(d=\omega(\log n)\cap n^{o(1)}\) & OVC barrier regime & \(O(n^2)\) & \(O(n)\) & brute-force pair enumeration \\
\bottomrule
\end{tabularx}
\caption{Complexity landscape for \#OV, using algorithms known to the authors as of June 2026.}
\label{tab:ov-regimes}
\end{table}

The reduction from \#OV to \(\delta(R_1)\) without the linear-space requirement is given in \cite[Theorem 7.2]{AL25}, and \cite{Wil24} gives a linear-space reduction from \#OV to the asymptotic equality rank constant \(\theta(R_1)\). Because our reduction targets degree rather than equality rank, ensuring linear space requires substantially more technical work.

As a second application, we give deterministic online data structures for the Subset Query and Partial Match primitives introduced by Charikar, Indyk, and Panigrahy \cite{CIP02}, as well as for a generalized \(c\)-ary Partial Match primitive.

In the online \textbf{Subset Query} problem, a database stores binary strings in \(\{0,1\}^d\) and supports a stream of insertion and query operations. An insertion adds a string to the database. A query specifies a pattern in \(\{0,*\}^d\) and asks for the number of currently stored strings matching the pattern; here a \(0\) in the query matches only a \(0\) in the database string, whereas a \(*\) matches any single symbol.

In the online \textbf{\(c\)-ary Partial Match} problem, a database stores strings in \(\{0,1,\ldots,c-1\}^d\) and supports a stream of insertion and query operations. The setup is the same as for Subset Query, except that a query specifies a pattern in \(\{0,1,\ldots,c-1,*\}^d\). The binary case \(c=2\) is the Partial Match problem defined in \cite{CIP02}.

Indeed, if we label the \(c\) rows of the \(c\times(c+1)\) matrix \(P_1^{(c)}=\begin{pmatrix} I_c&\one_c\end{pmatrix}\) by the database symbols \(0,1,\ldots,c-1\), and its \(c+1\) columns by the query symbols \(0,1,\ldots,c-1,*\), then \(P_1^{(c)}\) is precisely the incidence matrix of the problem.

\begin{theorem}[Deterministic Subset Query and Partial Match from \(\delta(R_1)\) and \(\delta(P_1^{(c)})\)]\label{thm:ds}
For every \(\eps>0\):
\begin{enumerate}
\item Subset Query in \(d\) dimensions admits an online deterministic data structure in which each insertion or query takes time
\[
O((\delta(R_1)+\eps)^d);
\]
\item for every fixed integer \(c\ge2\), \(c\)-ary Partial Match in \(d\) dimensions admits an online deterministic data structure in which each insertion or query takes time
\[
O((\delta(P_1^{(c)})+\eps)^d)\le O\!\left(\left(c^{\frac{c-1}{(c-1)+c\log_2 c}}+\eps\right)^d\right).
\]
\end{enumerate}
In particular, by \cref{thm:main}, these upper bounds are \(o(1.25^d)\) for Subset Query, \(o(1.26^d)\) for Partial Match, \(o(1.39^d)\) for \(3\)-ary Partial Match, and \(o(1.46^d)\) for \(4\)-ary Partial Match.
\end{theorem}

The proof of \cref{thm:ds} appears in \cref{sec:B.1}.

Finally, via the reductions established in \cite{AL25}, we derive smaller low-depth linear circuits for disjointness and for a certain class of non-Kronecker-power matrices.

For disjointness, \cref{thm:main} implies that the \(N\times N\) matrix \(R_n\) has a depth-2 linear circuit of size \(O(N^{1.244844})\), where \(N=2^n\). More generally, applying the reduction of \cite[Corollary 5.2]{AL25}, for every integer \(h\ge 1\), the matrix \(R_n\) has a depth-\(2h\) linear circuit of size \(O(N^{1+0.244844/h})\).

Moreover, for every function \(f:\{0,1\}^n\to\mathbb{R}\), define
\[
M_f[x,y]:=f(x_1\vee y_1,\ldots,x_n\vee y_n),
\qquad x,y\in\{0,1\}^n .
\]
Applying the reduction of \cite[Corollary 5.3]{AL25}, for every integer \(h\ge 1\), the matrix \(M_f\) has a depth-\(4h\) linear circuit over \(\mathbb{R}\), with arbitrary real weights allowed, of size \(O(N^{1+0.244844/h})\), where \(N=2^n\)\footnote{Indeed, let \(g\) be the upper M\"obius transform of \(f\), defined by \(g(z):=\sum_{w\ge z}(-1)^{\|w-z\|_1}f(w)\). Then
\[
M_f=A\Lambda_g A^{\T},
\qquad
A:=[\one_{x\le z}]_{x,z}=\begin{pmatrix}1 & 1\\0 & 1\end{pmatrix}^{\otimes n},
\]
where \(\Lambda_g\) is the diagonal matrix with entries \(g(z)\). The matrix \(A\) coincides, up to a permutation of rows and columns, with the Kronecker-power matrix \(R_n\). One may therefore compute the two factors \(A\) and \(A^{\T}\) using the circuits for \(R_n\), and absorb the diagonal factor \(\Lambda_g\) into an adjacent linear layer. This yields the claimed depth-\(4h\) real-weight linear circuit for \(M_f\) with the same size upper bound.}.

\section{Transfer Theorems and Circuit Instantiations}\label{sec:2}

Throughout the paper, for every positive integer \(N\ge 1\), we write \([N]:=\{0,\dots,N-1\}\), and we let \(\Delta_{[N]}\) denote the set of all probability distributions on \([N]\). All entropies, KL divergences, Lyapunov exponents, growth rates, and related quantities are measured in bits, that is, using base-\(2\) logarithms.

\subsection{Computation Trees and Scripted Rebalancing}\label{sec:2.1}

In this section, we formalize the scripted rebalancing described in \cref{sec:1.2}. Its input is a parameter set, and its output is a family of depth-2 linear circuits computing \(R_n\), one for each multiple \(n\) of \(k\) and each initial state.

Let \(k\geq1\) and \(H\geq1\) be integers, let \(I\) be a map from \([H]\) to circuits computing \(R_k\), let \(Z>1\) be a real number, and let \(\alpha,\beta\in(0,1)\) be real numbers. We call the tuple \(\mathrm{X}=(k,H,I,Z,\alpha,\beta)\) a \emph{parameter set}.

In our parameter sets, the map \(I\) is specified by an ordered pair \((\mathbf{c},\boldsymbol{\iota})\), where \(\mathbf{c}=(c_0,\ldots,c_{L-1})\) is an integer array satisfying \(0\leq c_0<\dots<c_{L-1}\leq2^k-1\), and \(\boldsymbol{\iota}:[H]\to[L]\) is a map. This specification means \(I(h)=S_k(c_{\boldsymbol{\iota}(h)})\), the \emph{Sergeev circuit} of order \(k\) with identifier \(c_{\boldsymbol{\iota}(h)}\) (\cite{Ser22}; \cite[Section 5.1]{AL25}).

For each identifier \(c\in\{0,\ldots,2^k-1\}\), \cref{sec:A.2} defines a depth-2 linear circuit \(S_k(c)\) computing \(R_k\) with exactly \(2^k\) middle gates. Equivalently, \(S_k(c)\) provides a rank-\(1\) decomposition of \(R_k\):
\[
S_k(c)=\sum_{r=0}^{2^k-1}(U_{c,r},V_{c,r}^\T),\qquad\text{so that }R_k=\sum_{r=0}^{2^k-1}U_{c,r}V_{c,r}^\T,
\]
where \(U_{c,r}\) and \(V_{c,r}\) are vectors indexed by \(\{0,1\}^k\). For a vector \(w\) indexed by \(\{0,1\}^k\), set
\[
\wnnz_{\gamma_0,\gamma_1}(w):=\sum_{x\in\{0,1\}^k}\one[w_x\neq0]\,(\gamma_0)^{k-\|x\|_1}(\gamma_1)^{\|x\|_1},
\]
where \(\|x\|_1\) denotes the number of \(1\)-bits of the index \(x\in\{0,1\}^k\). This is the total Bernoulli weight of the indices of the nonzero entries of \(w\), with weights \(\gamma_0\) and \(\gamma_1\) assigned to \(0\)-bits and \(1\)-bits, respectively.

For a rank-\(1\) term \((U_{c,r},V_{c,r}^\T)\), define its \emph{raw tilt} \(\theta_{c,r}\) as
\[
\theta_{c,r}:=\log_Z\left(\frac{\wnnz_{1-\beta,\beta}(V_{c,r})}{\wnnz_{1-\alpha,\alpha}(U_{c,r})}\right),
\]
and define its \emph{state offset} \(\Delta_{c,r}\) as the nearest integer \(\Delta_{c,r}:=\left\lfloor \theta_{c,r}+\frac12\right\rfloor\). Let
\[
\clamp_H(t):=
\begin{cases}
0,&t<0,\\
t,&0\le t\le H-1,\\
H-1,&t>H-1.
\end{cases}
\]
The circuits \(C_{\mathrm{X}}(n,h)\), for multiples \(n\) of \(k\) and states \(h\in[H]\), are defined recursively by
\[
C_{\mathrm{X}}(0,h):=([1],[1]),
\]
and, for \(n\ge k\),
\[
C_{\mathrm{X}}(n,h):=\sum_{r=0}^{2^k-1}(U_{c_i,r},V_{c_i,r}^\T)\otimes C_{\mathrm{X}}\!\left(n-k,\,\clamp_H(h+\Delta_{c_i,r})\right),\qquad\text{where }i:=\boldsymbol{\iota}(h).
\]

\begin{proposition}\label{prop:tree}
For every parameter set \(\mathrm{X}\) in the above format, every multiple \(n\) of \(k\), and every state \(h\in[H]\), the circuit \(C_{\mathrm{X}}(n,h)\) is a depth-2 linear circuit computing \(R_n\).
\end{proposition}

\begin{proof}[Proof sketch]
The circuit \(C_{\mathrm{X}}(n,h)\) is the root circuit of a computation tree \cite[Appendix B]{AL25} for the base matrix \(R_1\) with total weighted depth \(n\). A detailed proof appears in \cref{sec:A.3}.
\end{proof}

For a \(k\)-bit identifier \(c\in\{0,\ldots,2^k-1\}\), put \(\overline{c}:=2^k-1-c\), the \(k\)-bit complement of \(c\).

We next define a \emph{transposition} operation on parameter sets. If \(\mathrm{X}=(k,H,I,Z,\alpha,\beta)\) with \(I=(\mathbf{c},\boldsymbol{\iota})\), define \(\mathbf{c}^\T:=(\overline{c_{L-1}},\ldots,\overline{c_0})\), \(\boldsymbol{\iota}^\T(h):=L-1-\boldsymbol{\iota}(H-1-h)\), \(I^\T:=(\mathbf{c}^\T,\boldsymbol{\iota}^\T)\), and \(\mathrm{X}^\T:=(k,H,I^\T,Z,\beta,\alpha)\).

It is immediate from the definitions that applying this operation twice returns the original parameter set, \((\mathrm{X}^\T)^\T=\mathrm{X}\); the operation is involutive, as one expects of a transposition.

We now claim that the circuit family generated by the transposed parameter set is indeed isomorphic to the transpose of the original circuit family. The key ingredient is that \(S_k(c)^\T\cong S_k(\overline{c})\); see Fact~A.2.2.

\begin{proposition}\label{prop:transpose}
Suppose that no raw tilt \(\theta_{c,r}\) appearing in \(\mathrm{X}\) lies in \(\mathbb{Z}+\frac12\). Then
\[
\begin{aligned}
C_{\mathrm{X}}(n,h)^\T&\cong C_{\mathrm{X}^\T}(n,H-1-h),\\
C_{\mathrm{X}^\T}(n,h)^\T&\cong C_{\mathrm{X}}(n,H-1-h),
\end{aligned}
\]
for every multiple \(n\) of \(k\) and every \(h\in[H]\).
\end{proposition}

The proof of \cref{prop:transpose} appears in \cref{sec:B.2}.

\subsection{Nine Circuits for \texorpdfstring{\(R_1\)}{R1}, and \texorpdfstring{\(c+2\)}{c+2} Circuits for \texorpdfstring{\(P_1^{(c)}\)}{P1(c)}}\label{sec:2.2}

We now specify the nine circuits \(C_0,C_1,C_2,C_3,C_4,C_{\mathrm{J}},C_{\mathrm{S}},C_{\mathrm{R}},C_{\mathrm{R}^\T}\) announced in \cref{sec:1.2}.

The first five circuits \(C_0,C_1,C_2,C_3,C_4\) are defined directly, each as an ordered pair of matrices:

\begin{itemize}
\item \(C_0:=(A_0,B_0)\), where
\[
A_0=\begin{pmatrix}1&0\\0&1\end{pmatrix},\qquad B_0=\begin{pmatrix}1&1\\1&0\end{pmatrix};
\]
\item \(C_1:=(A_1,B_1)\), where
\[
A_1=\begin{pmatrix}1&1\\1&0\end{pmatrix},\qquad B_1=\begin{pmatrix}1&0\\0&1\end{pmatrix};
\]
\item \(C_2:=(A_2,B_2)\), where
\[
A_2=\begin{pmatrix}1&0\\1&1\end{pmatrix},\qquad B_2=\begin{pmatrix}1&1\\0&-1\end{pmatrix};
\]
\item \(C_3:=(A_3,B_3)\), where
\[
A_3=\begin{pmatrix}1&0&0&0\\0&1&0&1\\0&1&1&0\\0&1&0&0\end{pmatrix},\qquad B_3=\begin{pmatrix}1&1&1&1\\1&0&0&0\\0&1&0&0\\0&0&1&0\end{pmatrix};
\]
\item \(C_4:=(A_4,B_4)\), where
\[
A_4=\begin{pmatrix}1&1&0&0\\1&0&1&0\\1&0&0&1\\1&0&0&0\end{pmatrix},\qquad B_4=\begin{pmatrix}1&0&0&0\\0&1&1&1\\0&0&1&0\\0&1&0&0\end{pmatrix}.
\]
\end{itemize}

One checks directly that
\[
\begin{aligned}
A_0B_0=A_1B_1=A_2B_2&=\begin{pmatrix}1&1\\1&0\end{pmatrix}=R_1,\\
A_3B_3=A_4B_4&=\begin{pmatrix}1&1&1&1\\1&0&1&0\\1&1&0&0\\1&0&0&0\end{pmatrix}=R_2.
\end{aligned}
\]
Hence \(C_0,C_1,C_2\) are depth-2 linear circuits computing \(R_1\), and \(C_3,C_4\) are depth-2 linear circuits computing \(R_2\).

The remaining four circuits \(C_{\mathrm{J}},C_{\mathrm{S}},C_{\mathrm{R}},C_{\mathrm{R}^\T}\) are obtained by applying scripted rebalancing to four parameter sets, named Jessica, Sonetto, Regulus, and Regulus Transpose, which are specified in \cref{tab:parameter-sets}\footnote{How did we discover these parameter sets? At a high level, we first fix \(\alpha,\beta\) and choose a moderately sized \(H\). We then run a bottom-up dynamic program, inspired by \cite[Appendix B]{AL25}, whose purpose is to choose, for each state \(h\in[H]\), which Sergeev circuit identifier should be used at that state. We run this dynamic program for increasing tree depths until the resulting state-to-identifier map stops changing. We then discard the dynamic program itself and record the stabilized map \(I\), equivalently the multiplicity array shown in the table, as the parameter set. We do not present the definition or technical details of this dynamic program, since they are not relevant to the correctness of this paper.}. In \cref{tab:parameter-sets}, the map \(\boldsymbol{\iota}\) is specified by its \emph{multiplicity array}: if the multiplicity array is \((m_0,\ldots,m_{L-1})\), then \(\boldsymbol{\iota}(h)=i\) for the unique \(i\in[L]\) satisfying \(m_0+\cdots+m_{i-1}\le h<m_0+\cdots+m_i\).

\begin{table}[H]
\centering
\small
\begin{tabular}{@{}lllll@{}}
\toprule
Parameter & Jessica & Sonetto & Regulus & Regulus Transpose \\
\midrule
\(k\) & \(7\) & \(7\) & \(7\) & \(7\) \\
\(H\) & \(612\) & \(360\) & \(360\) & \(360\) \\
\(\mathbf{c}\) & \((0,20,42,85,107,127)\) & \((0,21,42,85,106,127)\) & \((0,42,43,45,\) & \((0,41,42,50,\) \\
 & & & \(\ \ 77,85,86,127)\) & \(\ \ 82,84,85,127)\) \\
multiplicities & \((14,22,270,270,22,14)\) & \((8,22,150,150,22,8)\) & \((8,36,20,20,\) & \((8,44,64,160,\) \\
 & & & \(\ \ 160,64,44,8)\) & \(\ \ 20,20,36,8)\) \\
\(Z\) & \(2^{1/3}\) & \(2^{1/3}\) & \(2^{1/3}\) & \(2^{1/3}\) \\
\(\alpha\) & \(0.5\) & \(0.38\) & \(0.336\) & \(0.414\) \\
\(\beta\) & \(0.5\) & \(0.38\) & \(0.414\) & \(0.336\) \\
\bottomrule
\end{tabular}
\caption{The four parameter sets used in this paper.}
\label{tab:parameter-sets}
\end{table}

One checks that the transpose of Jessica is Jessica itself, the transpose of Sonetto is Sonetto itself, and the transposes of Regulus and Regulus Transpose are each other.

Moreover, for these four parameter sets, all raw tilts are bounded away from half-integers: the smallest absolute distance to the nearest half-integer is approximately \(1.33\cdot10^{-3}\) for Jessica, \(1.71\cdot10^{-2}\) for Sonetto, and \(4.70\cdot10^{-3}\) for Regulus and Regulus Transpose. Hence \cref{prop:transpose} gives
\[
\begin{aligned}
C_{\mathrm{Jessica}}(n,h)^\T&\cong C_{\mathrm{Jessica}}(n,611-h),\\
C_{\mathrm{Sonetto}}(n,h)^\T&\cong C_{\mathrm{Sonetto}}(n,359-h),\\
C_{\mathrm{Regulus}}(n,h)^\T&\cong C_{\mathrm{Regulus}^{\T}}(n,359-h).
\end{aligned}
\]

We now fix \(m_*:=10^{11}\) and define the four circuits \(C_{\mathrm{J}},C_{\mathrm{S}},C_{\mathrm{R}},C_{\mathrm{R}^\T}\) by
\[
\begin{aligned}
C_{\mathrm{J}}&:=C_{\mathrm{Jessica}}(7m_*,0),\\
C_{\mathrm{S}}&:=C_{\mathrm{Sonetto}}(7m_*,0),\\
C_{\mathrm{R}}&:=C_{\mathrm{Regulus}}(7m_*,0),\\
C_{\mathrm{R}^\T}&:=C_{\mathrm{Regulus}^{\T}}(7m_*,0).
\end{aligned}
\]
By \cref{prop:tree}, the circuits \(C_{\mathrm{J}},C_{\mathrm{S}},C_{\mathrm{R}},C_{\mathrm{R}^\T}\) are depth-2 linear circuits computing \(R_{7m_*}=R_{7\cdot10^{11}}\).

Finally, recall that for an integer \(c\ge2\), the \(c\)-ary partial-match incidence matrix is
\[
P_1^{(c)}=\begin{pmatrix}I_c&\one_c\end{pmatrix}=
\begin{pmatrix}
1&0&\cdots&0&1\\
0&1&\cdots&0&1\\
\vdots&\vdots&\ddots&\vdots&\vdots\\
0&0&\cdots&1&1
\end{pmatrix}.
\]
We specify the \(c+2\) circuits \(C'_{\mathsf{pull}},C'_{\mathsf{push}},C'_{(0)},\ldots,C'_{(c-1)}\) announced in \cref{sec:1.2} in \cref{tab:pm-circuits}, where \(A'_{(i)}:=I_c+(\one_c-e_i)e_i^{\T}\) and \(B'_{(i)}:=\begin{pmatrix}I_c-(\one_c-e_i)e_i^{\T}&e_i\end{pmatrix}\).

\begin{table}[H]
\centering
\small
\begin{tabular}{@{}lccccc@{}}
\toprule
Circuit & \# middle gates & \(L_C(x)\) & \(R_C(y)\) & \(f_C(p)\) & \(g_C(q)\) \\
\midrule
\(C'_{\mathsf{pull}}:=(I_c,P_1^{(c)})\) & \(c\) & \(1\) & \(\begin{cases}1,&y\neq c\\ c,&y=c\end{cases}\) & \(0\) & \(q_c\log_2 c\) \\
\addlinespace
\(C'_{\mathsf{push}}:=(P_1^{(c)},I_{c+1})\) & \(c+1\) & \(2\) & \(1\) & \(1\) & \(0\) \\
\addlinespace
\(C'_{(i)}:=(A'_{(i)},B'_{(i)})\) & \(c\) & \(\begin{cases}2,&x\neq i\\ 1,&x=i\end{cases}\) & \(\begin{cases}1,&y\neq i\\ c,&y=i\end{cases}\) & \(1-p_i\) & \(q_i\log_2 c\) \\
\bottomrule
\end{tabular}
\caption{The \(c+2\) circuits for \(P_1^{(c)}\), together with their gate degrees and normalized log-degree density functions.}
\label{tab:pm-circuits}
\end{table}

Equivalently, in the language of gates and wires:
\begin{itemize}
\item in \(C'_{\mathsf{pull}}\), middle gate \(j\) receives weight \(1\) from input gate \(j\) and sends weight \(1\) to output gates \(j\) and \(c\);
\item in \(C'_{\mathsf{push}}\), for \(j<c\), middle gate \(j\) receives weight \(1\) from input gate \(j\) and sends weight \(1\) to output gate \(j\), while middle gate \(c\) receives weight \(1\) from every input gate and sends weight \(1\) to output gate \(c\);
\item in \(C'_{(i)}\), for \(j\neq i\), middle gate \(j\) receives weight \(1\) from input gate \(j\), sends weight \(1\) to output gate \(j\), and sends weight \(-1\) to output gate \(i\), while middle gate \(i\) receives weight \(1\) from every input gate and sends weight \(1\) to output gates \(i\) and \(c\).
\end{itemize}

One checks that \(I_cP_1^{(c)}=P_1^{(c)}\) and \(P_1^{(c)}I_{c+1}=P_1^{(c)}\). Also, since \(e_i^{\T}(\one_c-e_i)=0\) and \(e_i^{\T}e_i=1\),
\[
\begin{aligned}
A'_{(i)}B'_{(i)}
&=\left(I_c+(\one_c-e_i)e_i^{\T}\right)\begin{pmatrix}I_c-(\one_c-e_i)e_i^{\T}&e_i\end{pmatrix} \\
&=\begin{pmatrix}I_c-(\one_c-e_i)e_i^{\T}+(\one_c-e_i)e_i^{\T}-(\one_c-e_i)\bigl(e_i^{\T}(\one_c-e_i)\bigr)e_i^{\T}&\ \ e_i+(\one_c-e_i)(e_i^{\T}e_i)\end{pmatrix} \\
&=\begin{pmatrix}I_c&\one_c\end{pmatrix}=P_1^{(c)}.
\end{aligned}
\]
Hence all \(c+2\) circuits indeed compute \(P_1^{(c)}\).

\subsection{The Size and Degree Transfer Theorems}\label{sec:2.3}

For an integer \(\mathfrak{c}\ge2\) and an index \(x\in[\mathfrak{c}]^n\), define its \emph{normalized frequency array} by
\[
\freq(x)_i:=\frac1n\cdot|\{j\in[n]:x_j=i\}|,\qquad i=0,\ldots,\mathfrak{c}-1,
\]
so that \(\freq(x)\in\Delta_{[\mathfrak{c}]}\). For \(p\in\freq([\mathfrak{c}]^n)\), define the \emph{frequency class} \(\FC_n(p):=\{x\in[\mathfrak{c}]^n:\freq(x)=p\}\).

For an integer \(\mathfrak{c}\ge2\) and a probability distribution \(p=(p_0,\ldots,p_{\mathfrak{c}-1})\in\Delta_{[\mathfrak{c}]}\), write
\[
\mathrm{H}_\mathfrak{c}(p):=-\sum_{i=0}^{\mathfrak{c}-1}p_i\log_2 p_i,\qquad 0\log_2 0:=0.
\]
For \(p\in[0,1]\), we also write \(\mathrm{h}(p):=\mathrm{H}_2(\Ber(p))=-p\log_2 p-(1-p)\log_2(1-p)\).

\begin{definition}[Normalized log-degree density functions]\label{def:density}
Fix a matrix \(M\in\mathbb{Z}^{\mathfrak{c}_0\times\mathfrak{c}_1}\) with no identically zero row or column, and index its rows by \([\mathfrak{c}_0]=\{0,\ldots,\mathfrak{c}_0-1\}\) and its columns by \([\mathfrak{c}_1]=\{0,\ldots,\mathfrak{c}_1-1\}\).

If \(C=(A,B)\) is a depth-2 linear circuit computing \(M^{\otimes k}\) for some \(k\ge1\), let \(L_C(x):=\nnz(A[x,*])\) denote the degree of input gate \(x\) in \(C\), and let \(R_C(y):=\nnz(B[*,y])\) denote the degree of output gate \(y\) in \(C\). Define
\[
\begin{aligned}
f_C(p)&:=\frac1k\,\mathbb{E}_{X\sim p^k}\bigl[\log_2 L_C(X)\bigr],&p\in\Delta_{[\mathfrak{c}_0]},\\[4pt]
g_C(q)&:=\frac1k\,\mathbb{E}_{Y\sim q^k}\bigl[\log_2 R_C(Y)\bigr],&q\in\Delta_{[\mathfrak{c}_1]}.
\end{aligned}
\]
We call \(f_C\) and \(g_C\) the left and right \emph{normalized log-degree density functions} of \(C\), respectively.
\end{definition}

\begin{definition}[Manifestations of circuits and of finite circuit collections]\label{def:manifest}
Let \(M\in\mathbb{Z}^{\mathfrak{c}_0\times \mathfrak{c}_1}\) be as in \cref{def:density}.

If \(C\) is a circuit computing \(M^{\otimes k}\) for some \(k\ge1\), define a map \(P_C\) from \(\Delta_{[\mathfrak{c}_0]}\times\Delta_{[\mathfrak{c}_1]}\) to the Cartesian plane by
\[
P_C(p,q):=(f_C(p),g_C(q)).
\]
If \(\mathcal{C}\) is a finite collection of depth-2 linear circuits, each computing \(M^{\otimes k}\) for some possibly different \(k\ge1\), define a map \(P_\mathcal{C}\) from \(\Delta_{[\mathfrak{c}_0]}\times\Delta_{[\mathfrak{c}_1]}\) to convex polygons\footnote{Including degenerate polygons, such as a single point or a line segment.} in the Cartesian plane by
\[
P_\mathcal{C}(p,q):=\conv\{P_C(p,q):C\in\mathcal{C}\},
\]
the convex hull of the points \(P_C(p,q)\) over all \(C\in\mathcal{C}\).

We call \(P_C(p,q)\) and \(P_\mathcal{C}(p,q)\) the \emph{manifestations} of \(C\) and of \(\mathcal{C}\) at \((p,q)\), respectively.
\end{definition}

\begin{theorem}[Realization of convex-hull points]\label{thm:realization}
Let \(M\in\mathbb{Z}^{\mathfrak{c}_0\times \mathfrak{c}_1}\) be as in \cref{def:density}, and let \(\mathcal{C}\) be a finite circuit collection as in \cref{def:manifest}.

For every \(\eps>0\), there exists \(n_0\geq1\), depending only on \(M\), \(\mathcal{C}\), and \(\eps\), such that the following holds for every \(n\geq n_0\). For every \((p,q)\in\freq([\mathfrak{c}_0]^n)\times\freq([\mathfrak{c}_1]^n)\) and every \((a,b)\in P_\mathcal{C}(p,q)\), there exists a depth-2 linear circuit \(E_n\) that computes the submatrix \(M^{\otimes n}[\FC_n(p),\FC_n(q)]\) and satisfies
\[
L_{E_n}(x)\le 2^{(a+\eps)n},
\qquad
R_{E_n}(y)\le 2^{(b+\eps)n},
\]
for every \(x\in\FC_n(p)\) and every \(y\in\FC_n(q)\).
\end{theorem}

\begin{proof}[Proof sketch]
First suppose \(|\mathcal{C}|=2\), say \(\mathcal{C}=\{C_{\mathrm{hi}},C_{\mathrm{lo}}\}\). In this case, \((a,b)\) is a convex combination of \(P_{C_{\mathrm{hi}}}(p,q)\) and \(P_{C_{\mathrm{lo}}}(p,q)\). We form a Kronecker product of copies of \(C_{\mathrm{hi}}\) and \(C_{\mathrm{lo}}\), assigned to the corresponding fractions of tensor coordinates, and restrict the resulting circuit to \(M^{\otimes n}[\FC_n(p),\FC_n(q)]\). By Hoeffding's inequality, except for a negligible fraction of the input and output gates, the input-gate degrees are at most \(2^{an+n^{2/3}}\) and the output-gate degrees are at most \(2^{bn+n^{2/3}}\). Removing the exceptional gates yields a circuit computing a ``broken copy'' of \(M^{\otimes n}[\FC_n(p),\FC_n(q)]\). Finally, the hole-fixing lemma of \cite{AL25} converts this broken copy into a full circuit, increasing the degrees by only a subexponential multiplicative factor.

For an arbitrary finite collection \(\mathcal{C}\), start from \((a,b)\) and move along the ray \((a-t,b-t)\) until the boundary of the convex polygon \(P_{\mathcal{C}}(p,q)\) is reached. The resulting boundary point is a convex combination of \(P_{C_{\mathrm{I}}}(p,q)\) and \(P_{C_{\mathrm{II}}}(p,q)\) for some \(C_{\mathrm{I}},C_{\mathrm{II}}\in\mathcal{C}\). Applying the two-circuit case to this boundary point gives upper bounds at least as strong as those required by the theorem. A detailed proof appears in \cref{sec:B.2}.
\end{proof}

\subsubsection{The Size Transfer Theorem}\label{sec:2.3.1}

\begin{theorem}[Size transfer theorem]\label{thm:size-transfer}
Let \(M\in\mathbb{Z}^{\mathfrak{c}_0\times \mathfrak{c}_1}\) be as in \cref{def:density}, and let \(\mathcal{C}\) be any finite circuit collection as in \cref{def:manifest}. Then
\[
\sigma(M)\le2^{\xi(\mathcal{C})},\qquad\xi(\mathcal{C}):=
\max_{(p,q)\in\Delta_{[\mathfrak{c}_0]}\times\Delta_{[\mathfrak{c}_1]}}\left\{
\min_{(a,b)\in\conv\{P_C(p,q):C\in\mathcal{C}\}}
\max\{a+\mathrm{H}_{\mathfrak{c}_0}(p),\,b+\mathrm{H}_{\mathfrak{c}_1}(q)\}\right\}.
\]
\end{theorem}

\begin{proof}[Proof sketch]
Build the circuit piece by piece. For each pair \((p,q)\in\freq([\mathfrak{c}_0]^n)\times\freq([\mathfrak{c}_1]^n)\), choose \((a,b)\in\conv\{P_C(p,q):C\in\mathcal{C}\}\) minimizing \(\max\{a+\mathrm{H}_{\mathfrak{c}_0}(p),b+\mathrm{H}_{\mathfrak{c}_1}(q)\}\). By \cref{thm:realization}, there is a circuit computing \(M^{\otimes n}[\FC_n(p),\FC_n(q)]\) with input-gate degrees at most \(2^{(a+\eps)n}\) and output-gate degrees at most \(2^{(b+\eps)n}\). Glue these circuits together by taking their concatenation sum. There are only polynomially many possible pairs of frequency arrays \((p,q)\), and the frequency classes satisfy \(|\FC_n(p)|\le 2^{\mathrm{H}_{\mathfrak{c}_0}(p)n}\) and \(|\FC_n(q)|\le 2^{\mathrm{H}_{\mathfrak{c}_1}(q)n}\). Hence the resulting circuit has size at most \(2^{(\xi(\mathcal{C})+\eps)n+o(n)}\). Letting \(n\to\infty\) and then \(\eps\downarrow0\) gives \(\sigma(M)\le 2^{\xi(\mathcal{C})}\). A full proof appears in \cref{sec:B.2}.
\end{proof}

To see how \(\mathcal{C}\) performs at different pairs of probability distributions \((p,q)\), define
\[
\xi(\mathcal{C},p,q):=\min_{(a,b)\in\conv\{P_C(p,q):C\in\mathcal{C}\}}\max\{a+\mathrm{H}_{\mathfrak{c}_0}(p),\,b+\mathrm{H}_{\mathfrak{c}_1}(q)\}.
\]
We call the map \((p,q)\mapsto \xi(\mathcal{C},p,q)\) from \(\Delta_{[\mathfrak{c}_0]}\times\Delta_{[\mathfrak{c}_1]}\) to \(\mathbb{R}_{\geq0}\) the \emph{size landscape} of \(\mathcal{C}\).

\subsubsection{The Degree Transfer Theorem}\label{sec:2.3.2}

\begin{theorem}[Degree transfer theorem]\label{thm:degree-transfer}
Let \(M\in\mathbb{Z}^{\mathfrak{c}_0\times \mathfrak{c}_1}\) be as in \cref{def:density}, and let \(\mathcal{C}\) be any finite circuit collection as in \cref{def:manifest}. Then
\[
\delta(M)\le2^{\phi(\mathcal{C})},\qquad\phi(\mathcal{C}):=
\max_{(p,q)\in\Delta_{[\mathfrak{c}_0]}\times\Delta_{[\mathfrak{c}_1]}}\left\{
\min_{(a,b)\in\conv\{P_C(p,q):C\in\mathcal{C}\}}\max\{a,b\}\right\}.
\]
\end{theorem}

\begin{proof}[Proof sketch]
Build the circuit piece by piece. For each pair \((p,q)\in\freq([\mathfrak{c}_0]^n)\times\freq([\mathfrak{c}_1]^n)\), choose \((a,b)\in\conv\{P_C(p,q):C\in\mathcal{C}\}\) minimizing \(\max\{a,b\}\). By \cref{thm:realization}, there is a circuit computing \(M^{\otimes n}[\FC_n(p),\FC_n(q)]\) with input-gate degrees at most \(2^{(a+\eps)n}\) and output-gate degrees at most \(2^{(b+\eps)n}\). Glue these circuits together by taking their concatenation sum. A fixed input or output gate is involved in only polynomially many summands of the concatenation sum, so the resulting circuit has maximum input-output degree at most \(2^{(\phi(\mathcal{C})+\eps)n+o(n)}\). Letting \(n\to\infty\) and then \(\eps\downarrow0\) gives \(\delta(M)\le 2^{\phi(\mathcal{C})}\). A full proof appears in \cref{sec:B.2}.
\end{proof}

To see how \(\mathcal{C}\) performs at different pairs of probability distributions \((p,q)\), define
\[
\phi(\mathcal{C},p,q):=\min_{(a,b)\in\conv\{P_C(p,q):C\in\mathcal{C}\}}\max\{a,b\}.
\]
We call the map \((p,q)\mapsto \phi(\mathcal{C},p,q)\) from \(\Delta_{[\mathfrak{c}_0]}\times\Delta_{[\mathfrak{c}_1]}\) to \(\mathbb{R}_{\geq0}\) the \emph{degree landscape} of \(\mathcal{C}\).

\subsubsection{Instantiation of the Transfer Theorems}\label{sec:2.3.3}

For each of the four circuit families \(\mathrm{X}\in\{\mathrm{Jessica},\mathrm{Sonetto},\mathrm{Regulus},\mathrm{Regulus}^{\T}\}\), \cref{sec:3} will define a function \(\ell_{\mathrm{X}}:[0,1]\to\mathbb{R}_{\ge0}\).

By \cref{prop:envelope-props}, with \(m_*=10^{11}\), for every \(m\ge m_*\), every \(h\in[H]\), and every \(p\in[0,1]\), we have
\[
\begin{aligned}
f_{C_{\mathrm{Jessica}}(7m,h)}(\Ber(p))&\le\ell_{\mathrm{Jessica}}(p),\\
f_{C_{\mathrm{Sonetto}}(7m,h)}(\Ber(p))&\le\ell_{\mathrm{Sonetto}}(p),\\
f_{C_{\mathrm{Regulus}}(7m,h)}(\Ber(p))&\le\ell_{\mathrm{Regulus}}(p),\\
f_{C_{\mathrm{Regulus}^{\T}}(7m,h)}(\Ber(p))&\le\ell_{\mathrm{Regulus}^{\T}}(p).
\end{aligned}
\]
Combining these bounds with the definitions of \(C_{\mathrm{J}},C_{\mathrm{S}},C_{\mathrm{R}},C_{\mathrm{R}^{\T}}\) and \cref{prop:transpose}, we obtain, for all \(p,q\in[0,1]\),
\[
\begin{aligned}
f_{C_{\mathrm{J}}}(\Ber(p))&\le\ell_{\mathrm{Jessica}}(p),&
g_{C_{\mathrm{J}}}(\Ber(q))&\le\ell_{\mathrm{Jessica}}(q),\\
f_{C_{\mathrm{S}}}(\Ber(p))&\le\ell_{\mathrm{Sonetto}}(p),&
g_{C_{\mathrm{S}}}(\Ber(q))&\le\ell_{\mathrm{Sonetto}}(q),\\
f_{C_{\mathrm{R}}}(\Ber(p))&\le\ell_{\mathrm{Regulus}}(p),&
g_{C_{\mathrm{R}}}(\Ber(q))&\le\ell_{\mathrm{Regulus}^{\T}}(q),\\
f_{C_{\mathrm{R}^{\T}}}(\Ber(p))&\le\ell_{\mathrm{Regulus}^{\T}}(p),&
g_{C_{\mathrm{R}^{\T}}}(\Ber(q))&\le\ell_{\mathrm{Regulus}}(q).
\end{aligned}
\]

We first instantiate the size transfer theorem on the singleton collection \(\{C_{\mathrm{J}}\}\). By \cref{thm:size-transfer},
\[
\begin{aligned}
\sigma(R_1)\le2^{\xi(\{C_{\mathrm{J}}\})},\qquad\xi(\{C_{\mathrm{J}}\})&=
\max_{(p,q)\in\Delta_{\{0,1\}}\times\Delta_{\{0,1\}}}\left\{
\min_{(a,b)\in\{P_{C_{\mathrm{J}}}(p,q)\}}
\max\{a+\mathrm{H}_2(p),\,b+\mathrm{H}_2(q)\}\right\}\\
&=\max\left\{\max_{p\in[0,1]}\{f_{C_{\mathrm{J}}}(\Ber(p))+\mathrm{h}(p)\},\ \max_{q\in[0,1]}\{g_{C_{\mathrm{J}}}(\Ber(q))+\mathrm{h}(q)\}\right\}\\
&\le\max_{p\in[0,1]}\{\ell_{\mathrm{Jessica}}(p)+\mathrm{h}(p)\}.
\end{aligned}
\]
We claim that
\[
\max_{p\in[0,1]}\{\ell_{\mathrm{Jessica}}(p)+\mathrm{h}(p)\}<1.244844.
\tag{\(*\)}\label{eq:star}
\]
In \cref{sec:3.3.2}, we define a deterministic program \(\mathcal{V}_{\mathrm{size}}\), the \emph{size-inequality verifier}, and \cref{thm:size-verifier} proves that if \(\mathcal{V}_{\mathrm{size}}\) accepts, then \eqref{eq:star} holds.

\begin{corollary}[Upper bound on \(\sigma(R_1)\)]\label{cor:sigma}
\(\sigma(R_1)<2^{1.244844}\).
\end{corollary}

\begin{proof}
This follows from \cref{thm:size-verifier} together with the fact that we ran \(\mathcal{V}_{\mathrm{size}}\) and it accepts.
\end{proof}

For the five circuits \(C_0,\ldots,C_4\) defined in \cref{sec:2.2}, a direct computation gives
\[
\begin{aligned}
f_{C_0}(\Ber(p))&=0,&g_{C_0}(\Ber(q))&=1-q,\\
f_{C_1}(\Ber(p))&=1-p,&g_{C_1}(\Ber(q))&=0,\\
f_{C_2}(\Ber(p))&=p,&g_{C_2}(\Ber(q))&=q,\\
f_{C_3}(\Ber(p))&=p(1-p),&g_{C_3}(\Ber(q))&=\tfrac12(1-q^2),\\
f_{C_4}(\Ber(p))&=\tfrac12(1-p^2),&g_{C_4}(\Ber(q))&=q(1-q).
\end{aligned}
\]
For the three circuits \(C_{\mathrm{S}},C_{\mathrm{R}},C_{\mathrm{R}^{\T}}\), the bounds displayed above give
\[
\begin{aligned}
f_{C_{\mathrm{S}}}(\Ber(p))&\le\ell_{\mathrm{Sonetto}}(p),&
g_{C_{\mathrm{S}}}(\Ber(q))&\le\ell_{\mathrm{Sonetto}}(q),\\
f_{C_{\mathrm{R}}}(\Ber(p))&\le\ell_{\mathrm{Regulus}}(p),&
g_{C_{\mathrm{R}}}(\Ber(q))&\le\ell_{\mathrm{Regulus}^{\T}}(q),\\
f_{C_{\mathrm{R}^{\T}}}(\Ber(p))&\le\ell_{\mathrm{Regulus}^{\T}}(p),&
g_{C_{\mathrm{R}^{\T}}}(\Ber(q))&\le\ell_{\mathrm{Regulus}}(q).
\end{aligned}
\]
Let \(\mathcal{C}_8:=\{C_0,C_1,C_2,C_3,C_4,C_{\mathrm{S}},C_{\mathrm{R}},C_{\mathrm{R}^{\T}}\}\) be the collection of eight circuits used in the degree-transfer instantiation below.

Define eight points in the first quadrant of the Cartesian plane by
\[
\begin{aligned}
Q_0(p,q)&:=(0,\,1-q),\\
Q_1(p,q)&:=(1-p,\,0),\\
Q_2(p,q)&:=(p,\,q),\\
Q_3(p,q)&:=(p(1-p),\,\tfrac12(1-q^2)),\\
Q_4(p,q)&:=(\tfrac12(1-p^2),\,q(1-q)),\\
Q_5(p,q)&:=(\ell_{\mathrm{Sonetto}}(p),\,\ell_{\mathrm{Sonetto}}(q)),\\
Q_6(p,q)&:=(\ell_{\mathrm{Regulus}}(p),\,\ell_{\mathrm{Regulus}^{\T}}(q)),\\
Q_7(p,q)&:=(\ell_{\mathrm{Regulus}^{\T}}(p),\,\ell_{\mathrm{Regulus}}(q)),
\end{aligned}
\]
and set
\[
\mathbf{z}(p,q):=\min_{(a,b)\in\conv\{Q_0(p,q),\dots,Q_7(p,q)\}}\max\{a,b\}.
\]

We now instantiate the degree transfer theorem on the collection \(\mathcal{C}_8\). By \cref{thm:degree-transfer},
\[
\begin{aligned}
\delta(R_1)\le2^{\phi(\mathcal{C}_8)},\qquad\phi(\mathcal{C}_8)&=
\max_{(p,q)\in\Delta_{\{0,1\}}\times\Delta_{\{0,1\}}}\left\{
\min_{(a,b)\in\conv\{P_C(p,q):C\in\mathcal{C}_8\}}\max\{a,b\}\right\}\\
&\le\max_{(p,q)\in[0,1]^2}\mathbf{z}(p,q).
\end{aligned}
\]
We claim that
\[
\max_{(p,q)\in[0,1]^2}\mathbf{z}(p,q)<0.319866.
\tag{\(**\)}\label{eq:starstar}
\]
In \cref{sec:3.3.3}, we define a deterministic program \(\mathcal{V}_{\mathrm{degree}}\), the \emph{degree-inequality verifier}, and \cref{thm:degree-verifier} proves that if \(\mathcal{V}_{\mathrm{degree}}\) accepts, then \eqref{eq:starstar} holds.

\begin{corollary}[Upper bound on \(\delta(R_1)\)]\label{cor:delta}
\(\delta(R_1)<2^{0.319866}\).
\end{corollary}

\begin{proof}
This follows from \cref{thm:degree-verifier} together with the fact that we ran \(\mathcal{V}_{\mathrm{degree}}\) and it accepts.
\end{proof}

\begin{remark}
Examining \(\ell_{\mathrm{Jessica}}(p)+\mathrm{h}(p)\) for \(p\in[0,1]\), and \(\mathbf{z}(p,q)\) for \((p,q)\in[0,1]^2\), reveals concrete bottlenecks for future optimization. To improve the upper bound on \(\sigma(R_1)\), one should design a circuit \(C\) minimizing \(\frac{f_C(\Ber(p_0))+g_C(\Ber(p_0))}{2}\) near \(p_0\approx0.37\). To improve the upper bound on \(\delta(R_1)\), one should design a circuit \(C\) satisfying \(\frac{f_C(\Ber(p_1))+g_C(\Ber(p_1))}{2}<p_1\) for \(p_1\) as small as possible, with the relevant range here being \(p_1\approx0.319\).
\end{remark}

Finally, we instantiate the size transfer theorem on the collection \(\{C'_{\mathsf{pull}},C'_{\mathsf{push}}\}\) and the degree transfer theorem on the collection \(\{C'_{\mathsf{pull}},C'_{(0)},\ldots,C'_{(c-1)}\}\), all defined in \cref{sec:2.2}.

\begin{theorem}[Upper bounds on \(\sigma(P_1^{(c)})\) and \(\delta(P_1^{(c)})\)]\label{thm:pm}
For every integer \(c\ge2\),
\[
\sigma(P_1^{(c)})\leq c\cdot\psi(\log_2c),
\qquad
\delta(P_1^{(c)})\leq c^{\frac{c-1}{(c-1)+c\log_2c}},
\]
where \(\psi(r)\) denotes the unique root \(x>1\) of the equation \(x=1+\frac{1}{x^r}\).
\end{theorem}

\begin{proof}[Proof sketch]
Let \(x:=\psi(\log_2 c)\), so that \(x=1+\frac{1}{x^{\log_2 c}}\). For the upper bound on \(\sigma(P_1^{(c)})\), apply \cref{thm:size-transfer} to \(\{C'_{\mathsf{pull}},C'_{\mathsf{push}}\}\). At \((p,q)\), the two manifestation points are \((0,q_c\log_2 c)\) and \((1,0)\). We choose the point \((a,b)=(\log_2 x,(1-\log_2 x)q_c\log_2 c)\) on their line segment. Using the nonnegativity of \(D_{\mathrm{KL}}(p\Vert p^*)\), where \(p^*=(1/c,\ldots,1/c)\), and of \(D_{\mathrm{KL}}(q\Vert q^*)\), where \(q_i^*=1/(cx)\) for \(i<c\) and \(q_c^*=(x-1)/x\), together with the identity \(x-1=x^{-\log_2 c}=c^{-\log_2 x}\), we obtain \(a+\mathrm{H}_c(p)\le\log_2 c+\log_2 x\) and \(b+\mathrm{H}_{c+1}(q)\le\log_2 c+\log_2 x\). Hence \(\xi(\{C'_{\mathsf{pull}},C'_{\mathsf{push}}\})\le \log_2 c+\log_2 x\), and \cref{thm:size-transfer} gives \(\sigma(P_1^{(c)})\le 2^{\log_2 c+\log_2 x}=c\cdot\psi(\log_2 c)\).

For the upper bound on \(\delta(P_1^{(c)})\), apply \cref{thm:degree-transfer} to \(\{C'_{\mathsf{pull}},C'_{(0)},\ldots,C'_{(c-1)}\}\). At \((p,q)\), the \(c+1\) manifestation points are \((0,q_c\log_2 c)\) and \((1-p_i,q_i\log_2 c)\) for \(i=0,\ldots,c-1\). We choose the point of their convex hull obtained from the weights \(\frac{c-1}{(c-1)+c\log_2 c},\frac{\log_2 c}{(c-1)+c\log_2 c},\ldots,\frac{\log_2 c}{(c-1)+c\log_2 c}\). Using \(\sum_i(1-p_i)=c-1\), \(\sum_i q_i=1-q_c\), and \(\log_2 c\le c-1\), its left coordinate equals \(\frac{(c-1)\log_2 c}{(c-1)+c\log_2 c}\) and its right coordinate is \(\frac{\log_2 c}{(c-1)+c\log_2 c}\bigl((c-1)q_c+(\log_2 c)(1-q_c)\bigr)\le\frac{(c-1)\log_2 c}{(c-1)+c\log_2 c}\). Hence \(\phi(\{C'_{\mathsf{pull}},C'_{(0)},\ldots,C'_{(c-1)}\})\le \frac{(c-1)\log_2 c}{(c-1)+c\log_2 c}\), and \cref{thm:degree-transfer} gives \(\delta(P_1^{(c)})\le 2^{\frac{(c-1)\log_2 c}{(c-1)+c\log_2 c}}=c^{\frac{c-1}{(c-1)+c\log_2 c}}\). A full proof appears in \cref{sec:B.2}.
\end{proof}

\section{Lyapunov Exponents and Machine Verification}\label{sec:3}

Let \(\mathcal{M}\) be a finite collection of nonnegative matrices, and let \(\mu\) be a probability distribution on \(\mathcal{M}\). Consider a sequence of independent and identically distributed random matrices \(A_1, A_2, \ldots\) drawn from \(\mathcal{M}\) according to \(\mu\). The associated \emph{quenched Lyapunov exponent} is defined as the almost-sure limit
\[
\lambda=\lim_{m\to\infty}\frac1m\log_2\|A_mA_{m-1}\cdots A_1\|_1,\qquad\text{a.s.}
\]
The existence of this limit is well established: it is a direct consequence of the Furstenberg--Kesten theorem \cite{FK60}, and it can also be framed within the broader contexts of Kingman's subadditive ergodic theorem \cite{Kin73} and Oseledets' multiplicative ergodic theorem \cite{Ose68}. Despite this theoretical guarantee of existence, however, computing the exponent \(\lambda\) exactly remains notoriously difficult.

\subsection{Reducing Mean Log-Degree to a Lyapunov Exponent}\label{sec:3.1}

Fix a parameter set \(\mathrm{X}=(k,H,I,Z,\alpha,\beta)\) in the sense of \cref{sec:2.1}, and write \(I=(\mathbf{c},\boldsymbol{\iota})\).

For \(x\in\{0,1\}^k\) and \(u,v\in[H]\), define
\[
N_{u,v}(x):=\left|\left\{r:\left(U_{c_{\boldsymbol{\iota}(u)},r}\right)_x\neq0,\ \clamp_H\!\left(u+\Delta_{c_{\boldsymbol{\iota}(u)},r}\right)=v\right\}\right|.
\]
By the bit-permutation symmetry of the Sergeev decompositions, \(N_{u,v}(x)\) depends only on the Hamming weight \(\|x\|_1\). Hence, for each \(i=0,\ldots,k\), we obtain a matrix \(\mathsf{A}_i\in\mathbb{Z}_{\ge0}^{H\times H}\) by setting \(\mathsf{A}_i[u,v]:=N_{u,v}(x)\) for an arbitrary \(x\in\{0,1\}^k\) with \(\|x\|_1=i\). We call \(\mathsf{A}_0,\ldots,\mathsf{A}_k\) the \emph{transition matrices} of the circuit family generated by \(\mathrm{X}\).

\begin{lemma}\label{lem:transition}
The transition matrices \(\mathsf{A}_0,\ldots,\mathsf{A}_k\) are well defined.
\end{lemma}

The proof of \cref{lem:transition} appears in \cref{sec:B.3}.

\begin{proposition}\label{prop:degree-formula}
For every \(m\ge1\), every \(h\in[H]\), and every input-gate index \(x\in\{0,1\}^{km}\), write \(x\) as \(m\) consecutive \(k\)-bit blocks \(x=(x^{(1)},\ldots,x^{(m)})\), and put \(i_t:=\|x^{(t)}\|_1\). Then
\[
L_{C_{\mathrm{X}}(km,h)}(x)=\left\|e_h^{\T}\mathsf{A}_{i_1}\cdots \mathsf{A}_{i_m}\right\|_1.
\]
\end{proposition}

\begin{proof}[Proof sketch]
Recall that \(C_{\mathrm{X}}(km,h)\) is described by a computation tree with \(m\) layers. Intuitively, the proposition follows from the following loop invariant: the \(v\)-th entry of the row vector \(e_h^{\T}\mathsf{A}_{i_1}\cdots\mathsf{A}_{i_m}\) equals the number of nodes in layer \(m\) with state \(v\) for which the product of the rank-\(1\) terms along the corresponding root-to-node path has a left vector that is nonzero at the input index \(x\). Summing these entries over all final states gives the degree of the input gate \(x\). A detailed proof appears in \cref{sec:B.3}.
\end{proof}

For \(p\in[0,1]\), let \(\Bin(k,p)\) denote the binomial distribution on \(\{0,\ldots,k\}\). Averaging \cref{prop:degree-formula} over \(X\sim\Ber(p)^{km}\) gives
\[
f_{C_{\mathrm{X}}(km,h)}(\Ber(p))=\frac1{km}\,\mathbb{E}_{I_1,\ldots,I_m\sim\Bin(k,p)}\!\left[\log_2\left\|e_h^{\T}\mathsf{A}_{I_1}\cdots \mathsf{A}_{I_m}\right\|_1\right],
\]
since the block Hamming weights \(I_1,\ldots,I_m\) are i.i.d.\ with law \(\Bin(k,p)\). Consequently, the \(m\to\infty\) limiting value of \(k\cdot f_{C_{\mathrm{X}}(km,h)}(\Ber(p))\) is at most the quenched Lyapunov exponent of the random matrix product drawn from \(\mathsf{A}_0,\ldots,\mathsf{A}_k\) with law \(\Bin(k,p)\).

For the four circuit families used in this paper, the transition matrices can be computed directly from the parameter sets; their dimensions and total numbers of nonzero entries are recorded in \cref{tab:transition}.

\begin{table}[H]
\centering
\small
\begin{tabular}{@{}lrr@{}}
\toprule
Family & Matrix dimension \(H\) & Total nonzeros in \(\mathsf{A}_0,\ldots,\mathsf{A}_7\) \\
\midrule
Jessica & \(612\) & \(13452\) \\
Sonetto & \(360\) & \(6726\) \\
Regulus & \(360\) & \(7383\) \\
\(\mathrm{Regulus}^{\T}\) & \(360\) & \(6772\) \\
\bottomrule
\end{tabular}
\caption{Dimensions and sparsity of the transition matrices of the four circuit families.}
\label{tab:transition}
\end{table}

\subsection{Toll-and-Potential Certificates and Lattice Envelopes}\label{sec:3.2}

\begin{definition}[Toll-and-potential certificates]\label{def:certificate}
A \emph{toll-and-potential certificate} for transition matrices \(\mathsf{A}_0,\ldots,\mathsf{A}_k\) is an ordered pair \((\mathfrak{a},\mathfrak{b})\) of arrays \(\mathfrak{a}=\{\mathfrak{a}_{i,j}\}_{0\le i,j\le k}\) and \(\mathfrak{b}=\{\mathfrak{b}_{i,u}\}_{0\le i\le k,\,u\in[H]}\). The pair is called an \emph{\(\eps\)-valid certificate} if, for every \(i=0,\ldots,k\) and every \(u\in[H]\),
\[
\sum_{j=0}^k\sum_{v=0}^{H-1}\mathsf{A}_j[u,v]\,2^{-\mathfrak{a}_{i,j}-\mathfrak{b}_{i,u}+\mathfrak{b}_{j,v}}\le1-\eps.
\]
\end{definition}

For \(p\in[0,1]\), write \(\Bin(k,p)_i:=\binom{k}{i}p^i(1-p)^{k-i}\), and define the \emph{binomial entropy}
\[
E_k(p):=-\sum_{i=0}^k\Bin(k,p)_i\log_2\Bin(k,p)_i,\qquad0\log_2 0:=0.
\]

\begin{definition}[Crude and refined evaluations]\label{def:evaluation}
For an array \(\mathfrak{a}=\{\mathfrak{a}_{i,j}\}_{0\le i,j\le k}\), define its \emph{crude evaluation} at \(p\) by
\[
\Phi_{\mathfrak{a}}(p):=\sum_{i=0}^k\sum_{j=0}^k\Bin(k,p)_i\Bin(k,p)_j\,\mathfrak{a}_{i,j},
\]
and its \emph{refined evaluation} at \(p\) by \(\widetilde\Phi_{\mathfrak{a}}(p):=\frac1k(\Phi_{\mathfrak{a}}(p)-E_k(p))\).

For \(0<\eta\le0.5\), define the \emph{\(\eta\)-clipped binomial entropy}
\[
\underline{E}_{k,\eta}(p):=\begin{cases}
E_k(\eta)\cdot\frac{p}{\eta},&0\le p\le\eta,\\
E_k(p),&\eta\le p\le1-\eta,\\
E_k(\eta)\cdot\frac{1-p}{\eta},&1-\eta\le p\le1,
\end{cases}
\]
and define the \emph{\(\eta\)-clipped refined evaluation} of \(\mathfrak{a}\) by \(\widetilde{\widetilde\Phi}_{\mathfrak{a},\eta}(p):=\frac1k(\Phi_{\mathfrak{a}}(p)-\underline{E}_{k,\eta}(p))\).
\end{definition}

Since \(p\mapsto E_k(p)\) is concave on \([0,1]\) \cite[\S3, Theorem 4]{SO81}, and since \(E_k(0)=E_k(1)=0\) and \(E_k(1-\eta)=E_k(\eta)\), concavity gives \(\underline{E}_{k,\eta}(p)\le E_k(p)\) for all \(p\in[0,1]\); consequently, \(\widetilde{\widetilde\Phi}_{\mathfrak{a},\eta}(p)\ge\widetilde\Phi_{\mathfrak{a}}(p)\).

\begin{theorem}[Certificate theorem]\label{thm:certificate}
Suppose that, for some \(0<\eps<1\), the pair \((\mathfrak{a},\mathfrak{b})\) is an \(\eps\)-valid toll-and-potential certificate for the transition matrices \(\mathsf{A}_0,\ldots,\mathsf{A}_k\) of the circuit family generated by \(\mathrm{X}\). Set
\[
\Gamma(\mathfrak{b}):=\max_{i,u}\mathfrak{b}_{i,u}-\min_{i,u}\mathfrak{b}_{i,u}.
\]
Then, for every \(m\ge1\), every \(h\in[H]\), and every \(p\in[0,1]\),
\[
f_{C_{\mathrm{X}}(km,h)}(\Ber(p))\le\widetilde\Phi_{\mathfrak{a}}(p)+\frac{\Gamma(\mathfrak{b})+m\log_2(1-\eps)}{km}.
\]
In particular, if \(m\ge\left\lceil\Gamma(\mathfrak{b})/(-\log_2(1-\eps))\right\rceil\), then \(f_{C_{\mathrm{X}}(km,h)}(\Ber(p))\le\widetilde\Phi_{\mathfrak{a}}(p)\le\widetilde{\widetilde\Phi}_{\mathfrak{a},\eta}(p)\) uniformly over \(h\) and \(p\), for every \(0<\eta\le0.5\).
\end{theorem}

The proof of \cref{thm:certificate} appears in \cref{sec:B.3}.

For a fixed \(p\), a smaller value of \(\Phi_{\mathfrak{a}}(p)\) yields a better upper bound. In fact, minimizing \(\Phi_{\mathfrak{a}}(p)\) subject to \((\mathfrak{a},\mathfrak{b})\) being a \(0\)-valid certificate is a convex minimization problem\footnote{The objective is linear in \(\mathfrak{a}\), and the logarithmic form of each validity constraint is a log-sum-exp inequality in affine functions of \((\mathfrak{a},\mathfrak{b})\).}. In the i.i.d.\ case, this convex program is the Lagrangian dual of the convex program of Gharavi and Anantharam \cite{GA05}\footnote{For Markovian products of nonnegative matrices, \cite{GA05} defined an auxiliary quantity \(\widehat\lambda\), proved \(\widehat\lambda\ge\lambda\), and showed that \(\widehat\lambda\) can be expressed as the optimum of a concave maximization problem over a convex polytope.}. We emphasize, however, that \cref{thm:certificate} is proved directly in \cref{sec:B.3}, and the correctness of \cite{GA05} is not needed for the correctness of this paper.

\begin{theorem}[Lattice envelopes]\label{thm:envelope}
Fix a parameter set \(\mathrm{X}\) and a finite lattice \(0=p_0<p_1<\cdots<p_L=1\). Fix \(\eta\) and \(\eps\) with \(0<\eta\le0.5\) and \(0<\eps<1\). Suppose that, for each interior lattice point \(p_r\), \(r=1,\ldots,L-1\), we are given an \(\eps\)-valid certificate \((\mathfrak{a}^{(r)},\mathfrak{b}^{(r)})\).

For notational convenience, introduce formal boundary symbols \(\mathfrak{a}^{(0)}\) and \(\mathfrak{a}^{(L)}\) by declaring
\[
\widetilde{\widetilde\Phi}_{\mathfrak{a}^{(0)},\eta}(p)=\widetilde{\widetilde\Phi}_{\mathfrak{a}^{(L)},\eta}(p):=+\infty\qquad\text{for every }p\in[0,1].
\]
Define \(\ell_{\mathrm{X}}:[0,1]\to\mathbb{R}_{\ge0}\) by
\[
\ell_{\mathrm{X}}(p):=\min\left\{\widetilde{\widetilde\Phi}_{\mathfrak{a}^{(r)},\eta}(p),\,\widetilde{\widetilde\Phi}_{\mathfrak{a}^{(r+1)},\eta}(p)\right\},\qquad p\in[p_r,p_{r+1}),\quad r=0,\ldots,L-1,
\]
where the last interval is understood to include its right endpoint \(p_L=1\). Let
\[
m_{\mathrm{X}}:=\max_{1\le r\le L-1}\left\lceil\frac{\Gamma(\mathfrak{b}^{(r)})}{-\log_2(1-\eps)}\right\rceil.
\]
Then, for every \(m\ge m_{\mathrm{X}}\), the inequality \(f_{C_{\mathrm{X}}(km,h)}(\Ber(p))\le\ell_{\mathrm{X}}(p)\) holds uniformly over \(h\) and \(p\).

Moreover, if at every interior lattice point \(p_r\), \(r=1,\ldots,L-1\), one has
\[
\widetilde{\widetilde\Phi}_{\mathfrak{a}^{(r)},\eta}(p_r)\le\min\left\{\widetilde{\widetilde\Phi}_{\mathfrak{a}^{(r-1)},\eta}(p_r),\,\widetilde{\widetilde\Phi}_{\mathfrak{a}^{(r+1)},\eta}(p_r)\right\},
\]
then \(\ell_{\mathrm{X}}\) is continuous on \([0,1]\). We refer to these inequalities as the \emph{adjacent-certificate conditions}.
\end{theorem}

The proof of \cref{thm:envelope} appears in \cref{sec:B.3}.

Finally, we define the four lattice envelopes \(\ell_{\mathrm{Jessica}}\), \(\ell_{\mathrm{Sonetto}}\), \(\ell_{\mathrm{Regulus}}\), and \(\ell_{\mathrm{Regulus}^{\T}}\) by applying \cref{thm:envelope} with \(\eta=0.01\) and \(\eps=10^{-8}\), using the family-specific lattices listed below and the certificates stored in the verification repository at \url{https://osf.io/yz4fr/overview?view_only=2f77f3a2c69e44a2bb57738e01223149}.

In the descriptions below, \(a\xrightarrow{\;\delta\;}b\) means that the lattice contains all of the points \(a,a+\delta,a+2\delta,\ldots,b\).

For Jessica, the lattice, consisting of \(L_\mathrm{Jessica}=290\) subintervals, is given by
\[
\begin{aligned}
0&\xrightarrow{\;0.1\;}0.3\xrightarrow{\;0.01\;}0.34\xrightarrow{\;0.001\;}0.36\xrightarrow{\;0.0001\;}0.3708\\
&\xrightarrow{\;0.00003\;}0.372\xrightarrow{\;0.0001\;}0.38\xrightarrow{\;0.001\;}0.4\xrightarrow{\;0.01\;}0.5\xrightarrow{\;0.1\;}1.
\end{aligned}
\]
For Sonetto, the lattice, consisting of \(L_\mathrm{Sonetto}=370\) subintervals, is given by
\[
0\xrightarrow{\;0.01\;}0.2\xrightarrow{\;0.001\;}0.5\xrightarrow{\;0.01\;}1.
\]
For Regulus, the lattice, consisting of \(L_\mathrm{Regulus}=208\) subintervals, is given by
\[
0\xrightarrow{\;0.01\;}0.2\xrightarrow{\;0.001\;}0.32\xrightarrow{\;0.01\;}1.
\]
For Regulus\(^\T\), the lattice, consisting of \(L_{\mathrm{Regulus}^{\T}}=208\) subintervals, is given by
\[
0\xrightarrow{\;0.01\;}0.32\xrightarrow{\;0.001\;}0.44\xrightarrow{\;0.01\;}1.
\]
In practice, for each circuit family, we modeled a convex program in CVXPY \cite{DB16, AVDB18}, compiled the model once, and used MOSEK's exponential-cone solver \cite{DA22} to solve it at each interior lattice point. Since the certificates themselves are stored as \texttt{.txt} files and verified directly, however, the method used to find them is not relevant to the correctness of this paper.

\subsection{Machine Verification of Our Upper Bounds}\label{sec:3.3}

This section explains how the two claims \eqref{eq:star} and \eqref{eq:starstar} used in \cref{sec:2.3.3} are verified. The verification uses three deterministic programs: the information inspector \(\mathcal{I}\), the size-inequality verifier \(\mathcal{V}_{\mathrm{size}}\), and the degree-inequality verifier \(\mathcal{V}_{\mathrm{degree}}\). We prove on paper that if these programs accept, then the two inequalities needed for \cref{cor:sigma,cor:delta} hold; we then run the programs and observe that they do accept.

\subsubsection{Lipschitz Continuity of the Lattice Envelopes}\label{sec:3.3.1}

Let \(\mathrm{X}\) be one of the four circuit families, and let \(\{(\mathfrak{a}^{(r)},\mathfrak{b}^{(r)})\}_{r=1}^{L-1}\) be the certificates defining \(\ell_{\mathrm{X}}\). Since \(k=7\) for all four parameter sets, we have \(\widetilde{\widetilde\Phi}_{\mathfrak{a}^{(r)},0.01}(p)=\frac17(\Phi_{\mathfrak{a}^{(r)}}(p)-\underline{E}_{7,0.01}(p))\), where
\[
\Phi_{\mathfrak{a}^{(r)}}(p)=\sum_{i=0}^7\sum_{j=0}^7\mathfrak{a}^{(r)}_{i,j}\binom7i\binom7j p^{i+j}(1-p)^{14-i-j}.
\]
Equivalently, in Bernstein form,
\[
\Phi_{\mathfrak{a}^{(r)}}(p)=\sum_{t=0}^{14}\overline{\mathfrak{a}}^{(r)}_t\binom{14}{t}p^t(1-p)^{14-t},\qquad
\overline{\mathfrak{a}}^{(r)}_t:=\sum_{i=\max(0,t-7)}^{\min(7,t)}\frac{\binom7i\binom7{t-i}}{\binom{14}{t}}\,\mathfrak{a}^{(r)}_{i,t-i}.
\]

\begin{proposition}[Lipschitz continuity]\label{prop:lipschitz}
For a fixed certificate, the function \(\widetilde{\widetilde\Phi}_{\mathfrak{a}^{(r)},0.01}\) is Lipschitz continuous on \([0,1]\), with Lipschitz constant
\[
\Lip^{(r)}:=\frac17\left(14\max_{0\le t<14}\bigl\lvert\overline{\mathfrak{a}}^{(r)}_{t+1}-\overline{\mathfrak{a}}^{(r)}_t\bigr\rvert+\frac{E_7(0.01)}{0.01}\right).
\]
Moreover, if the adjacent-certificate condition of \cref{thm:envelope} holds at every interior lattice point, then \(\ell_{\mathrm{X}}\) is Lipschitz continuous on \([0,1]\) with Lipschitz constant \(\Lip_{\mathrm{X}}:=\max_r\Lip^{(r)}\).
\end{proposition}

The proof of \cref{prop:lipschitz} appears in \cref{sec:B.3}.

For each \(\mathrm{X}\in\{\mathrm{Jessica},\mathrm{Sonetto},\mathrm{Regulus},\mathrm{Regulus}^{\T}\}\), the information inspector \(\mathcal{I}\) performs the following checks. First, for every \(r\), \(i\), and \(u\), it computes \(S^{(r)}_{i,u}:=\sum_{j,v}\mathsf{A}_j[u,v]\,2^{-\mathfrak{a}^{(r)}_{i,j}-\mathfrak{b}^{(r)}_{i,u}+\mathfrak{b}^{(r)}_{j,v}}\). Second, for every \(r\) and \(t\), it computes the Bernstein coefficient \(\overline{\mathfrak{a}}^{(r)}_t\). Third, at every interior lattice point, it checks the adjacent-certificate condition of \cref{thm:envelope}. If any check fails, \(\mathcal{I}\) rejects immediately.

Running \(\mathcal{I}\) yields the following global information across all four circuit families:
\[
\begin{aligned}
\max_{r,i,u}S^{(r)}_{i,u} &\qquad& 0.9999999723\\
\max_{r,i,u}\mathfrak{b}^{(r)}_{i,u} && 94.3836\\
\min_{r,i,u}\mathfrak{b}^{(r)}_{i,u} && -89.7364\\
\max_{r,t}\bigl|\overline{\mathfrak{a}}^{(r)}_{t+1}-\overline{\mathfrak{a}}^{(r)}_t\bigr| && 3.354664\\
E_7(0.01) && 0.371591\\[4pt]
\text{All checks pass?} && \text{True}
\end{aligned}
\]
All numerical values above are truncated to the shown precision, that is, rounded toward zero; in particular, each true value is smaller than the displayed value increased by one unit in its last displayed digit.

\begin{proposition}\label{prop:envelope-props}
The following two statements hold.
\begin{enumerate}
\item For each \(\mathrm{X}\), the inequality \(f_{C_{\mathrm{X}}(7m,h)}(\Ber(p))\le\ell_{\mathrm{X}}(p)\) holds for all \(m\ge10^{11}\), all \(h\in[H]\), and all \(p\in[0,1]\).
\item For each \(\mathrm{X}\), the function \(\ell_{\mathrm{X}}\) is \(13\)-Lipschitz on \([0,1]\).
\end{enumerate}
\end{proposition}

The proof of \cref{prop:envelope-props} appears in \cref{sec:B.3}.

\subsubsection{Machine Verification of Our Upper Bound on \texorpdfstring{\(\sigma(R_1)\)}{sigma(R1)}}\label{sec:3.3.2}

The size-inequality verifier \(\mathcal{V}_{\mathrm{size}}\) is a deterministic program that establishes, if it accepts, that
\[
\max_{p\in[0,1]}\{\ell_{\mathrm{Jessica}}(p)+\mathrm{h}(p)\}\le1.2448439.
\]
This is stronger than the size inequality \eqref{eq:star} of \cref{sec:2.3.3}.

The program uses the Jessica lattice \(0=p_0<p_1<\cdots<p_L=1\). It partitions each lattice interval \([p_i,p_{i+1}]\) into \(1000\) equal subintervals, with endpoints \(p_{i,j}:=p_i+(p_{i+1}-p_i)\cdot\frac{j}{1000}\) for \(0\le j\le1000\). For each subinterval \([a,b]\), it computes
\[
M_{[a,b]}:=\ell_{\mathrm{Jessica}}(b)+13(b-a)+\mathrm{h}\!\left(\min\{\max\{1/2,a\},b\}\right).
\]
The program accepts if \(M_{[a,b]}\le1.2448439\) for every such subinterval, and rejects otherwise.

\begin{theorem}[Soundness of the size-inequality verifier]\label{thm:size-verifier}
If \(\mathcal{V}_{\mathrm{size}}\) accepts, then the size inequality \eqref{eq:star} of \cref{sec:2.3.3} holds.
\end{theorem}

The proof of \cref{thm:size-verifier} appears in \cref{sec:B.3}.

We ran \(\mathcal{V}_{\mathrm{size}}\) and observed that it accepts. Hence, by \cref{thm:size-verifier}, the inequality \eqref{eq:star} holds.

\Cref{fig:envelopes} shows the four lattice envelopes; it also plots \(\ell_{\mathrm{Jessica}}(p)+\mathrm{h}(p)-1\) against the target value \(1.2448439-1=0.2448439\).

\paperfigure{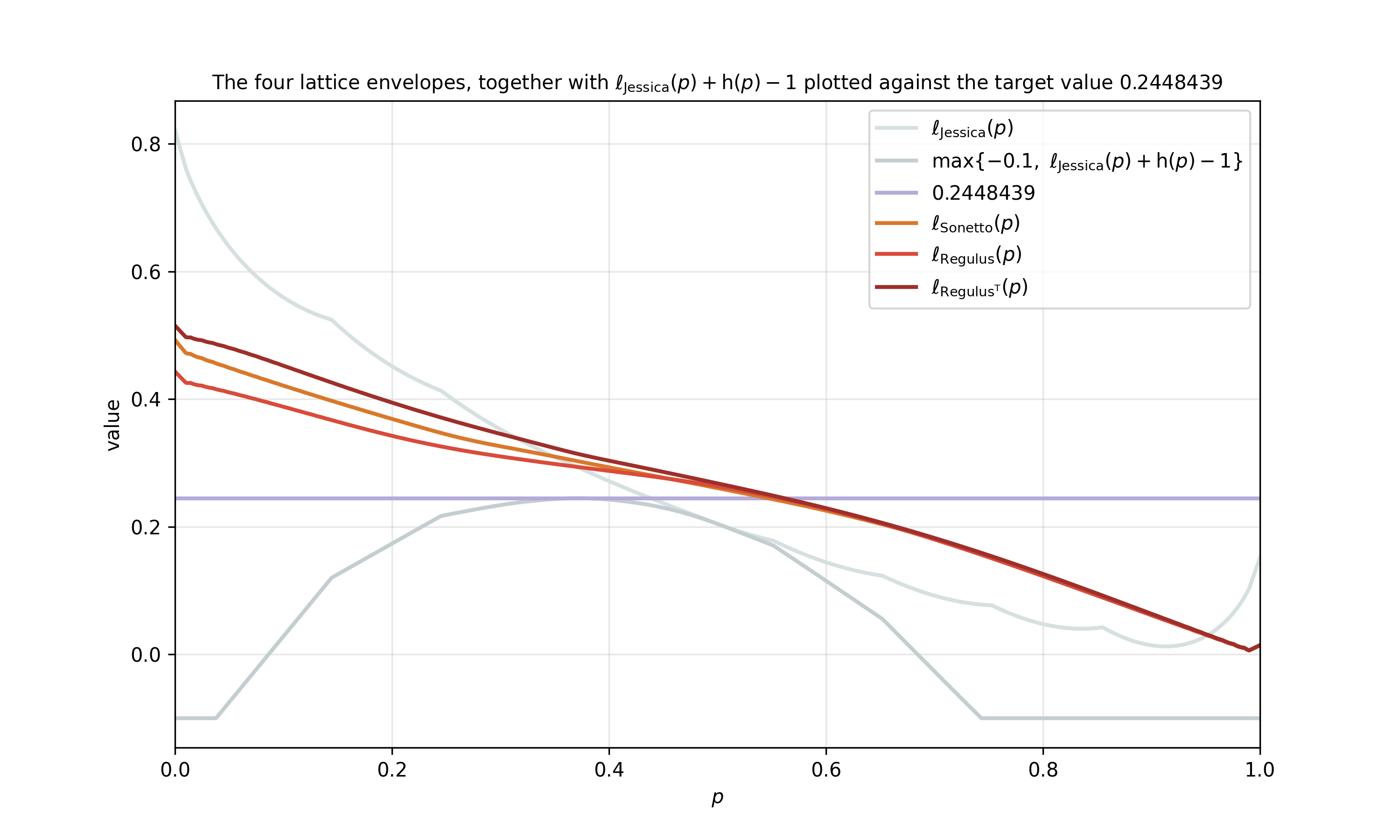}{The four lattice envelopes, together with \(\ell_{\mathrm{Jessica}}(p)+\mathrm{h}(p)-1\) plotted against the target value \(0.2448439\).}{0.8\textheight}{fig:envelopes}

\clearpage

\subsubsection{Machine Verification of Our Upper Bound on \texorpdfstring{\(\delta(R_1)\)}{delta(R1)}}\label{sec:3.3.3}

The degree-inequality verifier \(\mathcal{V}_{\mathrm{degree}}\) is a deterministic program that establishes, if it accepts, that
\[
\mathbf{z}(p,q)\le0.3198659\qquad\text{for all }p,q\in[0,1].
\]
This is stronger than the degree inequality \eqref{eq:starstar} of \cref{sec:2.3.3}.

The program defines a Boolean test on a rectangle \(R=[p_{\min},p_{\max}]\times[q_{\min},q_{\max}]\subseteq[0,1]^2\). At recursion depth \(d\), the test proceeds as follows.
\begin{enumerate}
\item Sample the \(101\times101\) grid \(\left(\frac{(100-i)p_{\min}+ip_{\max}}{100},\frac{(100-j)q_{\min}+jq_{\max}}{100}\right)\), for \(0\le i,j\le100\).
\item At each grid point \((p_i,q_j)\), form the eight points \(Q_{\mathrm{id}}(p_i,q_j)=(x_{\mathrm{id}},y_{\mathrm{id}})\) of \cref{sec:2.3.3}. Initialize \(Z_{i,j}:=\min_{\mathrm{id}}\max\{x_{\mathrm{id}},y_{\mathrm{id}}\}\). For every pair \(0\leq\mathrm{hi}<\mathrm{lo}\leq7\), if \((x_{\mathrm{hi}}-y_{\mathrm{hi}})(x_{\mathrm{lo}}-y_{\mathrm{lo}})<0\) and \((x_{\mathrm{hi}}-x_{\mathrm{lo}})(y_{\mathrm{hi}}-y_{\mathrm{lo}})<0\), set \(t_*:=\frac{x_{\mathrm{hi}}-y_{\mathrm{hi}}}{(x_{\mathrm{hi}}-y_{\mathrm{hi}})-(x_{\mathrm{lo}}-y_{\mathrm{lo}})}\) and update \(Z_{i,j}\leftarrow\min\{Z_{i,j},\,(1-t_*)x_{\mathrm{hi}}+t_*x_{\mathrm{lo}}\}\).
\item Return true if \(Z_{i,j}+13\max\left\{\frac{p_{\max}-p_{\min}}{200},\frac{q_{\max}-q_{\min}}{200}\right\}\le0.3198659\) at every grid point.
\item Otherwise, if \(d=30\), return false. If \(d\le29\), split \(R\) in half along the \(p\)-axis when \(d\) is even and along the \(q\)-axis when \(d\) is odd, apply the same test to both subrectangles at depth \(d+1\), and return true exactly when both recursive calls return true.
\end{enumerate}
The program \(\mathcal{V}_{\mathrm{degree}}\) accepts if this test returns true on \([0,1]^2\) at recursion depth \(0\).

\begin{theorem}[Soundness of the degree-inequality verifier]\label{thm:degree-verifier}
If \(\mathcal{V}_{\mathrm{degree}}\) accepts, then the degree inequality \eqref{eq:starstar} of \cref{sec:2.3.3} holds.
\end{theorem}

The proof of \cref{thm:degree-verifier} appears in \cref{sec:B.3}.

We ran \(\mathcal{V}_{\mathrm{degree}}\) and observed that it accepts. Hence, by \cref{thm:degree-verifier}, the inequality \eqref{eq:starstar} holds.

In this run, the verifier makes \(2981\) recursive calls, checks \(30409181\) grid points, and reaches maximum recursion depth \(d=24\). Among all grid points visited, the largest value of \(\mathbf{z}\) found is \(\mathbf{z}(0.3194995117,0.3202343750)\approx0.3198453660\). At this point, the best pair of coordinate points is \((Q_2,Q_5)\), corresponding to the two circuits \(C_2\) and \(C_{\mathrm{S}}\).

On a single CPU core of a personal computer, using Python with NumPy, the size-inequality verifier \(\mathcal{V}_{\mathrm{size}}\) of \cref{sec:3.3.2} finishes in about two seconds, while the degree-inequality verifier \(\mathcal{V}_{\mathrm{degree}}\) of \cref{sec:3.3.3} finishes in about half a minute. Both runs are fully deterministic: rerunning them reproduces exactly the same statistics and the same reported maximum.

\clearpage

\Cref{fig:z-full,fig:z-zoom} show the degree landscape \(\mathbf{z}(p,q)\) on \([0,1]^2\) and on \([0.25,0.4]^2\), respectively.

\paperfigure{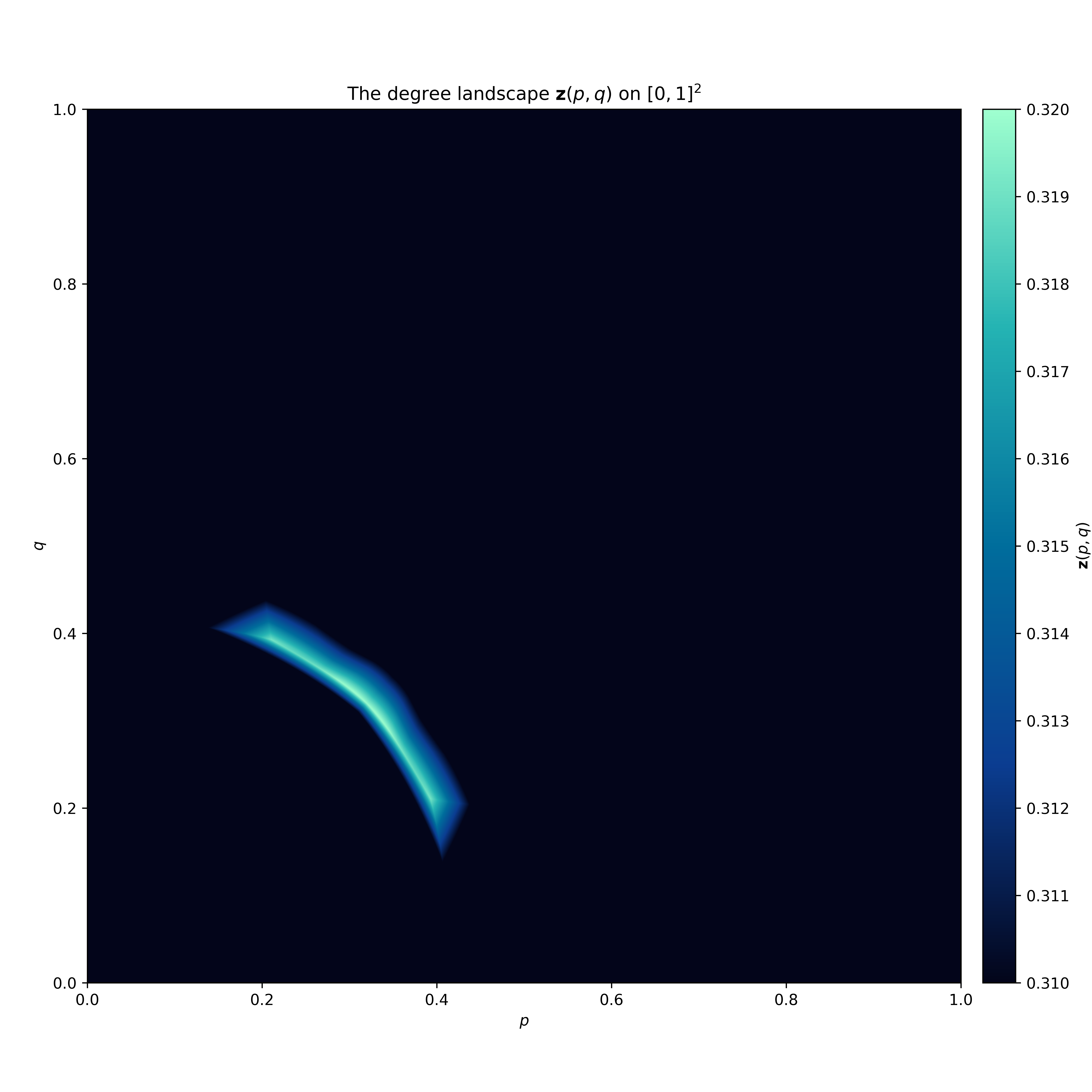}{The degree landscape \(\mathbf{z}(p,q)\) on \([0,1]^2\).}{0.4\textheight}{fig:z-full}

\paperfigure{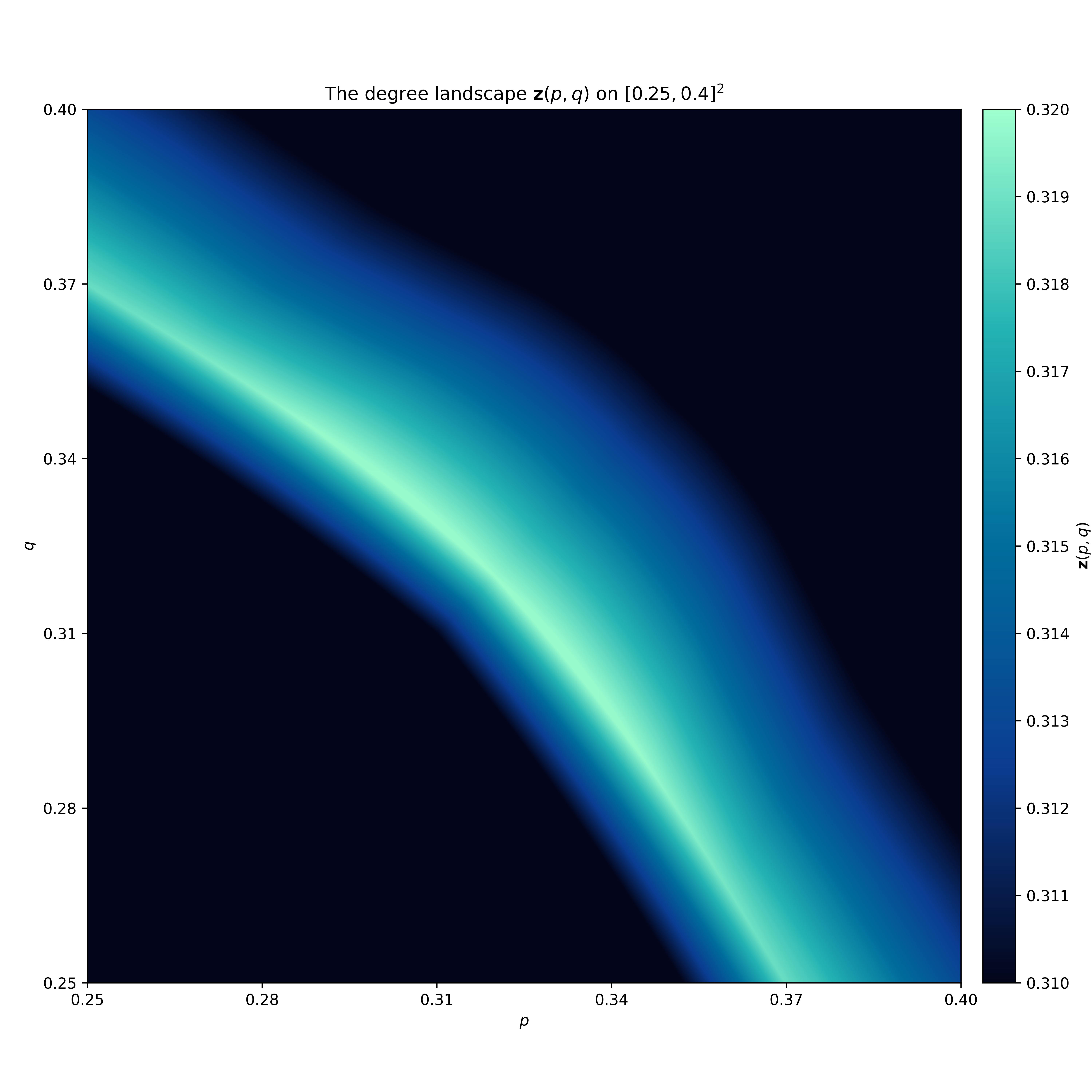}{The degree landscape \(\mathbf{z}(p,q)\) on \([0.25,0.4]^2\).}{0.4\textheight}{fig:z-zoom}

\clearpage

\Cref{fig:strategy-full,fig:strategy-zoom} show which pair of coordinate points attains \(\mathbf{z}(p,q)\) on \([0,1]^2\) and on \([0.25,0.4]^2\), respectively.

\paperfigure{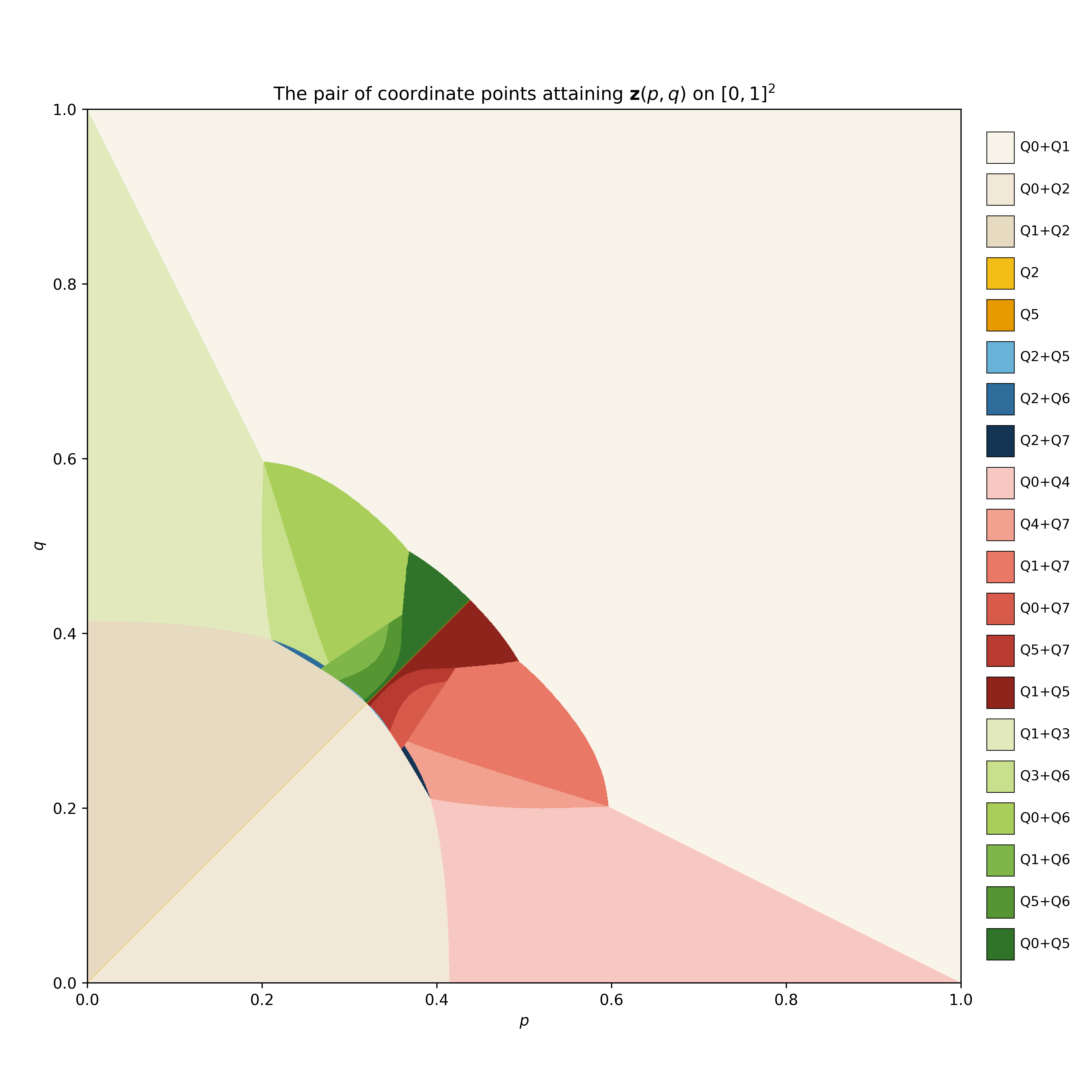}{The pair of coordinate points attaining \(\mathbf{z}(p,q)\) on \([0,1]^2\).}{0.4\textheight}{fig:strategy-full}

\paperfigure{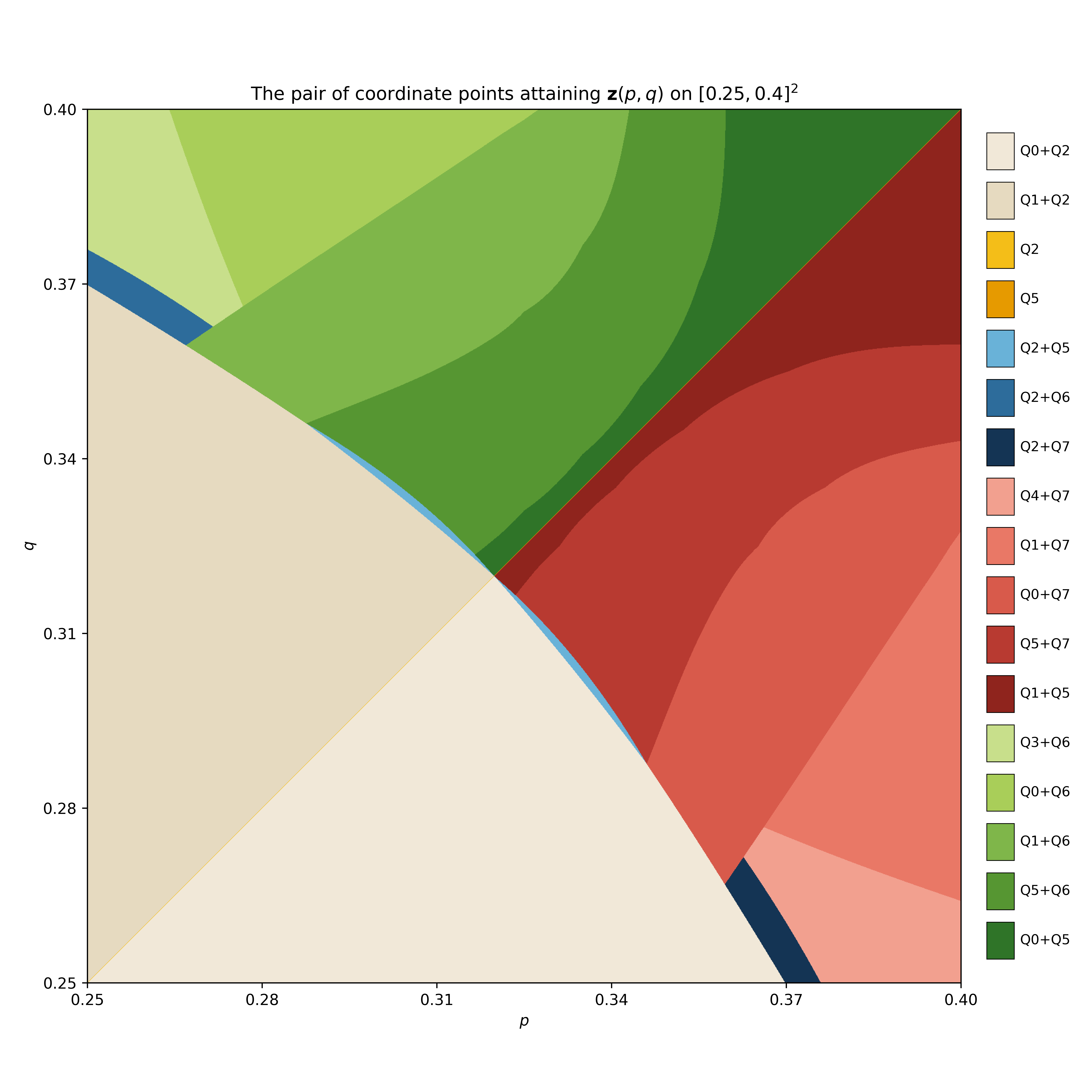}{The pair of coordinate points attaining \(\mathbf{z}(p,q)\) on \([0.25,0.4]^2\).}{0.4\textheight}{fig:strategy-zoom}

\appendix

\section{Preliminaries}\label{sec:A}

This appendix records the circuit model, the Sergeev decompositions, and the computation-tree language used throughout the paper.

\subsection{Circuit Model and Cost Measures}\label{sec:A.1}

In this paper, a \emph{depth-2 linear circuit} is an ordered pair \(C=(A,B)\) of matrices \(A\in\Mat_{n\times r}\) and \(B\in\Mat_{r\times m}\), and we say that \(C\) computes the \(n\times m\) matrix \(M=AB\). This definition corresponds to depth-2 linear circuits in the following sense:
\begin{itemize}
\item consider a circuit with three layers, an input layer, a middle layer, and an output layer, having \(n\), \(r\), and \(m\) gates, respectively;
\item each middle gate computes a fixed linear combination of the input gates;
\item each output gate computes a fixed linear combination of the middle gates;
\item the matrices \(A\) and \(B\) record the weights from the input gates to the middle gates and from the middle gates to the output gates, respectively;
\item since the input-output behavior of the circuit is exactly the linear transformation represented by \(M=AB\), saying that \(C\) computes \(M\) is justified.
\end{itemize}
In this convention, inputs are row vectors: an input \(x\) is transformed as \(x\mapsto xA\mapsto xAB\).

For a depth-2 linear circuit \(C=(A,B)\), its size is
\[
\mathrm{S}(C):=\nnz(A)+\nnz(B),
\]
and its maximum input-output degree is
\[
\mathrm{D}(C):=\max\left\{\max_x\nnz(A[x,*]),\,\max_y\nnz(B[*,y])\right\}.
\]
These are the conventions of \cite{JS13} and \cite[Section 3.2]{AL25}.

Let \(\mathcal{R}\) be a commutative semiring containing \(0\) and \(1\), and let \(\Gamma\subseteq\mathcal{R}\) be a multiplicatively closed subset containing \(0\) and \(1\). Let \(\mathsf{Circ}_{\mathcal{R},\Gamma}^{(2)}\) denote the set of all depth-2 linear circuits whose weights belong to \(\Gamma\).

For a matrix \(M\) with entries in \(\mathcal{R}\), define
\[
\mathrm{S}_{\mathcal{R},\Gamma}(M):=\min_{\substack{C=(A,B)\in\mathsf{Circ}_{\mathcal{R},\Gamma}^{(2)}\\ AB=M}}\mathrm{S}(C),
\qquad
\mathrm{D}_{\mathcal{R},\Gamma}(M):=\min_{\substack{C=(A,B)\in\mathsf{Circ}_{\mathcal{R},\Gamma}^{(2)}\\ AB=M}}\mathrm{D}(C),
\]
with the convention that each minimum is \(+\infty\) if no such circuit exists. Also define
\[
\sigma_{\mathcal{R},\Gamma}(M):=\lim_{n\to\infty}\mathrm{S}_{\mathcal{R},\Gamma}(M^{\otimes n})^{1/n},
\qquad
\delta_{\mathcal{R},\Gamma}(M):=\lim_{n\to\infty}\mathrm{D}_{\mathcal{R},\Gamma}(M^{\otimes n})^{1/n},
\]
whenever these limits exist.

By \cite[Theorem 3.1]{AL25}, if \(\mathbb{F}\) and \(\mathbb{K}\) are fields with \(\mathbb{K}\) a field extension of \(\mathbb{F}\), then
\[
\sigma_{\mathbb{F},\mathbb{F}}(M)=\sigma_{\mathbb{K},\mathbb{K}}(M),
\qquad
\delta_{\mathbb{F},\mathbb{F}}(M)=\delta_{\mathbb{K},\mathbb{K}}(M).
\]

In this paper, we use the restrictive model \(\mathcal{R}=\mathbb{Z}\) and \(\Gamma=\{0,\pm1\}\). Since our constructions use only coefficients in \(\{0,\pm1\}\), the corresponding upper bounds on
\[
\sigma_{\mathbb{F},\mathbb{F}}(R_1),\qquad
\delta_{\mathbb{F},\mathbb{F}}(R_1),\qquad
\sigma_{\mathbb{F},\mathbb{F}}(P_1^{(c)}),\qquad
\delta_{\mathbb{F},\mathbb{F}}(P_1^{(c)})
\]
also hold, with the same numerical constants, over an arbitrary field \(\mathbb{F}\): the elements \(0\), \(1\), and \(-1\) can be identified in \(\mathbb{F}\), although possibly \(1_{\mathbb{F}}=(-1)_{\mathbb{F}}\) in characteristic \(2\).

We define an equivalence relation \(\cong\) corresponding to relabeling the middle gates. Formally, two circuits \((A_0,B_0)\) and \((A_1,B_1)\) with \(r\) middle gates each are \emph{equivalent}, written
\[
(A_0,B_0)\cong(A_1,B_1),
\]
if and only if there exists an \(r\times r\) permutation matrix \(\Pi\) such that
\[
A_1=A_0\Pi,\qquad B_1=\Pi^{-1}B_0.
\]
It is straightforward to check that \(\cong\) is an equivalence relation. Let \(\widetilde{\mathsf{Circ}}_{\mathcal{R},\Gamma}^{(2)}\) denote the set of equivalence classes of \(\mathsf{Circ}_{\mathcal{R},\Gamma}^{(2)}\) under \(\cong\). From now on, we work in \(\widetilde{\mathsf{Circ}}_{\mathcal{R},\Gamma}^{(2)}\) unless stated otherwise.

We define two binary operations on depth-2 linear circuits: concatenation sum and Kronecker product.

\paragraph{Concatenation sum.}
If \(C=(A,B)\) and \(C'=(A',B')\), where \(A\) and \(A'\) have the same number of rows and \(B\) and \(B'\) have the same number of columns, then their \emph{concatenation sum} \(C+C'\) is defined as
\[
C+C':=\left(\begin{bmatrix}A&A'\end{bmatrix},\begin{bmatrix}B\\B'\end{bmatrix}\right).
\]
If \(C\) computes \(M\) and \(C'\) computes \(M'\), then \(C+C'\) computes \(M+M'\), since
\[
\begin{bmatrix}A&A'\end{bmatrix}\begin{bmatrix}B\\B'\end{bmatrix}=AB+A'B'=M+M'.
\]

\paragraph{Kronecker product.}
If \(C=(A,B)\) and \(C'=(A',B')\), then their \emph{Kronecker product} \(C\otimes C'\) is defined as
\[
C\otimes C':=(A\otimes A',\,B\otimes B').
\]
If \(C\) computes \(M\) and \(C'\) computes \(M'\), then \(C\otimes C'\) computes \(M\otimes M'\), since, by the mixed-product property of the Kronecker product,
\[
(A\otimes A')(B\otimes B')=(AB)\otimes(A'B')=M\otimes M'.
\]

Modulo the middle-gate equivalence relation \(\cong\), the concatenation sum is associative and commutative. The Kronecker product on circuits is associative but not commutative. Finally, \(\otimes\) distributes over \(+\):
\[
C\otimes(C'+C'')\cong(C\otimes C')+(C\otimes C''),
\]
since
\[
\left(A\otimes\begin{bmatrix}A'&A''\end{bmatrix},\,B\otimes\begin{bmatrix}B'\\B''\end{bmatrix}\right)\cong\left(\begin{bmatrix}A\otimes A'&A\otimes A''\end{bmatrix},\begin{bmatrix}B\otimes B'\\B\otimes B''\end{bmatrix}\right).
\]

We define one unary operation on depth-2 linear circuits: transposition.

\paragraph{Transposition.}
If \(C=(A,B)\), then its \emph{transposition} \(C^\T\) is defined as
\[
C^\T:=(B^\T,A^\T).
\]
If \(C\) computes \(M\), then \(C^\T\) computes \(M^\T\), since
\[
B^\T A^\T=(AB)^\T=M^\T.
\]

Each middle gate of a depth-2 linear circuit contributes one rank-\(1\) term: if the middle layer has \(r\) gates, then
\[
AB=\sum_{\ell=0}^{r-1}A[*,\ell]B[\ell,*].
\]
We call \((A[*,\ell],B[\ell,*])\) the \(\ell\)-th \emph{rank-\(1\) term} of the circuit. Depending on the context, we also regard it as a circuit with one middle gate, and we call \(A[*,\ell]B[\ell,*]\) the corresponding \emph{rank-\(1\) matrix}. Thus the circuit \((A,B)\) and the concatenation sum of its \(r\) rank-\(1\) terms are two equivalent viewpoints on the same object:
\[
(A,B)=\sum_{\ell=0}^{r-1}\,(A[*,\ell],B[\ell,*]).
\]
We switch between these viewpoints when discussing Sergeev decompositions and computation trees.

From now on, we fix \(\mathcal{R}\) and \(\Gamma\) and suppress them from the notation, writing \(\mathrm{S}(M)\), \(\mathrm{D}(M)\), \(\sigma(M)\), and \(\delta(M)\) for \(\mathrm{S}_{\mathcal{R},\Gamma}(M)\), \(\mathrm{D}_{\mathcal{R},\Gamma}(M)\), \(\sigma_{\mathcal{R},\Gamma}(M)\), and \(\delta_{\mathcal{R},\Gamma}(M)\), respectively. Furthermore, we restrict attention to the case in which \(\mathcal{R}\) has no zero divisors and \(M\) admits at least one depth-2 linear circuit over \((\mathcal{R},\Gamma)\); that is, there exists \(C\in\mathsf{Circ}_{\mathcal{R},\Gamma}^{(2)}\) computing \(M^{\otimes1}=M\). The matrices considered in this paper, namely \(R_1\) and \(P_1^{(c)}\), with \(\mathcal{R}=\mathbb{Z}\) and \(\Gamma=\{0,\pm1\}\), satisfy these assumptions.

Under these assumptions, both the size and the maximum input-output degree are submultiplicative under Kronecker products of matrices.

\begin{manuallemma}{A.1}\label{lem:A1}
Under the assumptions above, for every matrix \(M\) and all integers \(a,b\geq1\),
\[
\mathrm{S}(M^{\otimes(a+b)})\le\mathrm{S}(M^{\otimes a})\,\mathrm{S}(M^{\otimes b})
\qquad\text{and}\qquad
\mathrm{D}(M^{\otimes(a+b)})\le\mathrm{D}(M^{\otimes a})\,\mathrm{D}(M^{\otimes b}).
\]
\end{manuallemma}

\begin{proof}
For size, let \((A_a,B_a)\) and \((A_b,B_b)\) be size-optimal circuits for \(M^{\otimes a}\) and \(M^{\otimes b}\), respectively. Their Kronecker product \((A_a\otimes A_b,B_a\otimes B_b)\) computes \(M^{\otimes(a+b)}\), and since \(\mathcal{R}\) has no zero divisors, \(\nnz\) is multiplicative under Kronecker products. Therefore
\[
\begin{aligned}
\mathrm{S}(M^{\otimes(a+b)})
&\leq\mathrm{S}((A_a\otimes A_b,B_a\otimes B_b))\\
&=\nnz(A_a\otimes A_b)+\nnz(B_a\otimes B_b)\\
&=\nnz(A_a)\nnz(A_b)+\nnz(B_a)\nnz(B_b)\\
&\le(\nnz(A_a)+\nnz(B_a))(\nnz(A_b)+\nnz(B_b))\\
&=\mathrm{S}((A_a,B_a))\,\mathrm{S}((A_b,B_b))\\
&=\mathrm{S}(M^{\otimes a})\,\mathrm{S}(M^{\otimes b}).
\end{aligned}
\]
For the maximum input-output degree, let \((A_a,B_a)\) and \((A_b,B_b)\) be degree-optimal circuits for \(M^{\otimes a}\) and \(M^{\otimes b}\), respectively. Their Kronecker product \((A_a\otimes A_b,B_a\otimes B_b)\) computes \(M^{\otimes(a+b)}\). Therefore
\[
\begin{aligned}
\mathrm{D}(M^{\otimes(a+b)})
&\leq\mathrm{D}((A_a\otimes A_b,B_a\otimes B_b))\\
&=\max\Bigl\{\max_{(i,i')}\nnz((A_a\otimes A_b)[(i,i'),*]),\ \max_{(j,j')}\nnz((B_a\otimes B_b)[*,(j,j')])\Bigr\}\\
&=\max\Bigl\{\max_{i,i'}\nnz(A_a[i,*])\nnz(A_b[i',*]),\ \max_{j,j'}\nnz(B_a[*,j])\nnz(B_b[*,j'])\Bigr\}\\
&\leq\mathrm{D}((A_a,B_a))\,\mathrm{D}((A_b,B_b))\\
&=\mathrm{D}(M^{\otimes a})\,\mathrm{D}(M^{\otimes b}).
\end{aligned}
\]
This proves the claimed inequalities.
\end{proof}

If \(M\) is the zero matrix, we set \(\sigma(M)=\delta(M)=0\).

Now assume that \(M\) is nonzero. Then \(M^{\otimes n}\) is nonzero for every \(n\ge1\), so \(\mathrm{S}(M^{\otimes n})\) and \(\mathrm{D}(M^{\otimes n})\) are positive integers. By Lemma~A.1, the sequences
\[
\log_2\mathrm{S}(M^{\otimes n})\qquad\text{and}\qquad\log_2\mathrm{D}(M^{\otimes n})
\]
are subadditive, since the logarithm turns the multiplicative inequalities of Lemma~A.1 into additive ones. By Fekete's lemma \cite{Fek23},
\[
\lim_{n\to\infty}\frac1n\log_2\mathrm{S}(M^{\otimes n})=\inf_{n\ge1}\frac1n\log_2\mathrm{S}(M^{\otimes n})
\qquad\text{and}\qquad
\lim_{n\to\infty}\frac1n\log_2\mathrm{D}(M^{\otimes n})=\inf_{n\ge1}\frac1n\log_2\mathrm{D}(M^{\otimes n}).
\]
Equivalently, the limits
\[
\sigma(M):=\lim_{n\to\infty}\mathrm{S}(M^{\otimes n})^{1/n},
\qquad
\delta(M):=\lim_{n\to\infty}\mathrm{D}(M^{\otimes n})^{1/n}
\]
exist. For every nonzero matrix \(M\), we have \(\sigma(M),\delta(M)\ge1\).

\subsection{Sergeev Decompositions in Detail}\label{sec:A.2}

Sergeev decompositions (\cite{Ser22}; \cite[Section 5.1]{AL25}) are useful rank-\(1\) decompositions of the disjointness matrices; in this paper, they serve as the building blocks of the computation trees defined in \cref{sec:2.1}.

In this section, for convenience, we index the rows and columns of \(R_k=(R_1)^{\otimes k}\) by subsets of \([k]\), so that
\[
R_k(S,T)=\one[S\cap T=\varnothing].
\]

Fix \(k\ge1\) and an identifier \(c\in\{0,\ldots,2^k-1\}\). Define the word
\[
w(c):=(w_0,\ldots,w_{k-1},\texttt{R}),\qquad
w_i:=\begin{cases}
\texttt{R},&\text{the bit of }c\text{ at weight }2^{k-1-i}\text{ is }0,\\
\texttt{C},&\text{the bit of }c\text{ at weight }2^{k-1-i}\text{ is }1.
\end{cases}
\]
For instance, if \(k=7\) and \(c=43=(0101011)_2\), then \(w(c)=\texttt{RCRCRCCR}\). We refer to the final symbol \(\texttt{R}\) of \(w(c)\) as the \emph{sentinel}.

We define \(S_k(c)\) as the depth-2 linear circuit obtained by applying the following algorithm to the word \(w(c)\).
\begin{enumerate}
\item Initialize counters \(a=0\) and \(b=0\), and initialize an empty list of rank-\(1\) terms.
\item Read the word \(w(c)\) from left to right. For each symbol:
\begin{itemize}
\item if the symbol is \(\texttt{R}\), perform an \((\texttt{R},a,b)\) step, and then replace \(a\) by \(a+1\);
\item if the symbol is \(\texttt{C}\), perform a \((\texttt{C},a,b)\) step, and then replace \(b\) by \(b+1\).
\end{itemize}
\item The algorithm uses only the following two kinds of steps.
\begin{itemize}
\item An \((\texttt{R},a,b)\) step appends, for every subset \(S\subseteq[k]\) with \(|S|=a\), the rank-\(1\) term \((U,V^\T)\) to the list, where \(U\) is the indicator vector of the row \(S\), and \(V^\T\) is the indicator row vector of the set of columns
\[
\{T\subseteq[k]\setminus S:\ |T|\ge b\}.
\]
\item A \((\texttt{C},a,b)\) step appends, for every subset \(T\subseteq[k]\) with \(|T|=b\), the rank-\(1\) term \((U,V^\T)\) to the list, where \(V^\T\) is the indicator row vector of the column \(T\), and \(U\) is the indicator vector of the set of rows
\[
\{S\subseteq[k]\setminus T:\ |S|\ge a\}.
\]
\end{itemize}
\end{enumerate}
Finally, \(S_k(c)\) is defined as the concatenation sum of the \(2^k\) appended rank-\(1\) terms.

\begin{manualfact}{A.2.1}\label{fact:A21}
Every \(S_k(c)\) is a depth-2 linear circuit computing \(R_k\), with exactly \(2^k\) middle gates.
\end{manualfact}

\begin{proof}
Let \(k_0\) and \(k_1\) denote the numbers of non-sentinel \(\texttt{R}\) and \(\texttt{C}\) symbols in \(w(c)\), respectively, so that \(k_0+k_1=k\). The \(\texttt{R}\) steps have first counter \(a=0,\ldots,k_0\) (including the sentinel), and the \(\texttt{C}\) steps have second counter \(b=0,\ldots,k_1-1\). Therefore the number of appended rank-\(1\) terms is
\[
\sum_{i=0}^{k_0}\binom ki+\sum_{j=0}^{k_1-1}\binom kj=\sum_{i=0}^{k_0}\binom ki+\sum_{i=k_0+1}^{k}\binom ki=2^k,
\]
where the middle equality uses the symmetry \(\binom kj=\binom k{k-j}\).

Fix a pair \((S,T)\) with \(S,T\subseteq[k]\), and write \(s:=|S|\) and \(t:=|T|\).

If \(S\cap T\neq\varnothing\), then no appended rank-\(1\) term contributes to the entry \((S,T)\), since in every appended term each row set is disjoint from each column set. Hence the sum is \(0\) at this entry.

If \(S\cap T=\varnothing\), then a rank-\(1\) term containing the entry \((S,T)\) can be appended only in one of the following two types of steps:
\begin{enumerate}
\item an \((\texttt{R},a,b)\) step with \(a=s\) and \(b\leq t\);
\item a \((\texttt{C},a,b)\) step with \(a\leq s\) and \(b=t\).
\end{enumerate}
Indeed, for an \((\texttt{R},a,b)\) step to append a rank-\(1\) term containing \((S,T)\), the term must be the one indexed by the row \(S\), so \(a=s\) and \(b\le t\); for a \((\texttt{C},a,b)\) step, the term must be the one indexed by the column \(T\), so \(a\le s\) and \(b=t\).

Consider the monotone lattice path traced by the counters \((a,b)\) while scanning \(w(c)\): an \(\texttt{R}\) step moves right and a \(\texttt{C}\) step moves up. The path starts at \((0,0)\) and ends at \((k_0+1,k_1)\). Since \(S\cap T=\varnothing\), we have \(s+t\le k\), so the endpoint lies outside the rectangle
\[
\{0,\ldots,s\}\times\{0,\ldots,t\}.
\]
Since both counters only increase, the path leaves this rectangle exactly once. If it leaves by an \(\texttt{R}\) step, then the step has type \((\texttt{R},s,b)\) with \(b\le t\), and the term indexed by the row \(S\) contains \((S,T)\). If it leaves by a \(\texttt{C}\) step, then the step has type \((\texttt{C},a,t)\) with \(a\le s\), and the term indexed by the column \(T\) contains \((S,T)\). These are exactly the possible contributing steps, so exactly one appended rank-\(1\) term contributes to \((S,T)\). Hence the sum is \(1\) at this entry.
\end{proof}

In fact, the transpose of a Sergeev circuit is the Sergeev circuit of the complemented identifier. Set
\[
\overline{c}:=2^k-1-c.
\]
Since \(\overline{c}\) flips every bit of \(c\), one obtains \(w(\overline{c})\) from \(w(c)\) by swapping every non-sentinel \(\texttt{R}\) and \(\texttt{C}\).

\begin{manualfact}{A.2.2}\label{fact:A22}
For every \(k\ge1\) and every \(c\in\{0,\ldots,2^k-1\}\),
\[
S_k(c)^\T\cong S_k(\overline{c}).
\]
\end{manualfact}

\begin{proof}
We compare the appended rank-\(1\) terms. In the non-sentinel prefix, passing from \(c\) to \(\overline{c}\) exchanges \(\texttt{R}\) steps and \(\texttt{C}\) steps. Consequently, if a non-sentinel step of \(S_k(c)\) has counters \((a,b)\), then the corresponding step of \(S_k(\overline{c})\) has counters \((b,a)\).

An \((\texttt{R},a,b)\) step of \(S_k(c)\) appends, for each \(|S|=a\), the term
\[
\left(\one_{\{S\}},\ \one_{\{T\subseteq[k]\setminus S:\ |T|\ge b\}}\right).
\]
Its transpose is exactly the term appended by the corresponding \((\texttt{C},b,a)\) step of \(S_k(\overline{c})\). Similarly, every non-sentinel \((\texttt{C},a,b)\) step of \(S_k(c)\) transposes to the corresponding \((\texttt{R},b,a)\) step of \(S_k(\overline{c})\).

The only remaining step is the final sentinel \(\texttt{R}\). If the non-sentinel prefix of \(w(c)\) contains \(k_0\) symbols \(\texttt{R}\) and \(k_1\) symbols \(\texttt{C}\), then this sentinel appends the terms
\[
\left(\one_{\{S\}},\ \one_{\{[k]\setminus S\}}\right)\qquad (|S|=k_0).
\]
In \(S_k(\overline{c})\), the sentinel counters are \((k_1,k_0)\), so its sentinel terms are
\[
\left(\one_{\{T\}},\ \one_{\{[k]\setminus T\}}\right)\qquad (|T|=k_1).
\]
Taking \(T=[k]\setminus S\) gives exactly the transposed sentinel terms. Hence the two concatenation sums agree after transposition, up to relabeling middle gates.
\end{proof}

\subsection{From Computation Trees to Circuits}\label{sec:A.3}

In this section, we define computation trees \cite[Appendix B]{AL25} and use them to justify the recursively defined circuit families of \cref{sec:2.1}.

Fix a base matrix \(M\). A \emph{computation tree} for \(M^{\otimes N}\) is a finite rooted tree with the following data:
\begin{enumerate}
\item each internal node \(x\) has attached to it a finite depth-2 linear circuit \(C_x\) computing \(M^{\otimes n_x}\), where \(n_x\ge1\);
\item the children of \(x\) are in bijection with the rank-\(1\) terms of \(C_x\);
\item if a child corresponds to the rank-\(1\) term \(r\), then the edge to that child has \(r\) attached to it and has weight \(n_x\);
\item every root-to-leaf path has total edge weight \(N\).
\end{enumerate}
The \emph{weighted depth} of a node \(x\), denoted \(\operatorname{depth}(x)\), is the total edge weight on the path from the root to \(x\). The \emph{remaining weighted depth} at \(x\) is
\[
N_x:=N-\operatorname{depth}(x).
\]
The leaves are therefore exactly the nodes with \(N_x=0\).

If \(r=(U,V^\T)\) is a rank-\(1\) term, write \(\Mat(r):=UV^\T\). Define a depth-2 linear circuit \(\mathcal{C}(x)\) for each node \(x\) as follows. If \(x\) is a leaf, let
\[
\mathcal{C}(x):=([1],[1]),
\]
the trivial circuit computing \(M^{\otimes0}\). If \(x\) is an internal node, let
\[
\mathcal{C}(x):=\sum_{y\in\Ch(x)}r_{x,y}\otimes\mathcal{C}(y),
\]
where \(r_{x,y}\) is the rank-\(1\) term attached to the edge from \(x\) to \(y\).

By the distributivity of \(\otimes\) over \(+\), the circuit \(\mathcal{C}(x)\) is the concatenation sum, over all leaves below \(x\), of the Kronecker products of the rank-\(1\) terms attached to the edges along the corresponding paths.

\begin{manuallemma}{A.3.1}\label{lem:A31}
For every node \(x\) of a computation tree, the subtree circuit \(\mathcal{C}(x)\) computes \(M^{\otimes N_x}\). In particular, if \(\rho\) is the root, then \(\mathcal{C}(\rho)\) computes \(M^{\otimes N}\).
\end{manuallemma}

\begin{proof}
We prove the claim by induction from the leaves upward. If \(x\) is a leaf, then \(N_x=0\), and \(\mathcal{C}(x)=([1],[1])\) computes \(M^{\otimes0}\).

Now let \(x\) be an internal node and assume the claim for its children. Since the children of \(x\) are indexed by the rank-\(1\) terms of \(C_x\), we may write
\[
C_x=\sum_{y\in\Ch(x)}r_{x,y},
\]
and, because \(C_x\) computes \(M^{\otimes n_x}\),
\[
\sum_{y\in\Ch(x)}\Mat(r_{x,y})=M^{\otimes n_x}.
\]
For every child \(y\), the edge from \(x\) to \(y\) has weight \(n_x\), so \(N_y=N_x-n_x\). By the inductive hypothesis, \(\mathcal{C}(y)\) computes \(M^{\otimes(N_x-n_x)}\). Therefore \(\mathcal{C}(x)\) computes
\[
\begin{aligned}
\sum_{y\in\Ch(x)}\Mat(r_{x,y})\otimes M^{\otimes(N_x-n_x)}
&=\left(\sum_{y\in\Ch(x)}\Mat(r_{x,y})\right)\otimes M^{\otimes(N_x-n_x)}\\
&=M^{\otimes n_x}\otimes M^{\otimes(N_x-n_x)}=M^{\otimes N_x}.
\end{aligned}
\]
This completes the induction, and the statement for the root follows by taking \(x=\rho\).
\end{proof}

Finally, we connect computation trees with the time-homogeneous finite-state construction of \cref{sec:2.1}.

Fix a parameter set \(\mathrm{X}=(k,H,I,Z,\alpha,\beta)\) with \(I=(\mathbf{c},\boldsymbol{\iota})\), a multiple \(n\) of \(k\), and an initial state \(h\in[H]\). Define a computation tree for \(R_n\) with the following data:
\begin{enumerate}
\item each node carries a state in \([H]\), and the root has state \(h\);
\item each node of weighted depth less than \(n\) is internal; if such a node has state \(u\), then the Sergeev decomposition \(S_k(c_{\boldsymbol{\iota}(u)})\) is attached to it;
\item the children of a node with state \(u\) are in bijection with the rank-\(1\) terms \((U_{c_{\boldsymbol{\iota}(u)},r},V_{c_{\boldsymbol{\iota}(u)},r}^\T)\) of this Sergeev decomposition; the edge corresponding to the \(r\)-th term has this term attached to it and has weight \(k\);
\item the child corresponding to the \(r\)-th term has state \(\clamp_H(u+\Delta_{c_{\boldsymbol{\iota}(u)},r})\), and the nodes of weighted depth \(n\) are leaves.
\end{enumerate}
Since every edge has weight \(k\) and \(n\) is a multiple of \(k\), every root-to-leaf path has total edge weight \(n\).

\begin{proof}[Proof of \cref{prop:tree}]
In this computation tree, every internal node has a Sergeev decomposition \(S_k(c_{\boldsymbol{\iota}(u)})\) attached to it, and by Fact~A.2.1, every such \(S_k(c_{\boldsymbol{\iota}(u)})\) computes \(R_k\). Since every edge has weight \(k\) and every root-to-leaf path has total weight \(n\), Lemma~A.3.1 shows that the root circuit computes \(R_n\).

Unfolding the definition of the root circuit yields exactly the recurrence of \cref{sec:2.1}:
\[
C_{\mathrm{X}}(0,u)=([1],[1]),
\]
and, for every multiple \(n\) of \(k\) with \(n\ge k\),
\[
C_{\mathrm{X}}(n,u)=\sum_{r=0}^{2^k-1}(U_{c_i,r},V_{c_i,r}^\T)\otimes C_{\mathrm{X}}\!\left(n-k,\,\clamp_H(u+\Delta_{c_i,r})\right),\qquad\text{where }i=\boldsymbol{\iota}(u).
\]
Therefore \(C_{\mathrm{X}}(n,h)\) is the root circuit, and hence computes \(R_n\).
\end{proof}

\section{Deferred Proofs}\label{sec:B}

This appendix contains the proofs omitted from the main text.

\subsection{Proofs of the Theorems in \texorpdfstring{\Cref{sec:1}}{Section 1}}\label{sec:B.1}

\begin{proof}[Proof of \cref{thm:ov}]
If \(d=0\), then every pair is orthogonal and the answer is simply \(|U||V|\). We assume from now on that \(d\ge1\).

Write \(\delta:=\delta(R_1)\). By the definition of \(\delta(R_1)\), choose an integer \(b\ge1\) and a depth-2 linear circuit \(C=(A,B)\) computing \(R_b\) such that, with \(\Delta:=\mathrm{D}(C)\),
\[
\Delta\le(\delta+\eps/2)^b.
\tag{B.1}\label{eq:B1}
\]
Moreover, we may choose \(C\) so that every middle gate receives at least one nonzero-weight edge from an input gate. Let \(G\) be the number of middle gates of \(C\). For such a \(C\), we must have \(\Delta\ge2\) and \(G\le 2^b\Delta\)\footnote{We may first choose \(C\) arbitrarily and then remove all middle gates of \(C\) receiving no nonzero-weight edge from any input gate. This does not change the matrix computed by the circuit and does not increase \(\mathrm{D}(C)\). Now, since for any \(b\ge1\) the matrix \(R_b\) contains a copy of \(R_1\), and \(R_1\) cannot be computed by a depth-2 linear circuit with maximum input-output degree at most \(1\), we have \(\Delta\ge2\). Since each input gate has, by definition, at most \(\Delta\) edges to middle gates, we have \(G\le2^b\Delta\).}. Therefore, once \(\eps\) is fixed, the quantities \(b\), \(\Delta\), and \(G\) are bounded by constants.

Next, pad every vector in \(U\) and \(V\) with trailing zeros until the dimension is
\[
d'=bm,\qquad m=\lceil d/b\rceil.
\]
This does not change orthogonality. The \(m\)-fold Kronecker product \(C^{\otimes m}=(A^{\otimes m},B^{\otimes m})\) computes \(R_{d'}\). Its middle gates can be indexed by tuples \(g=(g_1,\ldots,g_m)\in[G]^m\). For \(i,j\in\{0,1\}^b\), define
\[
L(i):=\{g\in[G]:A[i,g]\ne0\},
\qquad
R(j):=\{g\in[G]:B[g,j]\ne0\},
\]
so that \(|L(i)|,|R(j)|\le\Delta\).

Write a padded vector \(x\in\{0,1\}^{d'}\) as \(x=(x^{(1)},\ldots,x^{(m)})\), where each \(x^{(t)}\in\{0,1\}^b\), and do the same for \(y\). For \(g=(g_1,\ldots,g_m)\), put
\[
a_x(g):=\prod_{t=1}^m A[x^{(t)},g_t],
\qquad
b_y(g):=\prod_{t=1}^m B[g_t,y^{(t)}].
\]
Since \(C^{\otimes m}\) computes \(R_{d'}\), we have \(R_{d'}[x,y]=\sum_{g\in[G]^m}a_x(g)b_y(g)\). Hence the answer, namely the desired number of orthogonal pairs, is
\[
\mathrm{ans}=\sum_{x\in U}\sum_{y\in V}R_{d'}[x,y]=\sum_{g\in[G]^m}\left(\sum_{x\in U}a_x(g)\right)\left(\sum_{y\in V}b_y(g)\right).
\tag{B.2}\label{eq:B2}
\]
If \(a_x(g)\ne0\), then \(g\in L(x^{(1)})\times\cdots\times L(x^{(m)})\); if \(b_y(g)\ne0\), then \(g\in R(y^{(1)})\times\cdots\times R(y^{(m)})\).

The remaining issue is to evaluate \eqref{eq:B2} without storing the two values \(\sum_{x\in U}a_x(g)\) and \(\sum_{y\in V}b_y(g)\) for every \(g\in[G]^m\). Choose integers \(r,s\ge0\) with \(r+s=m\), and write
\[
g=(p,q),\qquad p\in[G]^r,\quad q\in[G]^s.
\]
We call \(p\) the \emph{prefix} of \(g\) and \(q\) the \emph{suffix} of \(g\). For all sufficiently large \(n\), choose \(s\) so that
\[
G^s\le n,\qquad\text{and either }\Delta^s\ge\log_2n\text{ or }r=0.
\tag{B.3}\label{eq:B3}
\]
One can show that there exists \(n_0\), depending only on \(\eps\), such that whenever \(n\ge n_0\) such an \(s\) exists\footnote{One concrete choice is the following. Let \(s_0=\lceil\log_2\log_2n\rceil\). For all sufficiently large \(n\), depending only on \(G\), we have \(G^{s_0}\le n\); and since \(\Delta\ge2\), we also have \(\Delta^{s_0}\ge\log_2n\). If \(m\ge s_0\), take \(s=s_0\) and \(r=m-s\); otherwise take \(s=m\) and \(r=0\). This establishes \eqref{eq:B3}.}. If \(n<n_0\), we simply run the brute-force \(O(n^2d)\) algorithm; since the threshold depends only on \(\eps\), this is still \(O(n(\delta+\eps)^d)\), because \(n\) is then \(O(1)\) and \(\delta+\eps>1\).

Assume first that \(r>0\). For each \(x\in U\), create an iterator over
\[
L(x^{(1)})\times\cdots\times L(x^{(r)})
\]
in lexicographic order; its current item stores the vector \(x\), the current prefix \(p\), and the coefficient
\[
\alpha_x(p):=\prod_{t=1}^r A[x^{(t)},p_t].
\]
Similarly, for each \(y\in V\), create an iterator over
\[
R(y^{(1)})\times\cdots\times R(y^{(r)})
\]
in lexicographic order; its current item stores the vector \(y\), the current prefix \(p\), and the coefficient
\[
\beta_y(p):=\prod_{t=1}^r B[p_t,y^{(t)}].
\]
Each iterator uses \(O(d)\) space.

Maintain two min-heaps \(H_L\) and \(H_R\), which store iterators and are keyed by the iterators' current prefixes. Initially, put all iterators associated with vectors in \(U\) into \(H_L\) and all iterators associated with vectors in \(V\) into \(H_R\). We then maintain a global counter \(\mathrm{sum}\), initially \(0\), and merge the two sorted streams of prefixes:
\begin{itemize}
\item if the two heap minima have different prefixes, pop the smaller one, advance its iterator, and reinsert it unless it has ended;
\item if the two minima have the same prefix \(p\), pop all current iterator items with prefix \(p\) from both heaps; after popping each item, advance and reinsert its iterator if possible;
\item the items popped from \(H_L\) are exactly those with \(x\in U\) and \(p\in L(x^{(1)})\times\cdots\times L(x^{(r)})\), and the items popped from \(H_R\) are exactly those with \(y\in V\) and \(p\in R(y^{(1)})\times\cdots\times R(y^{(r)})\).
\end{itemize}
For this prefix \(p\), create two temporary maps \(T_L,T_R:[G]^s\to\mathbb{Z}\), initially zero. For each item popped from \(H_L\) and corresponding to \(x\), enumerate all suffixes
\[
q\in L(x^{(r+1)})\times\cdots\times L(x^{(m)})
\]
and perform the update
\[
T_L[q]\leftarrow T_L[q]+\alpha_x(p)\prod_{t=r+1}^m A[x^{(t)},q_{t-r}].
\]
For each item popped from \(H_R\) and corresponding to \(y\), enumerate all suffixes
\[
q\in R(y^{(r+1)})\times\cdots\times R(y^{(m)})
\]
and perform the update
\[
T_R[q]\leftarrow T_R[q]+\beta_y(p)\prod_{t=r+1}^m B[q_{t-r},y^{(t)}].
\]
Finally, traverse the nonzero entries of \(T_L\), add \(T_L[q]\,T_R[q]\) to the global counter \(\mathrm{sum}\), and clear both maps.

The case \(r=0\) is handled by the same procedure with the unique empty prefix, but without the two heaps: we simply build the two maps once, over all suffixes \(g\in[G]^m\), and take their dot product.

We now prove correctness. Fix a processed prefix \(p\). After the two temporary maps have been filled, for every \(q\in[G]^s\) we have
\[
T_L[q]=\sum_{x\in U}a_x(p,q),
\qquad
T_R[q]=\sum_{y\in V}b_y(p,q),
\]
where terms with zero coefficients contribute zero. Thus the amount added for the prefix \(p\) is
\[
\sum_{q\in[G]^s}\left(\sum_{x\in U}a_x(p,q)\right)\left(\sum_{y\in V}b_y(p,q)\right),
\]
which is precisely the contribution to \eqref{eq:B2} of all \(g=(p,q)\) with this prefix \(p\). The heap merge processes exactly those prefixes that occur on both sides, and all other prefixes contribute zero to \eqref{eq:B2}. Therefore the final value of \(\mathrm{sum}\) equals the right-hand side of \eqref{eq:B2}, which is the number of orthogonal pairs in \(U\times V\).

We implement the temporary maps as tries over the alphabet \([G]\). Each suffix \(q\in[G]^s\) has length \(s\le m=O(d)\), so every lookup, insertion, traversal step, or deletion costs \(O(d)\) time. Since \(G^s\le n\), the two tries contain at most \(O(n)\) keys at any moment. The two heaps contain at most \(n\) iterator items each, and for a fixed prefix, each iterator can be popped at most once from a given heap, so the temporary popped lists also have size \(O(n)\). As each stored item has length \(O(d)\), the total space usage is \(O(nd)\).

It remains to bound the running time. The suffix enumerations and trie operations cost
\[
O(nd\Delta^r\Delta^s)=O(nd\Delta^m).
\]
The heap operations and lexicographic prefix comparisons cost
\[
O(nd\Delta^r\log n).
\]
If \(\Delta^s\ge\log_2 n\), then this is \(O(nd\Delta^m)\); otherwise \(r=0\), and the heap part is absent from our implementation. Hence the total running time is \(O(nd\Delta^m)\).

Using \eqref{eq:B1} and \(bm\le d+b\), we obtain
\[
d\Delta^m\le d(\delta+\eps/2)^{bm}\le d(\delta+\eps/2)^{d+b}.
\tag{B.4}\label{eq:B4}
\]
Let
\[
\lambda:=\frac{\delta+\eps/2}{\delta+\eps}<1.
\]
Since \(d\lambda^d\le 1/(1-\lambda)=2(\delta+\eps)/\eps\) for every \(d\ge1\), the bound \eqref{eq:B4} is at most \(K(\delta+\eps)^d\) for a constant \(K\) depending only on \(\eps\). Therefore \(O(nd\Delta^m)=O(n(\delta+\eps)^d)\). This proves \cref{thm:ov}.
\end{proof}

\begin{proof}[Proof of \cref{thm:ds}]
We first prove the Subset Query statement. Again the case \(d=0\) is trivial, so assume \(d\ge1\). Following the identification of \cref{sec:1.3}, we represent a stored binary string by the indicator vector \(x\in\{0,1\}^d\) of its \(1\)-positions, and a query pattern in \(\{0,*\}^d\) by the indicator vector \(y\in\{0,1\}^d\) of its \(0\)-positions. Then \(x\) and \(y\) match if and only if \(\langle x,y\rangle=0\). Thus, if \(X\) denotes the current multiset of inserted strings, and the rows and columns of \(R_d\) are indexed by \(\{0,1\}^d\), the quantity requested by a query is
\[
\mathrm{ans}(y):=\sum_{x\in X}R_d[x,y].
\]

Choose \(b\), \(C=(A,B)\), \(\Delta\), \(G\), and \(m\) exactly as in the proof of \cref{thm:ov}, and pad stored strings by trailing \(0\)'s and query patterns by trailing \(*\)'s. In the vector \(y\) of zero-positions, this padding again appends zeros. With the same notation \(a_x(g)\) and \(b_y(g)\) as above, the circuit \(C^{\otimes m}\) gives
\[
\mathrm{ans}(y)=\sum_{g\in[G]^m}\left(\sum_{x\in X}a_x(g)\right)b_y(g).
\tag{B.5}\label{eq:B5}
\]

The data structure stores a trie-based map \(g\mapsto z_g\) for \(g\in[G]^m\), where
\[
z_g:=\sum_{x\in X}a_x(g),
\]
and missing keys have value \(0\). To insert \(x\), enumerate
\[
g\in L(x^{(1)})\times\cdots\times L(x^{(m)})
\]
and update \(z_g\leftarrow z_g+a_x(g)\). To answer a query \(y\), enumerate
\[
g\in R(y^{(1)})\times\cdots\times R(y^{(m)})
\]
and return the sum \(\sum_g z_g\,b_y(g)\) over the enumerated elements \(g\).

The invariant defining \(z_g\) holds initially and is preserved by each insertion. Therefore the query algorithm returns the right-hand side of \eqref{eq:B5}, which is exactly the number of currently stored strings matching the query pattern. Each insertion or query enumerates at most \(\Delta^m\) elements of \([G]^m\), and each trie operation costs \(O(m)=O(d)\) time. As in the final calculation of the proof of \cref{thm:ov},
\[
O(d\Delta^m)=O((\delta(R_1)+\eps)^d).
\]
This proves the claimed online deterministic data structure for Subset Query.

We next prove the \(c\)-ary Partial Match statement. The case \(d=0\) is again trivial, so assume \(d\ge1\). Fix an integer \(c\ge2\) and write \(\delta_c:=\delta(P_1^{(c)})\). As observed in \cref{sec:1.3}, the matrix \((P_1^{(c)})^{\otimes d}\) is precisely the incidence matrix of \(d\)-dimensional \(c\)-ary Partial Match. It remains to apply the same Kronecker-product data-structure argument to this base matrix.

Choose \(b\ge1\) and a depth-2 linear circuit \(C=(A,B)\) computing \((P_1^{(c)})^{\otimes b}\) such that
\[
\Delta:=\mathrm{D}(C)\le(\delta_c+\eps/2)^b.
\]
After deleting the middle gates not reached from the input layer, let \(G\) be the number of remaining middle gates. Since \(A\) has \(c^b\) rows, we have \(G\le c^b\Delta\), and \(b\), \(\Delta\), and \(G\) are constants depending only on \(c\) and \(\eps\).

Pad database strings by trailing \(0\)'s and query patterns by trailing \(*\)'s until the dimension is \(d'=bm\); the additional coordinates always contribute a factor \(1\) to the partial-match incidence value. Define \(L(i)\), \(R(j)\), \(a_x(g)\), and \(b_y(g)\) by the same formulas as before, now with
\[
i\in\{0,\ldots,c-1\}^b\quad\text{and}\quad j\in\{0,\ldots,c-1,*\}^b.
\]
The same trie-based data structure maintains
\[
z_g=\sum_{x\in X}a_x(g),
\]
updates these values on insertion, and answers a query \(y\) by returning the sum \(\sum_g z_g\,b_y(g)\) over the enumerated elements \(g\in R(y^{(1)})\times\cdots\times R(y^{(m)})\). Correctness follows from the same substitution into the expansion of \((P_1^{(c)})^{\otimes d'}[x,y]\) through \(C^{\otimes m}\).

Each operation costs \(O(d\Delta^m)\), which is \(O((\delta_c+\eps)^d)\) by the same absorption calculation as above. Substituting the upper bound on \(\delta(P_1^{(c)})\) from \cref{thm:pm} gives the displayed \(c\)-ary Partial Match running time.

Finally, the numerical ``in particular'' statements of \cref{thm:ds} follow by choosing \(\eps>0\) small enough. This completes the proof of \cref{thm:ds}.
\end{proof}

\subsection{Proofs of the Theorems in \texorpdfstring{\Cref{sec:2}}{Section 2}}\label{sec:B.2}

\begin{proof}[Proof of \cref{prop:transpose}]
Throughout this proof, circuit identities are understood modulo the middle-gate equivalence relation \(\cong\) of \cref{sec:A.1}.

Let \(\mathrm{X}=(k,H,I,Z,\alpha,\beta)\) and \(\mathrm{X}^\T=(k,H,I^\T,Z,\beta,\alpha)\). Fix \(h\in[H]\), and set \(\bar{h}:=H-1-h\) and \(i:=\boldsymbol{\iota}(h)\). Then \(\boldsymbol{\iota}^\T(\bar{h})=L-1-\boldsymbol{\iota}(H-1-\bar{h})=L-1-i\). Since \(\mathbf{c}^\T=(\overline{c_{L-1}},\ldots,\overline{c_0})\), it follows that \(I^\T(\bar{h})=S_k(\overline{c_i})\). By Fact~A.2.2, \(S_k(c_i)^\T\cong S_k(\overline{c_i})\). Relabel the middle gates of \(S_k(\overline{c_i})\) according to this equivalence, so that \(S_k(c_i)=\sum_r(U_{c_i,r},V_{c_i,r}^\T)\) and \(S_k(\overline{c_i})=\sum_r(V_{c_i,r},U_{c_i,r}^\T)\).

If the raw tilt of \((U_{c_i,r},V_{c_i,r}^\T)\) in \(\mathrm{X}\) is \(\theta_{c_i,r}\), then the raw tilt of the corresponding term \((V_{c_i,r},U_{c_i,r}^\T)\) in \(\mathrm{X}^\T\) is
\[
\log_Z\left(\frac{\wnnz_{1-\alpha,\alpha}(U_{c_i,r})}{\wnnz_{1-\beta,\beta}(V_{c_i,r})}\right)=-\theta_{c_i,r}.
\]
Because \(\theta_{c_i,r}\notin\mathbb{Z}+\frac12\), we have \(\left\lfloor-\theta_{c_i,r}+\frac12\right\rfloor=-\left\lfloor\theta_{c_i,r}+\frac12\right\rfloor\). Hence, if the state offset of \((U_{c_i,r},V_{c_i,r}^\T)\) is \(\Delta_{c_i,r}\), then the state offset of \((V_{c_i,r},U_{c_i,r}^\T)\) is \(-\Delta_{c_i,r}\).

Moreover, for every integer \(\Delta\) and every \(h\in[H]\), the definition of \(\clamp_H\) gives
\[
H-1-\clamp_H(h+\Delta)=\clamp_H(H-1-h-\Delta).
\]
Indeed, if \(h+\Delta<0\), then both sides equal \(H-1\); if \(0\le h+\Delta\le H-1\), then both sides equal \(H-1-h-\Delta\); and if \(h+\Delta>H-1\), then both sides equal \(0\).

We therefore obtain the following key fact: for every \(h\in[H]\) and every \(r\), if the rank-\(1\) term indexed by \(r\) in \(I(h)\) sends the state \(h\) to \(h_r:=\clamp_H(h+\Delta_{c_i,r})\), then the corresponding rank-\(1\) term indexed by \(r\) in \(I^\T(\bar{h})\) sends the state \(\bar{h}\) to \(H-1-h_r\).

Using this fact, we prove the first identity of \cref{prop:transpose} by induction on \(n/k\). For \(n=0\), both sides equal the trivial circuit \(([1],[1])\). Now suppose \(n\ge k\). By the recurrence defining \(C_{\mathrm{X}}\) and the inductive hypothesis,
\[
\begin{aligned}
C_{\mathrm{X}}(n,h)^\T
&\cong\left(\sum_{r=0}^{2^k-1}(U_{c_i,r},V_{c_i,r}^\T)\otimes C_{\mathrm{X}}(n-k,h_r)\right)^{\!\T}\\
&\cong\sum_{r=0}^{2^k-1}(V_{c_i,r},U_{c_i,r}^\T)\otimes C_{\mathrm{X}}(n-k,h_r)^\T\\
&\cong\sum_{r=0}^{2^k-1}(V_{c_i,r},U_{c_i,r}^\T)\otimes C_{\mathrm{X}^\T}(n-k,H-1-h_r).
\end{aligned}
\]
On the other hand, the recurrence defining \(C_{\mathrm{X}^\T}\), together with the key fact above, gives
\[
C_{\mathrm{X}^\T}(n,\bar{h})\cong\sum_{r=0}^{2^k-1}(V_{c_i,r},U_{c_i,r}^\T)\otimes C_{\mathrm{X}^\T}(n-k,H-1-h_r).
\]
Therefore \(C_{\mathrm{X}}(n,h)^\T\cong C_{\mathrm{X}^\T}(n,H-1-h)\) for every multiple \(n\) of \(k\) and every \(h\in[H]\).

Finally, the raw tilts appearing in \(\mathrm{X}^\T\) are the negatives of those appearing in \(\mathrm{X}\), so none of them lies in \(\mathbb{Z}+\frac12\). We may therefore apply the first identity to \(\mathrm{X}^\T\), obtaining
\[
C_{\mathrm{X}^\T}(n,h)^\T\cong C_{(\mathrm{X}^\T)^\T}(n,H-1-h)=C_{\mathrm{X}}(n,H-1-h).
\]
This proves the second identity of \cref{prop:transpose}.
\end{proof}

We next state and prove the hole-fixing lemma, which is the combinatorial engine behind \cref{thm:realization}. It is a quantitative variant of the hole-fixing lemma of \cite{AL25}.

\begin{namedlemma}[Hole-fixing]
Let \(X\subseteq[\mathfrak{c}_0]^n\) and \(Y\subseteq[\mathfrak{c}_1]^n\) be two frequency classes, and set \(N:=M^{\otimes n}[X,Y]\). For \(0<\eta\le0.1\), suppose that a depth-2 linear circuit \(C=(A,B)\) computes an \(\eta\)-\emph{broken copy} of \(N\), in the sense that there exist subsets \(X_{\mathrm{bad}}\subseteq X\) and \(Y_{\mathrm{bad}}\subseteq Y\) with \(|X_{\mathrm{bad}}|\le\eta|X|\) and \(|Y_{\mathrm{bad}}|\le\eta|Y|\) such that \(C\) is defined on \((X\setminus X_{\mathrm{bad}})\times(Y\setminus Y_{\mathrm{bad}})\) and computes \(N[X\setminus X_{\mathrm{bad}},Y\setminus Y_{\mathrm{bad}}]\). If \(C\) has input-gate degrees at most \(D_{\mathrm{L}}\) and output-gate degrees at most \(D_{\mathrm{R}}\), then there is a depth-2 linear circuit computing \(N\) with input-gate degrees at most \(KD_{\mathrm{L}}\) and output-gate degrees at most \(KD_{\mathrm{R}}\), where
\[
K:=2(|X||Y|)^{\boldsymbol{\alpha}(\eta)},\qquad\boldsymbol{\alpha}(\eta):=-1/\log_2(4\eta),\quad\text{so that }(4\eta)^{\boldsymbol{\alpha}(\eta)}=1/2.
\]
\end{namedlemma}

\begin{proof}
Let \(\mathfrak{S}_n\) denote the group of permutations of \([n]\). For \(\pi\in\mathfrak{S}_n\), define \((\pi\cdot x)_t:=x_{\pi^{-1}(t)}\) for \(x\in[\mathfrak{c}_0]^n\), and similarly \((\pi\cdot y)_t:=y_{\pi^{-1}(t)}\) for \(y\in[\mathfrak{c}_1]^n\). One checks that these are group actions, that they preserve frequency classes, and that they act transitively on each frequency class. Moreover, \(M^{\otimes n}[\pi\cdot x,\pi\cdot y]=M^{\otimes n}[x,y]\).

Let \(C=(A,B)\) be the circuit assumed in the lemma. For \(\pi\in\mathfrak{S}_n\), \(U\subseteq X\), and \(V\subseteq Y\), define another circuit \(\mathbf{C}(\pi,U,V):=(A',B')\) on \((U\setminus\pi^{-1}X_{\mathrm{bad}})\times(V\setminus\pi^{-1}Y_{\mathrm{bad}})\), with the same number of middle gates as \(C\), by setting \(A'[x,r]:=A[\pi\cdot x,r]\) for every \(x\in U\setminus\pi^{-1}X_{\mathrm{bad}}\) and every middle-gate index \(r\), and \(B'[r,y]:=B[r,\pi\cdot y]\) for every \(y\in V\setminus\pi^{-1}Y_{\mathrm{bad}}\) and every middle-gate index \(r\). Then \(\mathbf{C}(\pi,U,V)\) computes \(N[U\setminus\pi^{-1}X_{\mathrm{bad}},\,V\setminus\pi^{-1}Y_{\mathrm{bad}}]\).

We prove by induction on \(a+b\) the following claim: for all integers \(a,b\ge0\), whenever \(S\subseteq X\) and \(T\subseteq Y\) satisfy \(|S|\le a\) and \(|T|\le b\), the submatrix \(N[S,T]\) can be computed by a concatenation sum of at most \(F(a,b):=\max\{0,\,2(ab)^{\boldsymbol{\alpha}(\eta)}-1\}\) circuits of the form \(\mathbf{C}(\pi,U,V)\), after adding dummy gates with no wires so that each of these circuits is defined on \(S\times T\).

The case \(a=0\) or \(b=0\) is immediate, since an empty concatenation sum computes a zero matrix. Now assume \(a,b\geq1\), and put \(a_1:=\lfloor4\eta a\rfloor\) and \(b_1:=\lfloor4\eta b\rfloor\); since \(4\eta<1\), we have \(a_1<a\) and \(b_1<b\). Choose \(\pi\in\mathfrak{S}_n\) uniformly at random. By transitivity, \(\mathbb{E}_\pi[|S\cap\pi^{-1}X_{\mathrm{bad}}|]=|S||X_{\mathrm{bad}}|/|X|\le\eta|S|\). Hence Markov's inequality gives \(|S\cap\pi^{-1}X_{\mathrm{bad}}|\le4\eta|S|\) with probability at least \(3/4\); since this quantity is an integer and \(|S|\le a\), the bound implies \(|S\cap\pi^{-1}X_{\mathrm{bad}}|\le a_1\). Similarly, with probability at least \(3/4\), we have \(|T\cap\pi^{-1}Y_{\mathrm{bad}}|\le b_1\). By the union bound, a uniformly random \(\pi\) satisfies both conditions with probability at least \(1/2\); in particular, some \(\pi\in\mathfrak{S}_n\) satisfies both.

Fix such a \(\pi\), and put \(S_0:=S\setminus\pi^{-1}X_{\mathrm{bad}}\), \(T_0:=T\setminus\pi^{-1}Y_{\mathrm{bad}}\), \(S_1:=S\setminus S_0\), and \(T_1:=T\setminus T_0\). Then \(|S_1|\le a_1\) and \(|T_1|\le b_1\). The single circuit \(\mathbf{C}(\pi,S_0,T_0)\) computes \(N[S_0,T_0]\). Since
\[
N[S,T]=N[S_0,T]+N[S_1,T]=(N[S_0,T_0]+N[S_0,T_1])+N[S_1,T]
\]
as block decompositions, we can construct a circuit computing \(N[S,T]\) by applying the inductive hypothesis to \(N[S_1,T]\) and \(N[S_0,T_1]\), and then taking the concatenation sum of \(\mathbf{C}(\pi,S_0,T_0)\), the circuits obtained for \(N[S_1,T]\), and the circuits obtained for \(N[S_0,T_1]\). It remains to check that
\[
1+F(a_1,b)+F(a,b_1)\leq F(a,b)=2(ab)^{\boldsymbol{\alpha}(\eta)}-1.
\]
If \(a_1,b_1\ge1\), then, using \(a_1\le4\eta a\), \(b_1\le4\eta b\), and \((4\eta)^{\boldsymbol{\alpha}(\eta)}=1/2\),
\[
\begin{aligned}
1+F(a_1,b)+F(a,b_1)
&=2(a_1b)^{\boldsymbol{\alpha}(\eta)}+2(ab_1)^{\boldsymbol{\alpha}(\eta)}-1\\
&\le4(4\eta)^{\boldsymbol{\alpha}(\eta)}(ab)^{\boldsymbol{\alpha}(\eta)}-1\\
&=2(ab)^{\boldsymbol{\alpha}(\eta)}-1.
\end{aligned}
\]
If \(a_1=0\) and \(b_1\ge1\), then \(1+F(a_1,b)+F(a,b_1)=2(ab_1)^{\boldsymbol{\alpha}(\eta)}\le (ab)^{\boldsymbol{\alpha}(\eta)}\le2(ab)^{\boldsymbol{\alpha}(\eta)}-1\), where the first inequality uses \(b_1\le4\eta b\) and \((4\eta)^{\boldsymbol{\alpha}(\eta)}=1/2\), and the second uses \(ab\ge1\). The case \(a_1\ge1\) and \(b_1=0\) is symmetric. If \(a_1=b_1=0\), then \(1+F(a_1,b)+F(a,b_1)=1\le2(ab)^{\boldsymbol{\alpha}(\eta)}-1\), again since \(ab\ge1\). Thus the desired inequality holds in all cases, and the induction proves the claim.

Finally, apply the claim with \(S=X\), \(T=Y\), \(a=|X|\), and \(b=|Y|\). This yields a concatenation sum of at most \(F(|X|,|Y|)\le K\) circuits of the form \(\mathbf{C}(\pi,U,V)\) computing \(N\). Each summand has input-gate degrees at most \(D_{\mathrm{L}}\) and output-gate degrees at most \(D_{\mathrm{R}}\), and adding dummy gates with no wires does not increase these degrees. Hence the concatenation sum has input-gate degrees at most \(KD_{\mathrm{L}}\) and output-gate degrees at most \(KD_{\mathrm{R}}\). This proves the lemma.
\end{proof}

\begin{proof}[Proof of \cref{thm:realization}]
We first prove the case \(|\mathcal{C}|=2\). Let \(\mathcal{C}=\{C_{\mathrm{hi}},C_{\mathrm{lo}}\}\), where \(C_{\mathrm{hi}}\) computes \(M^{\otimes k_{\mathrm{hi}}}\) and \(C_{\mathrm{lo}}\) computes \(M^{\otimes k_{\mathrm{lo}}}\). In this case, \(P_\mathcal{C}(p,q)\) is the line segment between \(P_{C_{\mathrm{hi}}}(p,q)\) and \(P_{C_{\mathrm{lo}}}(p,q)\), so there exists \(\kappa\in[0,1]\) with \((a,b)=\kappa P_{C_{\mathrm{hi}}}(p,q)+(1-\kappa)P_{C_{\mathrm{lo}}}(p,q)\).

Recall that \(\mathcal{R}=\mathbb{Z}\) and \(\Gamma=\{0,\pm1\}\). Choose a padding circuit \(C_{\mathrm{pad}}=(A_{\mathrm{pad}},B_{\mathrm{pad}})\) with middle gates indexed by triples \((x,y,s)\), where \(1\le s\le |M[x,y]|\), such that \(A_{\mathrm{pad}}[x,(x,y,s)]=1\), \(B_{\mathrm{pad}}[(x,y,s),y]=\sgn(M[x,y])\), and all other entries are \(0\). This circuit computes \(M\), and
\[
\mathrm{D}(C_{\mathrm{pad}})\le\max\Bigl\{\max_x\sum_y |M[x,y]|,\ \max_y\sum_x |M[x,y]|\Bigr\}\le\sum_{x,y}|M[x,y]|.
\]

For an integer \(n\ge0\), set \(t_{\mathrm{hi}}:=\lfloor\kappa n/k_{\mathrm{hi}}\rfloor\), \(t_{\mathrm{lo}}:=\lfloor(1-\kappa)n/k_{\mathrm{lo}}\rfloor\), and \(r:=n-t_{\mathrm{hi}}k_{\mathrm{hi}}-t_{\mathrm{lo}}k_{\mathrm{lo}}\). Then \(t_{\mathrm{hi}}k_{\mathrm{hi}}\le\kappa n\), \(t_{\mathrm{lo}}k_{\mathrm{lo}}\le(1-\kappa)n\), and \(0\le r<k_{\mathrm{hi}}+k_{\mathrm{lo}}\). Let
\[
D_n:=C_{\mathrm{hi}}^{\otimes t_{\mathrm{hi}}}\otimes C_{\mathrm{lo}}^{\otimes t_{\mathrm{lo}}}\otimes C_{\mathrm{pad}}^{\otimes r}.
\]
Then \(D_n\) computes \((M^{\otimes k_{\mathrm{hi}}})^{\otimes t_{\mathrm{hi}}}\otimes(M^{\otimes k_{\mathrm{lo}}})^{\otimes t_{\mathrm{lo}}}\otimes M^{\otimes r}=M^{\otimes n}\).

Let \(X\sim p^n\) and \(Y\sim q^n\). Since \(M\) has no identically zero row or column, \(M^{\otimes r}\) has no identically zero row or column for any \(r\ge1\). Hence, in any circuit computing a Kronecker power of \(M\), every input gate and every output gate has degree at least \(1\), since otherwise the corresponding row or column of the computed matrix would be identically zero. All logarithms of degrees appearing below are therefore nonnegative.

Decompose \(X\) into the blocks used by the Kronecker product defining \(D_n\). Then \(\log_2 L_{D_n}(X)\) is a sum of independent, uniformly bounded random variables:
\[
\log_2 L_{D_n}(X)=\sum_{i=1}^{t_{\mathrm{hi}}}\log_2 L_{C_{\mathrm{hi}}}(X_i^{\mathrm{hi}})+\sum_{j=1}^{t_{\mathrm{lo}}}\log_2 L_{C_{\mathrm{lo}}}(X_j^{\mathrm{lo}})+\sum_{\ell=1}^{r}\log_2 L_{C_{\mathrm{pad}}}(X_\ell^{\mathrm{pad}}),
\]
where \(X_i^{\mathrm{hi}}\sim p^{k_{\mathrm{hi}}}\), \(X_j^{\mathrm{lo}}\sim p^{k_{\mathrm{lo}}}\), and \(X_\ell^{\mathrm{pad}}\sim p\), all independently. An analogous formula holds for \(\log_2 R_{D_n}(Y)\). By the definitions of \(f_C\) and \(g_C\),
\[
\mathbb{E}[\log_2 L_{D_n}(X)]\le t_{\mathrm{hi}}k_{\mathrm{hi}}f_{C_{\mathrm{hi}}}(p)+t_{\mathrm{lo}}k_{\mathrm{lo}}f_{C_{\mathrm{lo}}}(p)+O(1)\le an+O(1),
\]
and similarly \(\mathbb{E}[\log_2 R_{D_n}(Y)]\le bn+O(1)\), where the \(O(1)\) constants depend only on \(M\), \(C_{\mathrm{hi}}\), and \(C_{\mathrm{lo}}\).

By Hoeffding's inequality, there exist constants \(\rho>0\) and \(n_0\), depending only on \(M\), \(C_{\mathrm{hi}}\), and \(C_{\mathrm{lo}}\), such that, for every \(n\ge n_0\),
\[
\Pr_{X\sim p^n}\left[L_{D_n}(X)>2^{an+n^{2/3}}\right]\le2^{-\rho n^{1/3}},
\qquad
\Pr_{Y\sim q^n}\left[R_{D_n}(Y)>2^{bn+n^{2/3}}\right]\le2^{-\rho n^{1/3}}.
\]

Recall that, for every \(p\in\freq([\mathfrak{c}_0]^n)\), we have \(\Pr_{X\sim p^n}[X\in\FC_n(p)]\ge(n+1)^{-\mathfrak{c}_0}\)\footnote{Let \(Z=(Z_0,\ldots,Z_{\mathfrak{c}_0-1})\) be the count vector of \(X\sim p^n\). For any possible count vector \(m\),
\[
\Pr[Z=m]=\binom{n}{m_0,\ldots,m_{\mathfrak{c}_0-1}}\prod_{i=0}^{\mathfrak{c}_0-1}p_i^{m_i}.
\]
We claim that \(\Pr[Z=m]\leq\Pr[Z=np]\) for every \(m\neq np\). If \(\Pr[Z=m]=0\), this is trivial. Otherwise, \(m_i=0\) whenever \(p_i=0\). Since \(m\neq np\), choose indices \(i,j\) with \(m_i>np_i\) and \(m_j<np_j\); then \(p_i,p_j>0\). Moving one unit from coordinate \(i\) to coordinate \(j\) multiplies the probability by
\[
\frac{m_i}{m_j+1}\cdot\frac{p_j}{p_i}\ge1,
\]
because \(m_i>np_i\) and \(m_j+1\le np_j\). At each such step the probability does not decrease, and the process reaches \(np\) after finitely many steps, because each step decreases \(\sum_i |m_i-np_i|\) by \(2\). Hence \(\Pr[Z=m]\leq\Pr[Z=np]\). Since there are at most \((n+1)^{\mathfrak{c}_0}\) possible count vectors, we conclude \(\Pr[X\in\FC_n(p)]\ge(n+1)^{-\mathfrak{c}_0}\).}.

Thus, the fraction of input indices \(x\in\FC_n(p)\) with \(L_{D_n}(x)>2^{an+n^{2/3}}\) is at most \((n+1)^{\mathfrak{c}_0}2^{-\rho n^{1/3}}\), and similarly the fraction of output indices \(y\in\FC_n(q)\) with \(R_{D_n}(y)>2^{bn+n^{2/3}}\) is at most \((n+1)^{\mathfrak{c}_1}2^{-\rho n^{1/3}}\). Let \(X_{\mathrm{bad}}\) and \(Y_{\mathrm{bad}}\) denote these two exceptional sets, and put \(\eta_n:=(n+1)^{\max\{\mathfrak{c}_0,\mathfrak{c}_1\}}2^{-\rho n^{1/3}}\). Choose \(n_1\ge n_0\) such that \(\eta_n\le0.1\) for all \(n\ge n_1\).

For \(n\ge n_1\), restricting \(D_n\) to \((\FC_n(p)\setminus X_{\mathrm{bad}})\times(\FC_n(q)\setminus Y_{\mathrm{bad}})\) gives an \(\eta_n\)-broken copy of \(M^{\otimes n}[\FC_n(p),\FC_n(q)]\) whose input-gate degrees are at most \(2^{an+n^{2/3}}\) and whose output-gate degrees are at most \(2^{bn+n^{2/3}}\). Applying the hole-fixing lemma to this broken copy with \(\eta=\eta_n\), we construct a circuit for \(M^{\otimes n}[\FC_n(p),\FC_n(q)]\). Since \(|\FC_n(p)||\FC_n(q)|\le(\mathfrak{c}_0\mathfrak{c}_1)^n\), choose \(n_2\ge n_1\) such that, for every \(n\ge n_2\),
\[
\boldsymbol{\alpha}(\eta_n)=\frac{1}{\rho n^{1/3}-\max\{\mathfrak{c}_0,\mathfrak{c}_1\}\log_2(n+1)-2}\le\frac{2}{\rho n^{1/3}}.
\]
Then, for every \(n\ge n_2\),
\[
K=2(|\FC_n(p)||\FC_n(q)|)^{\boldsymbol{\alpha}(\eta_n)}\le2^{1+\frac{2\log_2(\mathfrak{c}_0\mathfrak{c}_1)}{\rho}n^{2/3}}\le2^{C n^{2/3}},
\]
where \(C\) depends only on \(M\), \(C_{\mathrm{hi}}\), and \(C_{\mathrm{lo}}\). Hence the constructed circuit has input-gate degrees at most \(2^{an+(1+C)n^{2/3}}\) and output-gate degrees at most \(2^{bn+(1+C)n^{2/3}}\). Given \(\eps\) from the theorem, choose \(n_3\ge n_2\) such that \((1+C)n^{2/3}\le\eps n\) for every \(n\ge n_3\). This proves the required upper bounds in the case \(|\mathcal{C}|=2\).

Now consider a general finite collection \(\mathcal{C}\). Fix \((p,q)\) and \((a,b)\in P_{\mathcal{C}}(p,q)\). Move from \((a,b)\) in the direction \((-1,-1)\) until reaching the boundary of \(P_{\mathcal{C}}(p,q)\), and call the endpoint \((a',b')\); then \(a'\le a\) and \(b'\le b\). Since \(P_{\mathcal{C}}(p,q)\) is the convex hull of finitely many points, \((a',b')\) lies on a boundary line segment whose endpoints are \(P_{C_{\mathrm{I}}}(p,q)\) and \(P_{C_{\mathrm{II}}}(p,q)\) for some \(C_{\mathrm{I}},C_{\mathrm{II}}\in\mathcal{C}\), where we allow \(C_{\mathrm{I}}=C_{\mathrm{II}}\). Let \(n_4\) be the maximum of the two-circuit thresholds \(n_3\) over the finitely many choices of \(C_{\mathrm{I}},C_{\mathrm{II}}\in\mathcal{C}\). Applying the two-circuit case to \(C_{\mathrm{I}}\) and \(C_{\mathrm{II}}\) gives, for every \(n\ge n_4\), a circuit with input-gate degrees at most \(2^{(a'+\eps)n}\) and output-gate degrees at most \(2^{(b'+\eps)n}\); these are at most \(2^{(a+\eps)n}\) and \(2^{(b+\eps)n}\), respectively. This proves \cref{thm:realization}.
\end{proof}

\begin{proof}[Proof of \cref{thm:size-transfer}]
Fix \(\eps>0\), and let \(n_4\) be the threshold of \cref{thm:realization} for \(M\), \(\mathcal{C}\), and slack \(\eps/2\).

In this proof, we use the fact that \(|\FC_n(p)|\le2^{\mathrm{H}_{\mathfrak{c}_0}(p)n}\) and \(|\FC_n(q)|\le2^{\mathrm{H}_{\mathfrak{c}_1}(q)n}\)\footnote{Let \(p\in\freq([\mathfrak{c}]^n)\). Under the product distribution \(p^n\), every \(x\in\FC_n(p)\) has probability \(\prod_{i:p_i>0}p_i^{np_i}=2^{-\mathrm{H}_{\mathfrak{c}}(p)n}\). Since the total probability of \(\FC_n(p)\) is at most \(1\), we get \(|\FC_n(p)|\,2^{-\mathrm{H}_{\mathfrak{c}}(p)n}\le1\), and hence \(|\FC_n(p)|\le2^{\mathrm{H}_{\mathfrak{c}}(p)n}\). Applying this with \(\mathfrak{c}=\mathfrak{c}_0\) and \(\mathfrak{c}=\mathfrak{c}_1\) proves the two stated bounds.}.

For every \(n\ge n_4\), we enumerate all pairs \((p,q)\in\freq([\mathfrak{c}_0]^n)\times\freq([\mathfrak{c}_1]^n)\), choose for each \((p,q)\) a point \((a_{p,q},b_{p,q})\in P_{\mathcal{C}}(p,q)\) minimizing \(\max\{a+\mathrm{H}_{\mathfrak{c}_0}(p),\,b+\mathrm{H}_{\mathfrak{c}_1}(q)\}\), and apply \cref{thm:realization} with slack \(\eps/2\) to obtain a circuit \(E_{p,q}\) computing \(M^{\otimes n}[\FC_n(p),\FC_n(q)]\) such that
\[
L_{E_{p,q}}(x)\le2^{(a_{p,q}+\eps/2)n}
\qquad\text{and}\qquad
R_{E_{p,q}}(y)\le2^{(b_{p,q}+\eps/2)n}
\]
for every \(x\in\FC_n(p)\) and every \(y\in\FC_n(q)\). Then
\[
\begin{aligned}
\mathrm{S}(E_{p,q})
&\le|\FC_n(p)|\,2^{(a_{p,q}+\eps/2)n}+|\FC_n(q)|\,2^{(b_{p,q}+\eps/2)n}\\
&\le2^{(a_{p,q}+\mathrm{H}_{\mathfrak{c}_0}(p)+\eps/2)n}+2^{(b_{p,q}+\mathrm{H}_{\mathfrak{c}_1}(q)+\eps/2)n}\\
&\le2^{(\xi(\mathcal{C})+\eps/2)n+1}.
\end{aligned}
\]

Finally, take the concatenation sum of all the circuits \(E_{p,q}\), after adding dummy gates with no wires so that each of them is defined on \([\mathfrak{c}_0]^n\times[\mathfrak{c}_1]^n\). This gives a circuit computing \(M^{\otimes n}\). Since the number of possible pairs \((p,q)\) is at most \((n+1)^{\mathfrak{c}_0+\mathfrak{c}_1}\), this circuit has size at most \((n+1)^{\mathfrak{c}_0+\mathfrak{c}_1}2^{(\xi(\mathcal{C})+\eps/2)n+1}\).

Choose \(n_5\ge n_4\) such that this quantity is at most \(2^{(\xi(\mathcal{C})+\eps)n}\) for all \(n\ge n_5\). Thus \(\mathrm{S}(M^{\otimes n})\le2^{(\xi(\mathcal{C})+\eps)n}\) for all \(n\ge n_5\), whence \(\sigma(M)\le2^{\xi(\mathcal{C})+\eps}\). Letting \(\eps\downarrow0\) proves \(\sigma(M)\le2^{\xi(\mathcal{C})}\).
\end{proof}

\begin{proof}[Proof of \cref{thm:degree-transfer}]
Fix \(\eps>0\), and let \(n_4\) be the threshold of \cref{thm:realization} for \(M\), \(\mathcal{C}\), and slack \(\eps/2\).

For every \(n\ge n_4\), we enumerate all pairs \((p,q)\in\freq([\mathfrak{c}_0]^n)\times\freq([\mathfrak{c}_1]^n)\), choose for each \((p,q)\) a point \((a_{p,q},b_{p,q})\in P_{\mathcal{C}}(p,q)\) minimizing \(\max\{a,b\}\), and apply \cref{thm:realization} with slack \(\eps/2\) to obtain a circuit \(E_{p,q}\) computing \(M^{\otimes n}[\FC_n(p),\FC_n(q)]\) such that
\[
L_{E_{p,q}}(x)\le2^{(a_{p,q}+\eps/2)n}
\qquad\text{and}\qquad
R_{E_{p,q}}(y)\le2^{(b_{p,q}+\eps/2)n}
\]
for every \(x\in\FC_n(p)\) and every \(y\in\FC_n(q)\). Then
\[
\mathrm{D}(E_{p,q})
\le\max\left\{2^{(a_{p,q}+\eps/2)n},\,2^{(b_{p,q}+\eps/2)n}\right\}
\le2^{(\phi(\mathcal{C})+\eps/2)n}.
\]
Finally, take the concatenation sum of all the circuits \(E_{p,q}\), after adding dummy gates with no wires so that each of them is defined on \([\mathfrak{c}_0]^n\times[\mathfrak{c}_1]^n\). This gives a circuit computing \(M^{\otimes n}\). For a fixed input gate \(x\in[\mathfrak{c}_0]^n\), only the circuits \(E_{\freq(x),q}\) can contribute wires to \(x\), so the degree of \(x\) is at most \((n+1)^{\mathfrak{c}_1}2^{(\phi(\mathcal{C})+\eps/2)n}\). Similarly, every output-gate degree is at most \((n+1)^{\mathfrak{c}_0}2^{(\phi(\mathcal{C})+\eps/2)n}\).

Choose \(n_5\ge n_4\) such that both quantities are at most \(2^{(\phi(\mathcal{C})+\eps)n}\) for all \(n\ge n_5\). Thus \(\mathrm{D}(M^{\otimes n})\le2^{(\phi(\mathcal{C})+\eps)n}\) for all \(n\ge n_5\), whence \(\delta(M)\le2^{\phi(\mathcal{C})+\eps}\). Letting \(\eps\downarrow0\) proves \(\delta(M)\le2^{\phi(\mathcal{C})}\).
\end{proof}

\begin{proof}[Proof of \cref{thm:pm}]
The circuits \(C'_{\mathsf{pull}},C'_{\mathsf{push}},C'_{(0)},\ldots,C'_{(c-1)}\) and their normalized log-degree density functions are given in \cref{sec:2.2}. Let \(x:=\psi(\log_2 c)\), so that \(x=1+\frac{1}{x^{\log_2 c}}\) and \(1<x<2\).

For the upper bound on \(\sigma(P_1^{(c)})\), apply \cref{thm:size-transfer} to \(\{C'_{\mathsf{pull}},C'_{\mathsf{push}}\}\). At \((p,q)\), the two manifestation points are \((0,q_c\log_2 c)\) and \((1,0)\). Since \(0<\log_2 x<1\), the point \((a,b):=(\log_2 x,\,(1-\log_2 x)q_c\log_2 c)\) lies on the line segment between these two manifestations.

We first prove \(a+\mathrm{H}_c(p)\le \log_2 c+\log_2 x\) directly. Let \(p^*\) be the uniform probability vector, \(p_i^*=1/c\) for \(0\le i<c\). By the nonnegativity of the KL divergence,
\[
D_{\mathrm{KL}}(p\Vert p^*)=\sum_{i=0}^{c-1}p_i\log_2\frac{p_i}{p_i^*}=-\mathrm{H}_c(p)+\log_2 c\ge0.
\]
Hence \(\mathrm{H}_c(p)\le\log_2 c\), and since \(a=\log_2 x\), we get \(a+\mathrm{H}_c(p)\le\log_2 c+\log_2 x\).

We next prove \(b+\mathrm{H}_{c+1}(q)\le \log_2 c+\log_2 x\) directly. Let \(q^*\) be the probability vector proportional to \((1,\ldots,1,c(x-1))\), so that \(q_i^*=1/(cx)\) for \(0\le i<c\) and \(q_c^*=(x-1)/x\). By the nonnegativity of the KL divergence,
\[
D_{\mathrm{KL}}(q\Vert q^*)=\sum_{i=0}^c q_i\log_2\frac{q_i}{q_i^*}=-\mathrm{H}_{c+1}(q)-q_c\log_2(c(x-1))+\log_2(cx)\ge0.
\]
Since \(x-1=x^{-\log_2 c}=c^{-\log_2 x}\), we have \(c(x-1)=c^{1-\log_2 x}\). Hence
\[
\begin{aligned}
b+\mathrm{H}_{c+1}(q)
&=\mathrm{H}_{c+1}(q)+(1-\log_2 x)q_c\log_2 c\\
&=\mathrm{H}_{c+1}(q)+q_c\log_2(c^{1-\log_2 x})\\
&=\mathrm{H}_{c+1}(q)+q_c\log_2(c(x-1))\\
&\le\log_2(cx)\\
&=\log_2 c+\log_2 x.
\end{aligned}
\]
Hence \(\xi(\{C'_{\mathsf{pull}},C'_{\mathsf{push}}\})\le\log_2 c+\log_2 x\), and \cref{thm:size-transfer} yields \(\sigma(P_1^{(c)})\le2^{\log_2 c+\log_2 x}=c\cdot\psi(\log_2 c)\).

For the upper bound on \(\delta(P_1^{(c)})\), apply \cref{thm:degree-transfer} to \(\{C'_{\mathsf{pull}},C'_{(0)},\ldots,C'_{(c-1)}\}\). At \((p,q)\), the \(c+1\) manifestation points are \((0,q_c\log_2 c)\) and \((1-p_i,q_i\log_2 c)\) for \(i=0,\ldots,c-1\). Consider the convex combination with weight \(\frac{c-1}{(c-1)+c\log_2 c}\) on \(C'_{\mathsf{pull}}\) and weight \(\frac{\log_2 c}{(c-1)+c\log_2 c}\) on each \(C'_{(i)}\). Its left coordinate is
\[
\frac{\log_2 c}{(c-1)+c\log_2 c}\sum_{i=0}^{c-1}(1-p_i)=\frac{(c-1)\log_2 c}{(c-1)+c\log_2 c},
\]
using \(\sum_i(1-p_i)=c-1\). Its right coordinate is
\[
\frac{(c-1)q_c\log_2 c}{(c-1)+c\log_2 c}+\frac{(\log_2 c)^2}{(c-1)+c\log_2 c}\sum_{i=0}^{c-1}q_i=\frac{\log_2 c}{(c-1)+c\log_2 c}\bigl((c-1)q_c+(\log_2 c)(1-q_c)\bigr),
\]
using \(\sum_i q_i=1-q_c\). Since \(\log_2 c\le c-1\) for every integer \(c\ge2\), the right coordinate is at most \(\frac{(c-1)\log_2 c}{(c-1)+c\log_2 c}\).

Hence \(\phi(\{C'_{\mathsf{pull}},C'_{(0)},\ldots,C'_{(c-1)}\})\le\frac{(c-1)\log_2 c}{(c-1)+c\log_2 c}\), and \cref{thm:degree-transfer} yields \(\delta(P_1^{(c)})\le2^{\frac{(c-1)\log_2 c}{(c-1)+c\log_2 c}}=c^{\frac{c-1}{(c-1)+c\log_2 c}}\).

Together with the upper bound on \(\sigma(P_1^{(c)})\), this proves the theorem.
\end{proof}

\subsection{Proofs of the Theorems in \texorpdfstring{\Cref{sec:3}}{Section 3}}\label{sec:B.3}

\begin{proof}[Proof of \cref{lem:transition}]
It suffices to show that, for fixed \(u,v\in[H]\), the quantity
\[
N_{u,v}(x)=\left|\left\{r:\left(U_{c_{\boldsymbol{\iota}(u)},r}\right)_x\neq0,\ \clamp_H\!\left(u+\Delta_{c_{\boldsymbol{\iota}(u)},r}\right)=v\right\}\right|
\]
depends only on the Hamming weight \(\|x\|_1\), and not on \(x\) itself.

Fix \(u\) and \(v\), and write \(c:=c_{\boldsymbol{\iota}(u)}\). Let \(x,y\in\{0,1\}^k\) have the same Hamming weight. Identifying \(x\) and \(y\) with subsets of \([k]\), choose a permutation \(\pi\in\mathfrak{S}_k\) with \(\pi x=y\).

Define a bijection \(\tau_\pi:\{0,\ldots,2^k-1\}\to\{0,\ldots,2^k-1\}\) on the rank-\(1\) terms of \(S_k(c)\) as follows:
\begin{enumerate}
\item if the rank-\(1\) term indexed by \(r\) is appended in an \((\texttt{R},a,b)\) step and is indexed by a subset \(S\subseteq[k]\) with \(|S|=a\), then \(\tau_\pi(r)\) is the term appended in the same \((\texttt{R},a,b)\) step and indexed by \(\pi S\);
\item if the rank-\(1\) term indexed by \(r\) is appended in a \((\texttt{C},a,b)\) step and is indexed by a subset \(T\subseteq[k]\) with \(|T|=b\), then \(\tau_\pi(r)\) is the term appended in the same \((\texttt{C},a,b)\) step and indexed by \(\pi T\).
\end{enumerate}
Since \(|\pi S|=|S|\) and \(|\pi T|=|T|\), the map \(\tau_\pi\) is well defined, and its inverse is \(\tau_{\pi^{-1}}\).

By the definition of the appended rank-\(1\) terms in the Sergeev decomposition, this bijection satisfies
\[
\left(U_{c,\tau_\pi(r)}\right)_{\pi z}=\left(U_{c,r}\right)_z\qquad\text{and}\qquad\left(V_{c,\tau_\pi(r)}\right)_{\pi z}=\left(V_{c,r}\right)_z\qquad\text{for every }z\in\{0,1\}^k.
\]
Indeed, in an \((\texttt{R},a,b)\) step, the term indexed by \(S\) has \(U\) equal to the indicator vector of the row \(S\), and \(V\) equal to the indicator vector of the columns \(T\subseteq[k]\setminus S\) with \(|T|\ge b\). Replacing \(S\) by \(\pi S\) gives the term indexed by \(\pi S\); moreover, \(\pi z=\pi S\) if and only if \(z=S\), while \(\pi z\subseteq[k]\setminus\pi S\) and \(|\pi z|\ge b\) if and only if \(z\subseteq[k]\setminus S\) and \(|z|\ge b\). The \((\texttt{C},a,b)\) case is the same argument with the roles of rows and columns interchanged.

In particular, both vectors of \(\tau_\pi(r)\) have the same numbers of nonzero entries at each Hamming weight as those of \(r\), so the raw tilts coincide: \(\theta_{c,\tau_\pi(r)}=\theta_{c,r}\). Using \(y=\pi x\) and the definition of the state offset, we conclude that, for every \(r=0,\ldots,2^k-1\),
\[
\left(U_{c,\tau_\pi(r)}\right)_y\neq0\iff\left(U_{c,r}\right)_x\neq0,\qquad\Delta_{c,\tau_\pi(r)}=\Delta_{c,r}.
\]
It follows that \(\clamp_H(u+\Delta_{c,\tau_\pi(r)})=v\) if and only if \(\clamp_H(u+\Delta_{c,r})=v\).

Thus \(r\mapsto\tau_\pi(r)\) is a bijection between the set counted by \(N_{u,v}(x)\) and the set counted by \(N_{u,v}(y)\), so \(N_{u,v}(x)=N_{u,v}(y)\) whenever \(\|x\|_1=\|y\|_1\).

Hence \(N_{u,v}(x)\) depends only on \(\|x\|_1\), and consequently \(\mathsf{A}_i[u,v]\) is independent of the choice of \(x\) with \(\|x\|_1=i\). Since all entries are nonnegative integers, each \(\mathsf{A}_i\) is a well-defined matrix in \(\mathbb{Z}_{\ge0}^{H\times H}\).
\end{proof}

\begin{proof}[Proof of \cref{prop:degree-formula}]
Expand \(C_{\mathrm{X}}(km,h)\) using the computation tree associated with the recursive definition of \(C_{\mathrm{X}}\). Each root-to-leaf path has \(m\) layers, and at layer \(t\) it chooses one rank-\(1\) term \((U_{c_{\boldsymbol\iota(u_{t-1})},r_t},V_{c_{\boldsymbol\iota(u_{t-1})},r_t}^{\T})\) from the Sergeev decomposition assigned to the current state \(u_{t-1}\), where
\[
u_0=h,\qquad u_t=\clamp_H\!\left(u_{t-1}+\Delta_{c_{\boldsymbol\iota(u_{t-1})},r_t}\right).
\]
By the distributivity of \(\otimes\) over the concatenation sum, this path contributes one rank-\(1\) term to the concatenation-sum expansion of \(C_{\mathrm{X}}(km,h)\), whose left vector is
\[
U_{c_{\boldsymbol\iota(u_0)},r_1}\otimes U_{c_{\boldsymbol\iota(u_1)},r_2}\otimes\cdots\otimes U_{c_{\boldsymbol\iota(u_{m-1})},r_m}.
\]
Therefore, for \(x=(x^{(1)},\ldots,x^{(m)})\in\{0,1\}^{km}\), the middle gate corresponding to this path is counted by \(L_{C_{\mathrm{X}}(km,h)}(x)\) if and only if
\[
\left(U_{c_{\boldsymbol\iota(u_{t-1})},r_t}\right)_{x^{(t)}}\neq 0\qquad\text{for every }t=1,\ldots,m.
\]
Thus \(L_{C_{\mathrm{X}}(km,h)}(x)\) is exactly the number of root-to-leaf paths whose rank-\(1\) terms have left vectors nonzero on the respective blocks \(x^{(1)},\ldots,x^{(m)}\); we call such paths \emph{valid}.

Now fix a final state \(v\in[H]\). We claim that \(\left(e_h^{\T}\mathsf{A}_{i_1}\cdots \mathsf{A}_{i_m}\right)_v\) counts the valid paths ending in state \(v\), where \(i_t=\|x^{(t)}\|_1\). Indeed, expanding the matrix product gives
\[
\left(e_h^{\T}\mathsf{A}_{i_1}\cdots \mathsf{A}_{i_m}\right)_v=\sum_{u_1,\ldots,u_{m-1}}\mathsf{A}_{i_1}[h,u_1]\,\mathsf{A}_{i_2}[u_1,u_2]\cdots\mathsf{A}_{i_m}[u_{m-1},v].
\]
By definition, \(\mathsf{A}_i[u,w]\) is the number of rank-\(1\) choices \(r\) at state \(u\) that both send the state to \(w\) and have left vector nonzero at any fixed \(k\)-bit input of Hamming weight \(i\). Since \(i_t=\|x^{(t)}\|_1\), the entry \(\mathsf{A}_{i_t}[u_{t-1},u_t]\) counts exactly the admissible choices of \(r_t\) for the transition \(u_{t-1}\to u_t\). Thus, along a fixed state sequence \(h=u_0,u_1,\ldots,u_m=v\), the product counts the admissible choices of \(r_1,\ldots,r_m\), and summing over the intermediate states yields all valid paths ending at \(v\). This proves the claim.

Finally, the entries of \(e_h^{\T}\mathsf{A}_{i_1}\cdots \mathsf{A}_{i_m}\) are nonnegative, so its \(\ell_1\)-norm is the sum of its entries:
\[
\left\|e_h^{\T}\mathsf{A}_{i_1}\cdots \mathsf{A}_{i_m}\right\|_1=\sum_{v\in[H]}\left(e_h^{\T}\mathsf{A}_{i_1}\cdots \mathsf{A}_{i_m}\right)_v.
\]
The right-hand side counts all valid root-to-leaf paths, after summing over the final state. As shown above, this is exactly the number of middle gates adjacent to the input gate \(x\) in \(C_{\mathrm{X}}(km,h)\), namely \(L_{C_{\mathrm{X}}(km,h)}(x)\). Therefore \(L_{C_{\mathrm{X}}(km,h)}(x)=\left\|e_h^{\T}\mathsf{A}_{i_1}\cdots \mathsf{A}_{i_m}\right\|_1\).
\end{proof}

\begin{proof}[Proof of \cref{thm:certificate}]
Fix \(m\ge1\), \(h\in[H]\), and \(p\in[0,1]\). Put \(P_i:=\Bin(k,p)_i\), and let \(I_1,\ldots,I_m\) be independent random variables with \(\Pr[I_t=i]=P_i\).

For a sequence \(s=(i_1,\ldots,i_m)\), set
\[
D(s):=\left\|e_h^{\T}\mathsf{A}_{i_1}\cdots\mathsf{A}_{i_m}\right\|_1.
\]
By \cref{prop:degree-formula}, for every input-gate index whose consecutive \(k\)-bit blocks have Hamming weights \(i_1,\ldots,i_m\), the quantity \(D(s)\) is the degree of that input gate in \(C_{\mathrm{X}}(km,h)\).

Temporarily fix an auxiliary initial index \(i_0\in\{0,\ldots,k\}\), to be chosen later. Since \((\mathfrak{a},\mathfrak{b})\) is an \(\eps\)-valid toll-and-potential certificate, for every \(i\in\{0,\ldots,k\}\) and every \(u\in[H]\),
\[
\sum_{j=0}^k\sum_{v=0}^{H-1}\mathsf{A}_j[u,v]\,2^{-\mathfrak{a}_{i,j}-\mathfrak{b}_{i,u}+\mathfrak{b}_{j,v}}\le1-\eps.
\tag{B.6}\label{eq:B6}
\]
Iterating \eqref{eq:B6} for \(m\) steps, starting from \((i_0,h)\), gives
\[
\sum_{i_1,\ldots,i_m}\ \sum_{u_1,\ldots,u_m}\prod_{t=1}^m\mathsf{A}_{i_t}[u_{t-1},u_t]\,2^{-\sum_{t=1}^m\mathfrak{a}_{i_{t-1},i_t}}\,2^{-\mathfrak{b}_{i_0,h}+\mathfrak{b}_{i_m,u_m}}\le(1-\eps)^m,
\]
where \(u_0:=h\). By the definition of \(\Gamma(\mathfrak{b})\), we have \(2^{-\mathfrak{b}_{i_0,h}+\mathfrak{b}_{i_m,u_m}}\ge2^{-\Gamma(\mathfrak{b})}\). Using this lower bound, then summing over the state sequences and applying the definition of \(D(s)\), we obtain
\[
\sum_{i_1,\ldots,i_m}2^{-\sum_{t=1}^m\mathfrak{a}_{i_{t-1},i_t}}\,D(i_1,\ldots,i_m)\le2^{\Gamma(\mathfrak{b})}(1-\eps)^m.
\tag{B.7}\label{eq:B7}
\]

Let \(q(s):=\prod_{t=1}^mP_{i_t}\), and let \(\Omega_p:=\{s:q(s)>0\}\). Since all summands in \eqref{eq:B7} are nonnegative, the same inequality holds after restricting the sum to \(\Omega_p\). For \(s\in\Omega_p\), write
\[
X_s:=2^{-\sum_{t=1}^m\mathfrak{a}_{i_{t-1},i_t}}\,D(s).
\]
Let \(S=(I_1,\ldots,I_m)\), so that \(S\sim q\). Then
\[
\mathbb{E}_{S\sim q}\!\left[\frac{X_S}{q(S)}\right]=\sum_{s\in\Omega_p}q(s)\!\left(\frac{X_s}{q(s)}\right)=\sum_{s\in\Omega_p}X_s\le2^{\Gamma(\mathfrak{b})}(1-\eps)^m.
\]
Jensen's inequality for the concave function \(\log_2\) gives
\[
\mathbb{E}_{S\sim q}\!\left[\log_2\!\left(\frac{X_S}{q(S)}\right)\right]\le\log_2\!\left(\mathbb{E}_{S\sim q}\!\left[\frac{X_S}{q(S)}\right]\right)\le\Gamma(\mathfrak{b})+m\log_2(1-\eps).
\]
Expanding the logarithm, and using \(\mathbb{E}_{S\sim q}[\log_2q(S)]=-mE_k(p)\), we obtain
\[
\mathbb{E}_{S\sim q}[\log_2D(S)]\le\mathbb{E}_{S\sim q}\!\left[\sum_{t=1}^m\mathfrak{a}_{I_{t-1},I_t}\right]-mE_k(p)+\Gamma(\mathfrak{b})+m\log_2(1-\eps),
\tag{B.8}\label{eq:B8}
\]
where \(I_0:=i_0\).

We now choose the auxiliary index \(i_0\). For \(i\in\{0,\ldots,k\}\), set \(\nu_i(p):=\sum_{j=0}^kP_j\,\mathfrak{a}_{i,j}\), and choose \(i_0\) to minimize \(\nu_i(p)\). Then
\[
\nu_{i_0}(p)\le\sum_{i=0}^kP_i\,\nu_i(p)=\sum_{i=0}^k\sum_{j=0}^kP_iP_j\,\mathfrak{a}_{i,j}=\Phi_{\mathfrak{a}}(p).
\]
Since \(I_1,\ldots,I_m\) are independent and identically distributed with law \(\Bin(k,p)\),
\[
\mathbb{E}_{S\sim q}\!\left[\sum_{t=1}^m\mathfrak{a}_{I_{t-1},I_t}\right]=\nu_{i_0}(p)+(m-1)\Phi_{\mathfrak{a}}(p)\le m\Phi_{\mathfrak{a}}(p).
\]
Substituting this upper bound into \eqref{eq:B8} and dividing by \(km\) gives
\[
\frac1{km}\,\mathbb{E}_{S\sim q}[\log_2D(S)]\le\widetilde\Phi_{\mathfrak{a}}(p)+\frac{\Gamma(\mathfrak{b})+m\log_2(1-\eps)}{km}.
\]
By \cref{prop:degree-formula} and the definition of \(f_{C_{\mathrm{X}}(km,h)}\), the left-hand side equals \(f_{C_{\mathrm{X}}(km,h)}(\Ber(p))\). This proves the first claim.

If \(m\ge\lceil\Gamma(\mathfrak{b})/(-\log_2(1-\eps))\rceil\), then \(\Gamma(\mathfrak{b})+m\log_2(1-\eps)\le0\), and hence \(f_{C_{\mathrm{X}}(km,h)}(\Ber(p))\le\widetilde\Phi_{\mathfrak{a}}(p)\). Finally, \(\underline{E}_{k,\eta}(p)\le E_k(p)\) for every \(0<\eta\le0.5\), so \(\widetilde\Phi_{\mathfrak{a}}(p)\le\widetilde{\widetilde\Phi}_{\mathfrak{a},\eta}(p)\). This proves the theorem.
\end{proof}

\begin{proof}[Proof of \cref{thm:envelope}]
For \(r=1,\dots,L-1\), write \(F_r(p):=\widetilde{\widetilde\Phi}_{\mathfrak{a}^{(r)},\eta}(p)\), and set \(F_0(p)=F_L(p):=+\infty\) for all \(p\in[0,1]\).

By \cref{thm:certificate}, if \(m\ge\lceil\Gamma(\mathfrak{b}^{(r)})/(-\log_2(1-\eps))\rceil\), then, for every \(h\) and every \(p\),
\[
f_{C_{\mathrm{X}}(km,h)}(\Ber(p))\le\widetilde\Phi_{\mathfrak{a}^{(r)}}(p)\le\widetilde{\widetilde\Phi}_{\mathfrak{a}^{(r)},\eta}(p)=F_r(p).
\]
Since \(m\ge m_{\mathrm{X}}\), this inequality holds for every \(r=1,\dots,L-1\).

Now fix \(p\in[0,1]\), and choose \(r\) such that \(p\in[p_r,p_{r+1})\), with the last interval understood to include \(p_L=1\). By definition, \(\ell_{\mathrm{X}}(p)=\min\{F_r(p),F_{r+1}(p)\}\).

If \(r=0\), then \(F_0(p)=+\infty\), so \(\ell_{\mathrm{X}}(p)=F_1(p)\). If \(r=L-1\), then \(F_L(p)=+\infty\), so \(\ell_{\mathrm{X}}(p)=F_{L-1}(p)\). In all remaining cases, both \(F_r\) and \(F_{r+1}\) correspond to interior certificates. In every case, therefore, the upper bound above gives
\[
f_{C_{\mathrm{X}}(km,h)}(\Ber(p))\le\min\{F_r(p),F_{r+1}(p)\}=\ell_{\mathrm{X}}(p).
\]
It follows that, for every \(m\ge m_{\mathrm{X}}\), the inequality \(f_{C_{\mathrm{X}}(km,h)}(\Ber(p))\le\ell_{\mathrm{X}}(p)\) holds uniformly over \(h\) and \(p\).

It remains to prove the continuity claim. For each \(r=1,\dots,L-1\), the function \(F_r\) is continuous on \([0,1]\), since \(\Phi_{\mathfrak{a}^{(r)}}\) is a polynomial in \(p\) and \(\underline{E}_{k,\eta}\) is continuous. Hence \(\ell_{\mathrm{X}}\) is continuous on each open interval \((p_r,p_{r+1})\), being the minimum of two continuous functions there.

It remains to check the lattice points. Let \(p_r\) be an interior lattice point. The left limit of \(\ell_{\mathrm{X}}\) at \(p_r\) is \(\min\{F_{r-1}(p_r),F_r(p_r)\}\), while the value of \(\ell_{\mathrm{X}}\) at \(p_r\) and its right limit are both \(\min\{F_r(p_r),F_{r+1}(p_r)\}\). By the adjacent-certificate condition,
\[
F_r(p_r)\le\min\{F_{r-1}(p_r),F_{r+1}(p_r)\},
\]
so both minima equal \(F_r(p_r)\). Hence the left limit, the right limit, and the value of \(\ell_{\mathrm{X}}\) agree at every interior lattice point. Continuity at \(0\) and \(1\) follows from the continuity of \(F_1\) and \(F_{L-1}\), respectively.

Consequently, \(\ell_{\mathrm{X}}\) is continuous on \([0,1]\), and the theorem follows.
\end{proof}

\begin{proof}[Proof of \cref{prop:lipschitz}]
Fix a certificate \(\mathfrak{a}^{(r)}\), and abbreviate
\[
F_r(p):=\widetilde{\widetilde\Phi}_{\mathfrak{a}^{(r)},0.01}(p)=\frac17\left(\Phi_{\mathfrak{a}^{(r)}}(p)-\underline{E}_{7,0.01}(p)\right).
\]
By the Bernstein-form expression of \cref{sec:3.3.1},
\[
\Phi_{\mathfrak{a}^{(r)}}(p)=\sum_{t=0}^{14}\overline{\mathfrak{a}}^{(r)}_t\binom{14}{t}p^t(1-p)^{14-t}.
\]
The derivative formula for Bernstein polynomials gives
\[
\frac{d}{dp}\Phi_{\mathfrak{a}^{(r)}}(p)=14\sum_{t=0}^{13}\left(\overline{\mathfrak{a}}^{(r)}_{t+1}-\overline{\mathfrak{a}}^{(r)}_t\right)\binom{13}{t}p^t(1-p)^{13-t}.
\]
Since the Bernstein basis polynomials \(\binom{13}{t}p^t(1-p)^{13-t}\), \(t=0,\ldots,13\), are nonnegative and sum to \(1\), it follows that
\[
\left|\frac{d}{dp}\Phi_{\mathfrak{a}^{(r)}}(p)\right|\le14\max_{0\le t<14}\left|\overline{\mathfrak{a}}^{(r)}_{t+1}-\overline{\mathfrak{a}}^{(r)}_t\right|\qquad\text{for all }p\in[0,1].
\]
Thus \(\Phi_{\mathfrak{a}^{(r)}}\) is Lipschitz on \([0,1]\) with this constant.

For the clipped entropy term, take \(\eta=0.01\). The function \(\underline{E}_{7,\eta}\) is linear on \([0,\eta]\) and on \([1-\eta,1]\), with slopes \(\frac{E_7(\eta)}{\eta}\) and \(-\frac{E_7(\eta)}{\eta}\), and it equals \(E_7\) on \([\eta,1-\eta]\). Since \(E_7\) is concave, symmetric about \(1/2\), and satisfies \(E_7(0)=E_7(1)=0\), we have \(0\le E_7'(p)\le \frac{E_7(\eta)}{\eta}\) for \(\eta\le p\le1/2\), while symmetry gives \(|E_7'(p)|\le \frac{E_7(\eta)}{\eta}\) for \(1/2\le p\le1-\eta\). Hence \(\underline{E}_{7,\eta}\) is \(\frac{E_7(\eta)}{\eta}\)-Lipschitz on \([0,1]\).

Therefore \(F_r=\frac17(\Phi_{\mathfrak{a}^{(r)}}-\underline{E}_{7,\eta})\) is Lipschitz on \([0,1]\) with constant
\[
\Lip^{(r)}=\frac17\left(14\max_{0\le t<14}\left|\overline{\mathfrak{a}}^{(r)}_{t+1}-\overline{\mathfrak{a}}^{(r)}_t\right|+\frac{E_7(0.01)}{0.01}\right).
\tag{B.9}\label{eq:B9}
\]
This proves the first claim.

Now assume that the adjacent-certificate condition of \cref{thm:envelope} holds at every interior lattice point. On each lattice interval \([p_r,p_{r+1}]\), the envelope \(\ell_{\mathrm{X}}\) is the minimum of two functions among the \(F_r\)'s. Since the minimum of two \(L\)-Lipschitz functions is again \(L\)-Lipschitz, setting \(\Lip_{\mathrm{X}}:=\max_r \Lip^{(r)}\) shows that \(\ell_{\mathrm{X}}\) is \(\Lip_{\mathrm{X}}\)-Lipschitz on each individual lattice interval.

By \cref{thm:envelope}, the adjacent-certificate condition implies that \(\ell_{\mathrm{X}}\) is continuous at every lattice point. Therefore, for any \(p<q\), splitting the interval \([p,q]\) at the finitely many lattice points it crosses gives
\[
|\ell_{\mathrm{X}}(q)-\ell_{\mathrm{X}}(p)|\le\left(\max_r \Lip^{(r)}\right)(q-p).
\]
Thus \(\ell_{\mathrm{X}}\) is Lipschitz continuous on \([0,1]\) with Lipschitz constant \(\Lip_{\mathrm{X}}=\max_r\Lip^{(r)}\). This proves the proposition.
\end{proof}

\begin{proof}[Proof of \cref{prop:envelope-props}]
We first prove the uniform upper bound on the mean log-degree. Since the reported values are truncated toward zero, the inspector output implies \(\max_{r,i,u}S^{(r)}_{i,u}<0.9999999724<1-10^{-8}\). Thus every stored certificate is \(10^{-8}\)-valid.

The inspector output also implies \(\Gamma(\mathfrak{b}^{(r)})<94.3837+89.7365<200\) for every stored certificate. Since \(-\log_2(1-10^{-8})>10^{-8}\), we get \(\frac{\Gamma(\mathfrak{b}^{(r)})}{-\log_2(1-10^{-8})}<2\cdot 10^{10}<10^{11}\). Hence
\[
m_{\mathrm{X}}=\max_r\left\lceil\frac{\Gamma(\mathfrak{b}^{(r)})}{-\log_2(1-10^{-8})}\right\rceil<10^{11}
\]
for each of the four circuit families \(\mathrm{X}\).

The functions \(\ell_{\mathrm{X}}\) were defined by applying \cref{thm:envelope} with \(k=7\), \(\eta=0.01\), and \(\eps=10^{-8}\). Therefore, for every \(m\ge 10^{11}\), every \(h\in[H]\), and every \(p\in[0,1]\), we have \(f_{C_{\mathrm{X}}(7m,h)}(\Ber(p))\le\ell_{\mathrm{X}}(p)\). This proves the first claim.

We next prove the \(13\)-Lipschitz continuity. Again by the truncation convention, the inspector output implies \(\max_{r,t}\left|\overline{\mathfrak{a}}^{(r)}_{t+1}-\overline{\mathfrak{a}}^{(r)}_t\right|<3.355\) and \(E_7(0.01)<0.372\). Therefore, by \eqref{eq:B9},
\[
\Lip^{(r)}<\frac17\left(14\cdot 3.355+\frac{0.372}{0.01}\right)=\frac{84.17}{7}<12.03<13.
\]
The inspector also checks all adjacent-certificate conditions and reports that all checks pass. Hence, by \cref{prop:lipschitz}, each lattice envelope \(\ell_{\mathrm{X}}\) is \(13\)-Lipschitz. This proves the proposition.
\end{proof}

\begin{proof}[Proof of \cref{thm:size-verifier}]
Assume that \(\mathcal{V}_{\mathrm{size}}\) accepts. We prove that
\[
\max_{p\in[0,1]}\{\ell_{\mathrm{Jessica}}(p)+\mathrm{h}(p)\}<1.244844.
\]

By \cref{prop:envelope-props}, the function \(\ell_{\mathrm{Jessica}}\) is \(13\)-Lipschitz. Let \([a,b]\) be one of the subintervals checked by the verifier, and let \(p\in[a,b]\). Then
\[
\ell_{\mathrm{Jessica}}(p)\le\ell_{\mathrm{Jessica}}(b)+13(b-p)\le\ell_{\mathrm{Jessica}}(b)+13(b-a).
\]

Now consider the entropy term. The binary entropy function satisfies \(\mathrm{h}'(p)=\log_2\frac{1-p}{p}\), so \(\mathrm{h}\) is increasing on \([0,1/2]\), decreasing on \([1/2,1]\), and maximized at \(1/2\). Hence the maximum of \(\mathrm{h}\) over the interval \([a,b]\) is attained at \(\min\{\max\{1/2,a\},b\}\), and therefore \(\mathrm{h}(p)\le \mathrm{h}(\min\{\max\{1/2,a\},b\})\).

Combining the two inequalities gives
\[
\ell_{\mathrm{Jessica}}(p)+\mathrm{h}(p)\le\ell_{\mathrm{Jessica}}(b)+13(b-a)+\mathrm{h}\!\left(\min\{\max\{1/2,a\},b\}\right)=M_{[a,b]}.
\]
Since \(\mathcal{V}_{\mathrm{size}}\) accepts, it has verified that \(M_{[a,b]}\le 1.2448439\) for every subinterval \([a,b]\). The subintervals cover \([0,1]\), so \(\ell_{\mathrm{Jessica}}(p)+\mathrm{h}(p)\le1.2448439\) for every \(p\in[0,1]\). Therefore
\[
\max_{p\in[0,1]}\{\ell_{\mathrm{Jessica}}(p)+\mathrm{h}(p)\}\le1.2448439<1.244844,
\]
which is exactly the size inequality \eqref{eq:star}.
\end{proof}

\begin{proof}[Proof of \cref{thm:degree-verifier}]
Assume that \(\mathcal{V}_{\mathrm{degree}}\) accepts. We prove that
\[
\max_{(p,q)\in[0,1]^2}\mathbf{z}(p,q)<0.319866.
\]
We first prove that \(\mathbf{z}\) is \(13\)-Lipschitz with respect to the \(\ell_\infty\)-metric.

By \cref{prop:envelope-props}, the functions \(\ell_{\mathrm{Sonetto}}\), \(\ell_{\mathrm{Regulus}}\), and \(\ell_{\mathrm{Regulus}^{\T}}\) are all \(13\)-Lipschitz on \([0,1]\). The remaining coordinate functions appearing in \(Q_0,\ldots,Q_4\), namely \(0\), \(1-q\), \(1-p\), \(p\), \(q\), \(p(1-p)\), \(\frac12(1-q^2)\), \(\frac12(1-p^2)\), and \(q(1-q)\), are all \(1\)-Lipschitz on \([0,1]\), and hence also \(13\)-Lipschitz.

Thus, for every \(\mathrm{id}\in\{0,\ldots,7\}\), we have \(\|Q_{\mathrm{id}}(p,q)-Q_{\mathrm{id}}(p',q')\|_\infty\le13\max\{|p-p'|,|q-q'|\}\). Choose a convex combination
\[
(a',b')=\sum_{\mathrm{id}=0}^7 w_{\mathrm{id}}\, Q_{\mathrm{id}}(p',q')
\]
attaining \(\mathbf{z}(p',q')\), where \(w_{\mathrm{id}}\ge0\) and \(\sum_{\mathrm{id}}w_{\mathrm{id}}=1\). Using the same weights at \((p,q)\), define \((a,b):=\sum_{\mathrm{id}=0}^7 w_{\mathrm{id}}\, Q_{\mathrm{id}}(p,q)\). Then \((a,b)\in\conv\{Q_0(p,q),\ldots,Q_7(p,q)\}\), and each coordinate of \((a,b)\) differs from the corresponding coordinate of \((a',b')\) by at most \(13\max\{|p-p'|,|q-q'|\}\). Hence
\[
\begin{aligned}
\mathbf{z}(p,q)
&\le\max\{a,b\}\\
&\le\max\{a',b'\}+13\max\{|p-p'|,|q-q'|\}\\
&=\mathbf{z}(p',q')+13\max\{|p-p'|,|q-q'|\}.
\end{aligned}
\]
Interchanging \((p,q)\) and \((p',q')\) gives the reverse inequality, and together they yield
\[
|\mathbf{z}(p,q)-\mathbf{z}(p',q')|\le13\max\{|p-p'|,|q-q'|\}.
\tag{B.10}\label{eq:B10}
\]

Now consider any rectangle \(R=[p_{\min},p_{\max}]\times[q_{\min},q_{\max}]\) on which the verifier returns true at Step 3. Let
\[
\Delta_R:=\max\left\{\frac{p_{\max}-p_{\min}}{200},\frac{q_{\max}-q_{\min}}{200}\right\}.
\]
The verifier samples the \(101\times101\) grid in \(R\), so for every point \((p,q)\in R\), there is a sampled grid point \((p_i,q_j)\) with \(\max\{|p-p_i|,|q-q_j|\}\le \Delta_R\). Therefore, by \eqref{eq:B10}, \(\mathbf{z}(p,q)\le\mathbf{z}(p_i,q_j)+13\Delta_R\).

At the grid point \((p_i,q_j)\), the quantity \(Z_{i,j}\) is defined as the smallest value of \(\max\{x,y\}\) among the candidates tested by the verifier, namely the eight vertices \(Q_{\mathrm{id}}(p_i,q_j)\) and the two-vertex convex combinations considered in Step 2. All of these candidates lie in \(\conv\{Q_0(p_i,q_j),\ldots,Q_7(p_i,q_j)\}\). Since \(\mathbf{z}(p_i,q_j)\) minimizes \(\max\{x,y\}\) over the whole convex hull, minimizing only over this smaller candidate list cannot yield a smaller value. Hence \(\mathbf{z}(p_i,q_j)\le Z_{i,j}\).

Since the verifier returned true at Step 3 on \(R\), it verified that \(Z_{i,j}+13\Delta_R\le0.3198659\) at every sampled grid point. Thus, for every \((p,q)\in R\),
\[
\mathbf{z}(p,q)\le\mathbf{z}(p_i,q_j)+13\Delta_R\le Z_{i,j}+13\Delta_R\le0.3198659.
\]

Finally, since \(\mathcal{V}_{\mathrm{degree}}\) accepts on \([0,1]^2\), its recursive procedure covers the whole square by rectangles on which the Step 3 test returns true. Every point of \([0,1]^2\) therefore satisfies \(\mathbf{z}(p,q)\le0.3198659\), and consequently
\[
\max_{(p,q)\in[0,1]^2}\mathbf{z}(p,q)\le0.3198659<0.319866,
\]
which is exactly the degree inequality \eqref{eq:starstar}.
\end{proof}

\end{document}